\documentclass[
a4paper
]{article}
\usepackage[top=15truemm,bottom=15truemm,left=15truemm,right=15truemm]{geometry}

\usepackage[whole,autotilde]{bxcjkjatype} %
\usepackage{hyperref}
\hypersetup{
setpagesize=false,
 bookmarksnumbered=true,
 bookmarksopen=true,
 colorlinks=true,
 linkcolor=blue,
 citecolor=red,
}

\usepackage{amsmath,amssymb}
\usepackage{url}
\usepackage{braket}
\usepackage{graphicx}
\usepackage{todonotes}
\presetkeys{todonotes}{inline}{}
\usepackage{cite}
\usepackage{bm}
\usepackage{xcolor}
\usepackage{algorithm,algorithmic}
\usepackage{mathtools}
\usepackage{authblk}
\numberwithin{equation}{section}
\usepackage{subfigure}
\usepackage{afterpage}
\usepackage{comment}

\usepackage{amsthm}
\newtheorem{theorem}{Theorem}
\newtheorem{definition}[theorem]{Definition}
\newtheorem{lemma}[theorem]{Lemma}
\newtheorem{corollary}[theorem]{Corollary}
\newtheorem{proposition}[theorem]{Proposition}

\newtheorem{conjecture}[theorem]{Conjecture}

\newcommand{\argmin}{\mathop{\rm arg~min}\limits}

\title{Generative model for learning quantum ensemble \\ via optimal transport loss}

\author[,1,2,3]{Hiroyuki~Tezuka\footnote{These authors equally contributed to this work.}}
\newcommand\CoAuthorMark{\footnotemark[\arabic{footnote}]} %
\author[,2,4]{Shumpei~Uno\protect\CoAuthorMark}
\author[,2,5]{Naoki~Yamamoto\thanks{
		e-mail address: \texttt{yamamoto@appi.keio.ac.jp}
	}}

\affil[1]{Sony Group Corporation, 1-7-1 Konan, Minato-ku, Tokyo, 108-0075, Japan}
\affil[2]{Quantum Computing Center, Keio University, Hiyoshi 3-14-1, Kohoku-ku, Yokohama 223-8522, Japan}
\affil[3]{Graduate School of Science and Technology, Keio University, 3-14-1 Hiyoshi, Kohoku-ku, Yokohama, Kanagawa, 223- 8522, Japan}
\affil[4]{Mizuho Research \& Technologies, Ltd., 2-3 Kanda-Nishikicho, Chiyoda-ku, Tokyo, 101-8443, Japan}
\affil[5]{Department of Applied Physics and Physico-Informatics, Keio University, Hiyoshi 3-14-1, Kohoku-ku, Yokohama 223-8522, Japan}
\date{}

\usepackage{qcircuit}
\begin{document}

\maketitle
\begin{abstract}
Generative modeling is an unsupervised machine learning framework, that exhibits 
strong performance in various machine learning tasks. 
Recently we find several quantum version of generative model, some of which 
are even proven to have quantum advantage. 
However, those methods are not directly applicable to construct a generative 
model for learning a set of quantum states, i.e., ensemble. 
In this paper, we propose a quantum generative model that can learn quantum 
ensemble, in an unsupervised machine learning framework. 
The key idea is to introduce a new loss function calculated based on optimal 
transport loss, which have been widely used in classical machine learning 
due to its several good properties; e.g., no need to ensure the common support 
of two ensembles. 
We then give in-depth analysis on this measure, such as the scaling property 
of the approximation error. 
We also demonstrate the generative modeling with the application to quantum 
anomaly detection problem, that cannot be handled via existing methods. 
The proposed model paves the way for a wide application such as the health 
check of quantum devices and efficient initialization of quantum computation. 
\end{abstract}

\section{Introduction}
\label{sec:introduction}

In the recent great progress of quantum algorithms for both noisy near-term and 
future fault-tolerant quantum devices, particularly the quantum machine learning 
(QML) attracts huge attention. 
QML is largely categorised into two regimes in view of the type of data, which 
can be roughly called classical data and quantum data. 
The former has a conventional meaning used in the classical case; for the 
supervised learning scenario, e.g., a quantum system is trained to give 
a prediction for a given classical data such as an image. 
As for the latter, on the other hand, the task is to predict some property 
for a given quantum state drawn from a set of states, e.g., the phase 
of a many-body quantum state, again in the supervised learning scenario. 
Thanks to the obvious difficulty to directly represent a huge quantum state 
classically, some quantum advantage have been proven in QML for quantum 
data \cite{aharonov2022quantum, wu2021provable,huang2021quantum}.

In the above paragraph we used the supervised learning setting to explain the 
difference of classical and quantum data. 
But the success of unsupervised learning in classical machine learning, 
particularly the generative modeling, is of course notable; actually a variety of 
algorithms have demonstrated strong performance in several applications, 
such as image generation~\cite{bao2017cvae, brock1809large, kulkarni2015deep}, 
molecular design~\cite{gomez2018automatic}, and anomaly detection 
\cite{zhou2017anomaly}. 
Hence, it is quite reasonable that several quantum unsupervised learning 
algorithms have been actively developed, such as quantum circuit born machine 
(QCBM) \cite{benedetti2019generative, coyle2020born}, quantum generative 
adversarial network (QGAN) \cite{lloyd2018quantum, dallaire2018quantum}, 
and quantum autoencoder (QAE) \cite{romero2017quantum, wan2017quantum}. 
Also, Ref. \cite{dallaire2018quantum} studied the generative modeling problem for quantum data; 
the task is to construct a model quantum system producing a set of quantum 
states, i.e., {\it quantum ensemble}, that approximates a given quantum ensemble. 
The model quantum system contains latent variables, the change of which 
corresponds to the change of output quantum state of the system. 
In classical case, such generative model governed by latent variables is 
called an implicit model. 
It is known that, to efficiently train an implicit model, we are often 
encouraged to take the policy to minimize a distance between the model 
dataset and training dataset, rather than minimizing e.g., the divergence 
between two probability distributions. 
The {\it optimal transport loss (OTL)}, which typically leads to the Wasserstein 
distance, is suitable for the purpose of measuring the distance of two dataset; 
actually the quantum version of Wasserstein distance was proposed in \cite{zhou2022quantum,de2021quantum} and 
was applied to construct a generative model for quantum ensemble in QGAN 
framework \cite{NEURIPS2019_f35fd567,kiani2022learning}.

Along this line of research, in this paper we also focus on the generative 
modeling problem for quantum ensemble. 
We are motivated from the fact that the above-mentioned existing works employed 
the Wasserstein distance defined for two mixed quantum states corresponding to 
the training and model quantum ensembles, where each mixed state is obtained by 
compressing all element of the quantum ensemble to a single mixed state. 
This is clearly problematic, because this compression process loses a lot of 
information of the ensemble; for instance, single qubit pure states uniformly 
distributed on the equator of the Bloch sphere may be compressed to a maximally 
mixed state, which clearly does not recover the original ensemble. 
Accordingly, it is obvious that learning a single mixed state produced from 
the training ensemble does not lead to a model system that can approximate the 
original training ensemble.

In this paper, hence, we propose a new quantum OTL, which directly measures 
the difference between two quantum ensembles. 
The generative model can then be obtained by minimizing this quantum OTL 
between a training quantum ensemble and the ensemble of pure quantum states 
produced from the model. 
As the generative model, we use a parameterized quantum circuit (PQC) that 
contains tuning parameters and latent variables, which are both served by 
the angles of single-qubit rotation gates. 
A notable feature of the proposed OTL is that this has a form of sum of local 
functions that operates on a few neighboring qubits. 
This condition (i.e., the locality of the cost) is indeed necessary to 
train the model without suffering from the so-called vanishing gradient 
issue \cite{mcclean2018barren}, meaning that the gradient vector with respect 
to the parameters decreases exponentially fast when increasing the number 
of qubits.

Using the proposed quantum OTL, which will be formally defined in Section~\ref{sec:proposed_algorithm}, 
we will show the following result. 
The first result is given in Section~\ref{sec:PerformanceAnalysis}, which 
provides performance analysis of OTL and its gradient from several aspects; 
e.g., scaling properties of OTL as a function of the number of training data 
and the number of measurement. 
We also numerically confirm that the gradient of OTL is certainly free from 
the vanishing gradient issue. 
The second result is provided in Section~\ref{sec:demonstration}, showing 
some example of constructing a generative model for quantum ensemble by 
minimizing the OTL. 
This demonstration includes the application of quantum generative model 
to an anomaly detection problem of quantum data; that is, once a generative 
model is constructed by learning a given quantum ensemble, then it can be 
used to detect an anomaly quantum state by measuring the distance of this 
state to the output ensemble of the model. 
Section~\ref{sec:conclusion} gives a concluding remark, and some supporting 
materials including the proof of theorems are given in Appendix.

\section{Preliminaries}

In this section, we first review the implicit generative model for classical 
machine learning problems in Sec.~\ref{sec:implicitGenerativeModel}. 
Next, Sec.~\ref{sec:OptimalTransportLoss} is devoted to describe the general 
OTL, which can be effectively used as a cost function to train a generative 
model.

\subsection{Implicit Generative Model}
\label{sec:implicitGenerativeModel}

The generative model is used to approximate an unknown probability 
distribution that produces a given training dataset. 
The basic strategy to construct a generative model is as follows; assuming the 
probability distribution $\alpha(\bm{x})$ behind the given training dataset 
$\{\bm{x}_i\}_{i=1}^{M}\in \mathcal{X}^{M}$, where $\mathcal{X}$ denotes 
the space of random variables, we prepare a parameterized probability 
distribution $\beta_{\bm{\theta}}(\bm{x})$ and learn the parameters 
$\bm{\theta}$ that minimize an appropriate loss function defined on the 
training dataset.

In general, generative models are categorized to two types: \textit{prescribed 
models} and \textit{implicit models}. 
The prescribed generative modeling explicitly defines 
$\beta_{\bm{\theta}}(\bm{x})$ with the parameters $\bm{\theta}$. 
Then we can calculate the log-likelihood function of 
$\beta_{\bm{\theta}}(\bm{x}_i)$, and the parameters are determined by the 
maximum likelihood method, which corresponds to minimizing the Kullback-Leibler 
divergence $KL(\alpha \| \beta_{\bm{\theta}})$. 
On the other hand, the implicit generative modeling does not give us an explicit 
form of $\beta_{\bm{\theta}}(\bm{x})$. 
An important feature of the implicit generative model is that it can easily 
describe a probability distribution whose random variables are confined on 
a hidden low-dimensional manifold; also the data-generation process can be 
interpreted as a physical process from a latent variable to the data  
\cite{bottou2018geometrical}. 
Examples of the implicit generative model includes Variational Auto-Encoders 
\cite{kingma2013auto}, Generative Adversarial Networks 
\cite{goodfellow2014generative}, and the flow-based method 
\cite{rezende2015variational}. 
This paper focuses on the implicit generative model.

In the implicit generative model we usually assume that the distribution behind 
the training data resides on a relatively low-dimensional manifold. 
That is, an implicit generative model is expressed as a map of a random latent 
variable $\bm{z}$ onto $\mathcal{X}$; $\bm{z}$ resides in a latent space 
$\mathcal{Z}$ whose dimension $N_z$ is significantly smaller than that of the 
sample space, $N_x$. 
The latent random variable $\bm{z}$ follows a known distribution $\gamma(\bm{z})$ 
such as a uniform distribution or a Gaussian distribution. 
That is, the implicit model distribution is given by 
$\beta_{\bm{\theta}}=G_{\bm{\theta}}{\#}\gamma$, where $\#$ is called the 
push-forward operator \cite{peyre2019computational} which moves the 
distribution $\gamma$ on $\mathcal{Z}$ to a probability distribution 
on $\mathcal{X}$ through the map $G_{\bm{\theta}}$. 
This implicit generative model is trained so that the set of samples generated 
from the model distribution are close to the set of training data, by adjusting 
the parameters $\bm{\theta}$ to minimize some appropriate cost function 
$\mathcal{L}$ as follows: 
\begin{equation}
\label{eq:def of optimal parameter}
	\theta^\star = \argmin_{\bm{\theta}} \mathcal{L}(\hat{\alpha}_{M},\hat{\beta}_{\bm{\theta},{M_g}}). 
\end{equation}
$\hat{\alpha}_{M}(\bm{x})$ and $\hat{\beta}_{\bm{\theta},{M_g}}=G_{\bm{\theta}}\#\hat{\gamma}_{M_g}(\bm{z})$ denote 
empirical distributions defined with the sampled data $\{\bm{x}_i\}_{i=1}^{M}$ 
and $\{\bm{z}_i\}_{i=1}^{M_g}$, which follow the probability distributions 
$\alpha(\bm{x})$ and $\gamma(\bm{z})$, respectively: 
\begin{equation}
\label{eq:empiricalDistribution}
      \hat{\alpha}_{M}(\bm{x}) 
        = \frac{1}{{M}}\sum_{i=1}^{M} \delta(\bm{x}-\bm{x}_i), ~~~ 
      \hat{\gamma}_{M_g}(\bm{z}) 
        = \frac{1}{{M_g}}\sum_{i=1}^{M_g} \delta(\bm{z}-\bm{z}_i).
\end{equation}

\subsection{Optimal Transport Loss}
\label{sec:OptimalTransportLoss}

The OTL is used in various fields such as image analysis, natural language 
processing, and finance 
\cite{ollivier_pajot_villani_2014,santambrogio2015optimal,peyre2019computational,10.1093/imaiai/iaz032}. 
In particular, the OTL is widely used as a loss function in the generative 
modeling, mainly because it can be applicable even when the support of 
probability distributions do not match, and it can naturally incorporate the 
distance in the sample space $\mathcal{X}$ 
\cite{montavon2016wasserstein,bernton2017inference,arjovsky2017wasserstein,tolstikhin2017wasserstein,genevay2017gan,bousquet2017optimal}. 
The OTL is defined as the minimum cost of moving a probability distribution 
$\alpha$ to another distribution $\beta$:

\begin{definition}
[Optimal Transport Loss \cite{kantorovich1942translocation}]
\begin{equation}
		\begin{split}
			& \mathcal{L}_c(\alpha,\beta)  = \min_{\pi} \int c(\bm{x},\bm{y}) d\pi(\bm{x},\bm{y}),               \\
			& \mathrm{subject\ to} \quad  
			\int \pi(\bm{x},\bm{y}) d\bm{x}=\beta(\bm{y}), ~~ 
			\int \pi(\bm{x},\bm{y}) d\bm{y}=\alpha(\bm{x}), ~~ 
			\pi(\bm{x},\bm{y}) \ge 0,
		\end{split}
		\label{eq:optimaltransportloss}
\end{equation}
where $c(\bm{x},\bm{y}) \ge 0$ is a non-negative function on 
$\mathcal{X}\times \mathcal{X}$ that represents the transport cost from 
$\bm{x}$ to $\bm{y}$, and is called the \textit{ground cost}. 
Also, we call the set of couplings $\pi$ that minimizes 
$\mathcal{L}_c(\alpha,\beta)$ as the \textit{optimal transport plan}.
\end{definition}

In general, the OTL does not meet the axioms of metric between probability 
distributions; but it does when the ground cost is represented in terms of 
a metric function as follows:

\begin{definition}[p-Wasserstein distance \cite{villani2009optimal}]
	When the ground cost $c(\bm{x},\bm{y})$ is expressed as $c(\bm{x},\bm{y})=d(\bm{x},\bm{y})^p$ with a metric function 
	$d(\bm{x},\bm{y})$ and a real positive constant $p$, the p-Wasserstein 
	distance is defined as 
	\begin{equation}
		\mathcal{W}_p(\alpha,\beta) = \mathcal{L}_{d^p}(\alpha,\beta)^{1/p}.
	\end{equation}
\end{definition}

The p-Wasserstein distance satisfies the conditions of metric between probability 
distributions. 
That is, for arbitrary probability distributions $\alpha, \beta, \gamma$, the  
p-Wasserstein distance $\mathcal{W}_p$ satisfies 
$\mathcal{W}_p(\alpha,\beta)\ge 0$ 
and $\mathcal{W}_p(\alpha,\beta)=\mathcal{W}_p(\beta,\alpha)$; also it satisfies 
$\mathcal{W}_p(\alpha,\beta)=0\Leftrightarrow\alpha = \beta$ and the triangle 
inequality $\mathcal{W}_p(\alpha,\gamma)\le\mathcal{W}_p(\alpha,\beta)+\mathcal{W}_p(\beta,\gamma)$.

In general it is difficult to directly handle the probability distributions 
$\alpha$ and $\beta_{\bm{\theta}}$ to minimize the OTL 
$\mathcal{L}_c(\alpha,\beta_{\bm{\theta}})$. 
Instead, as mentioned in Eq.~\eqref{eq:def of optimal parameter}, we try to 
minimize the approximation of the OTL via the empirical distributions 
\eqref{eq:empiricalDistribution}:

\begin{definition}
[Empirical estimator for optimal transport loss \cite{villani2009optimal}]
	\begin{equation}
		\begin{split}
			& \mathcal{L}_c\left( \hat{\alpha}_{M},  \hat{\beta}_{\bm{\theta},{M_g}} \right) 
			= \min_{\{\pi_{i,j}\}_{i,j=1}^{{M},{M_g}} } 
			\sum_{i=1}^{{M}} \sum_{j=1}^{{M_g}}
			c(\bm{x}_i,G_{\bm{\theta}}(\bm{z}_j))\pi_{i,j}, \\
			& \mathrm{subject\ to} \quad  
			\sum_{i=1}^{M}\pi_{i,j} = \frac{1}{{M_g}}, ~~ 
			\sum_{j=1}^{M_g}\pi_{i,j} = \frac{1}{{M}}, ~~ 
			\pi_{i,j} \ge 0.
			\label{eq:EmpiricalEstimatorOT}
		\end{split}
	\end{equation}
\end{definition}

The empirical estimator converges as $\mathcal{L}_c( \hat{\alpha}_{M},\hat{\beta}_{\bm{\theta},{M}})\to	\mathcal{L}_c(\alpha,\beta_{\bm{\theta}})$ in the limit ${M=M_g}\to\infty$. 
In general, the speed of this convergence is of the order of $O(M^{-1/N_x})$ 
with $N_x$ the dimension of the sample space $\mathcal{X}$ \cite{dudley1969speed}, 
but the p-Wasserstein distance enjoys the following convergence law 
\cite{weed2019sharp}.

\begin{theorem}[Convergence rate of p-Wasserstein distance]
	\label{th:ConvergenceSpeedWasserstein}
	For the upper Wasserstein dimension $d_p^*(\alpha)$ (which is given in 
	Definition~4 of \cite{weed2019sharp}) of the probability distribution $\alpha$, the following expression holds when $s$ is larger than $d_p^*(\alpha)$:
    \begin{equation}
		\mathbb{E}\left[\mathcal{W}_p(\alpha,\hat{\alpha}_{M})\right] \lesssim O({M}^{-1/s}),
		\label{eq:dimensiondependence_self}
	\end{equation}
where the expectation $\mathbb{E}$ is taken with respect to the samples 
drawn from the empirical distribution $\hat{\alpha}_{M}$.
\end{theorem}

Intuitively, the upper Wasserstein dimension $d_p^*(\alpha)$ can be interpreted 
as the support dimension of the probability distribution $\alpha$, which 
corresponds to the dimension of the latent space, $N_z$, in the implicit 
generative model. 
Exploiting the metric properties of the p-Wasserstein distance, the following 
corollaries are immediately derived from Theorem~\ref{th:ConvergenceSpeedWasserstein}:

\begin{corollary}[Convergence rate of p-Wasserstein distance between empirical distributions sampled from a common distribution]
    \label{co:ConvergenceSpeedWassSame}
	Let $\hat{\alpha}_{1,M}$ and $\hat{\alpha}_{2,M}$ be two different 
	empirical distributions sampled from a common distribution $\alpha$. 
	The number of samples is $M$ in both empirical distributions. 
	Then the following expression holds for $s > d_p^*(\alpha)$:
	\begin{equation}
		\mathbb{E}\left[\mathcal{W}_p(\hat{\alpha}_{1,M},\hat{\alpha}_{2,{M}})\right] \lesssim O({M}^{-1/s}),
		\label{eq:dimensiondependence_identicaldistribution}
	\end{equation}
where the expectation $\mathbb{E}$ is taken with respect to the samples 
drawn from the empirical distributions $\hat{\alpha}_{1,M}$ and 
$\hat{\alpha}_{2,M}$. 
\end{corollary}

\begin{corollary}[Convergence rate of p-Wasserstein distance between different empirical distributions]
    \label{co:ConvergenceSpeedWassDiff}
	Suppose that the upper Wasserstein dimension of the probability distributions $\alpha$ and $\beta_{\bm{\theta}}$ is at most $d_p^*$, then the following expression holds for $s>d_p^*$:
	\begin{equation}
		\mathbb{E}\left[\left|\mathcal{W}_p(\alpha,\beta_{\bm{\theta}})-\mathcal{W}_p(\hat{\alpha}_{M},\hat{\beta}_{\bm{\theta},{M}})\right|\right] \lesssim O({M}^{-1/s}),
		\label{eq:dimensiondependence_differentdistribution}
	\end{equation}
where the expectation $\mathbb{E}$ is taken with respect to the samples 
drawn from the empirical distribution $\hat{\alpha}_{M}$ and 
$\hat{\beta}_{\bm{\theta},{M}}$. 
\end{corollary}

These corollaries indicate that the empirical estimator 
\eqref{eq:EmpiricalEstimatorOT} is a good estimator if the intrinsic dimension 
of the training data and the dimension of the latent space $N_z$ are sufficiently 
small, because the Wasserstein dimension $d_p^*$ is almost the same as the 
intrinsic dimension of the training data and the latent dimension. 
In Sec.~\ref{sec:dependenceDataDimension}, we numerically see that similar 
convergence laws hold even when the OTL is not the p-Wasserstein distance.

\section{Learning algorithm of generative model for quantum ensemble}
\label{sec:proposed_algorithm}

In Sec.~\ref{sec:vanishingGradientVQA}, we define the new quantum OTL that can be suitably used 
in the learning algorithm of the generative model for quantum ensemble. 
The learning algorithm is provided in Sec.~\ref{sec:LearningAlgorithm}.

\subsection{Optimal transport loss with local ground cost}
\label{sec:vanishingGradientVQA}

Our idea is to directly use Eq.~\eqref{eq:EmpiricalEstimatorOT} yet with 
the ground cost for quantum states, $c(\ket{\psi}, \ket{\phi})$, rather than 
that for classical data vectors, $c(\bm{x}, \bm{y})$. 
This actually enables us to define the OTL between quantum ensembles 
$\{\ket{\psi_i} \}$ and $\{\ket{\phi_i} \}$, as follows: 
\begin{equation}
\label{OTL_pre}
	\begin{split}
		\mathcal{L}_c\left( \{\ket{\psi_i}\}, \{ \ket{\phi_i}\} \right)
		= \min_{\{\pi_{i,j}\} } 
		     \sum_{i,j} c\left( \ket{\psi_i}, \ket{\phi_i} \right)\pi_{i,j}, ~~ 
		\mathrm{subject\ to} \quad 
		   \sum_{i}\pi_{i,j} = q_j, ~~ \sum_{j}\pi_{i,j} = p_i, ~~ \pi_{i,j} \ge 0, 
	\end{split}
\end{equation}
where $p_i$ and $q_j$ are probabilities that $\ket{\psi_i}$ and $\ket{\phi_i}$ 
appears, respectively. 
Note that we can define the transport loss between the corresponding mixed states 
$\sum_i p_i \ket{\psi_i}\bra{\psi_i}$ and $\sum_i q_i \ket{\phi_i}\bra{\phi_i}$ 
or some modification of them, as discussed in \cite{NEURIPS2019_f35fd567}; 
but as mentioned in Sec.~\ref{sec:introduction} , such mixed state loses the original configuration 
of ensemble (e.g., single qubit pure states uniformly distributed on the equator 
of the Bloch sphere) and thus are not applicable to our purpose.

Then our question is how to define the ground cost $c(\ket{\psi}, \ket{\phi})$. 
An immediate choice might be the trace distance: 
\begin{definition}[Trace distance for pure states\cite{nielsen2002quantum}]
	\begin{equation}
		\begin{split}
			c_{\mathrm{tr}}(\ket{\psi},\ket{\phi}) & = \sqrt{1-|\braket{\psi|\phi}|^2}.
			\label{eq:tracedistance}
		\end{split}
	\end{equation}
\end{definition}

Because the trace distance satisfies the axioms of metric, we can define the 
p-Wasserstein distance for quantum ensembles, 
$\mathcal{W}_p( \{\ket{\psi_i}\}, \{ \ket{\phi_i}\} ) = \mathcal{L}_{d^p}( \{\ket{\psi_i}\}, \{ \ket{\phi_i}\} )^{1/p}$, which allows us to have some 
useful properties described in Corollary~\ref{co:ConvergenceSpeedWassSame}. 
It is also notable that the trace distance is relatively easy to compute on 
a quantum computer, using e.g. the swap test\cite{buhrman2001quantum} or 
the inversion test\cite{havlivcek2019supervised}.

We now give an important remark. 
As will be formally described, our goal is to find a quantum circuit that 
produces a quantum ensemble (via changing latent variables) which best 
approximates a given quantum ensemble. 
This task can be executed by the gradient descent method for a parametrized 
quantum circuit, but a naive setting leads to the vanishing gradient issue, 
meaning that the gradient vector decays to zero exponentially fast with 
respect to the number of qubits \cite{mcclean2018barren}. 
There have been several proposals found in the literature 
\cite{nakaji2021expressibility,cerezo2021cost}, but a common prerequisite is 
that the cost should be a {\it local} one. 
To explain the meaning, let us consider the case where $\ket{\phi}$ is given 
by $\ket{\phi}=U\ket{0}^{\otimes n}$ where $U$ is a unitary matrix (which will 
be a parametrized unitary matrix $U(\bm{\theta})$ defining the generative model) 
and $n$ is the number of qubits. 
Then the trace distance is based on the fidelity 
$|\braket{\psi|\phi}|^2 = |\bra{\psi}U\ket{0}^{\otimes n}|^2$. 
This is the probability to get all zeros via the {\it global} measurement 
on the state $U^\dagger \ket{\psi}$ in the computational basis, which thus 
means that the trace distance is a global cost; accordingly, the fidelity-based 
learning method suffers from the vanishing gradient issue. 
On the other hand, we find that the following cost function is based on the 
localized fidelity measurement.

\begin{definition}
[Ground cost for quantum states only with local measurements 
\cite{khatri2019quantum,sharma2020noise}]
	\begin{equation}
		\begin{split}
			\mbox{}& 
			c_{\rm local}(\ket{\psi}, \ket{\phi}) 
			 = c_{\rm local}(\ket{\psi}, U\ket{0}^{\otimes n}) 
			 = \sqrt{\frac{1}{n}\sum_{k=1}^n(1-p^{(k)})},   \\
			\mbox{}& p^{(k)}  = \mathrm{Tr} 
			  \left[P_0^k  U^\dagger \ket{\psi}\bra{\psi}U\right], ~~~ 
			P_0^k  = \mathbb{I}_1\otimes \mathbb{I}_2\otimes\cdots \otimes\overbrace{\ket{0}\bra{0}_k}^{k\text{-}\mathrm{th\ bit}}\otimes\cdots\otimes\mathbb{I}_n,
			\label{eq:localcost}
		\end{split}
	\end{equation}
where $n$ is the number of qubits. 
Also, $\mathbb{I}_i$ and $\ket{0}\bra{0}_i$ denote the identity operator and 
the projection operator that act on the $i$-th qubit, respectively; 
thus $p^{(k)}$ represents the probability of getting $0$ when observing the 
$k$-th qubit. 
\end{definition}

Equation~\eqref{eq:localcost} is certainly a local cost, and thus it may be 
used for realizing effective learning free from the vanishing gradient issue 
provided that some additional conditions (which will be described in Section~\ref{sec:demonstration}) 
are satisfied. 
However, importantly, $c_{\rm local}(\ket{\psi},\ket{\phi})$ is not a distance 
between the two quantum states, because it is not symmetric and it does not 
satisfy the triangle inequality, while the trace distance \eqref{eq:tracedistance} 
satisfies the axiom of distance. 
Yet $c_{\rm local}(\ket{\psi},\ket{\phi})$ is always non-negative and becomes 
zero only when $\ket{\psi}=\ket{\phi}$, meaning that 
$c_{\rm local}(\ket{\psi},\ket{\phi})$ functions as a divergence. 
Then we can prove that, in general, the OTL defined with a divergence ground 
cost also functions as a divergence, as follows. 
The proof is given in Appendix~\ref{sec:proof_optimaltransport_divergence}.

\begin{proposition}
	When the ground cost $c(\bm{x},\bm{y})$ is a divergence satisfying 
	\begin{equation}
		\begin{split}
			c(\bm{x},\bm{y}) & \ge 0, \\
			c(\bm{x},\bm{y}) & = 0 ~~~ \mbox{iff} ~~~ \bm{x} = \bm{y},
		\end{split}
	\end{equation}
	then the OTL $\mathcal{L}_c(\alpha,\beta)$ with $c(\bm{x},\bm{y})$ is 
	also a divergence. 
	That is, $\mathcal{L}_c(\alpha,\beta)$ satisfies the following properties 
	for arbitrary probability distributions $\alpha$ and $\beta$:
    \begin{equation}
		\begin{split}
			\mathcal{L}_c(\alpha,\beta) & \ge 0, \\
			\mathcal{L}_c(\alpha,\beta) & = 0 ~~~ \mbox{iff} ~~~ \alpha =\beta.
		\end{split}
	\end{equation}
	\label{pro:optimaltransport_divergence}
\end{proposition}

Therefore, the OTL 
$\mathcal{L}_c\left( \{\ket{\psi_i}\}, \{ \ket{\phi_i}\} \right)$ given in 
Eq.~\eqref{OTL_pre} with the local ground cost 
$c_{\rm local}(\ket{\psi}, \ket{\phi})$ given in Eq.~\eqref{eq:localcost} 
functions as a divergence. 
This means that $\mathcal{L}_c\left( \{\ket{\psi_i}\}, \{ \ket{\phi_i}\} \right)$ 
can be suitably used for evaluating the difference between a given quantum 
ensemble and the set of output states of the generative model. 
At the same time, recall that, for the purpose of avoiding the gradient vanishing 
issue, we had to give up using the fidelity measure and accordingly the distance 
property of the OTL. 
Hence we directly cannot use the desirable properties described in Theorem~\ref{th:ConvergenceSpeedWasserstein}, 
Corollary~\ref{co:ConvergenceSpeedWassSame}, and Corollary~\ref{co:ConvergenceSpeedWassDiff}; nonetheless, in Section~\ref{sec:PerformanceAnalysis}, we will discuss if 
similar properties do hold even for the divergence measure.

\subsection{Learning Algorithm}
\label{sec:LearningAlgorithm}

The goal of our task is to train an implicit generative model so that it 
outputs a quantum ensemble approximating a given ensemble 
$\{\ket{\psi_i}\}_{i=1}^{M}$; that is, our generative model contains 
tunable parameters and latent variables, as in the classical case described 
in Section~\ref{sec:implicitGenerativeModel}. 
In this paper, we employ the following implicit generative model:

\begin{definition}[Implicit generative model on a quantum circuit]
	\label{def:quantumImplicitGenerativeModel}
	Using the initial state $\ket{0}^{\otimes n}$ and the parameterized quantum circuit $U(\bm{z},\bm{\theta})$, the implicit generative model on a quantum circuit 
	is defined as
	\begin{equation}
	\label{eq:quantumImplicitGenerativeModel}
		\ket{\phi_{\bm{\theta}}(\bm{z})}=U(\bm{z},\bm{\theta})\ket{0}^{\otimes n}.
	\end{equation}
    Here $\bm{\theta}$ is the vector of tunable parameters and $\bm{z}$ is 
	the vector of latent variables that follow a known probability distribution; 
	both $\bm{\theta}$ and $\bm{z}$ are encoded in the rotation angles of 
	rotation gates in $U(\bm{z},\bm{\theta})$. 
\end{definition}

The similar circuit model is also found in meta-VQE \cite{cervera2021meta}, 
which uses physical parameters such as the distance of atomic 
nucleus instead of random latent variables $\bm{z}$. 
Also, the model proposed in \cite{dallaire2018quantum} introduces the latent 
variables $\bm{z}$ as the computational basis of an initial state in the form 
$\ket{\phi_{\bm{\theta}}(\bm{z})}=U(\bm{\theta})\ket{\bm{z}}$; however, in this 
model, states with different latent variables are always orthogonal with each 
other, and thus the model cannot capture a small change of state in the Hilbert 
space via changing the latent variables. 
In contrast, the model \eqref{eq:quantumImplicitGenerativeModel} fulfills 
this purpose as long as the expressivity of the state with respect to $\bm{z}$ 
is enough. 
In addition, our model is advantageous in that the analytical derivative is 
available by the parameter shift rule \cite{mitarai2018quantum,schuld2019evaluating} 
not only for the tunable parameters $\bm{\theta}$ but also for the latent 
variables $\bm{z}$. 
This feature will be effectively utilized in the anomaly detection problem 
in Sec.~\ref{sec:demonstration}.

Next, as for the learning cost, we take the following empirical estimator of 
OTL, calculated from the training data $\{\ket{\psi_i}\}_{i=1}^{M}$ and the 
samples of latent variables $\{\bm{z}_j\}_{j=1}^{M_g}$: 
\begin{equation}
	\begin{split}
		\mathcal{L}_{c_{\text{local}}}\left(\{\ket{\psi_i}\}_{i=1}^{M},\{ \ket{\phi_{\bm{\theta}}(\bm{z}_j)}\}_{j=1}^{M_g}\right)
		&= \min_{\{\pi_{i,j}\}_{i,j=1}^{{M},{M_g}} } 
		\sum_{i=1}^{M}\sum_{j=1}^{M_g}
		c_{\text{local},i,j}\pi_{i,j}, \\
		\mathrm{subject\ to} \quad    & \sum_{i=1}^{M}\pi_{i,j} = \frac{1}{{M_g}}, ~~ 
		\sum_{j=1}^{M_g}\pi_{i,j} = \frac{1}{{M}}, ~~ 
		\pi_{i,j} \ge 0.
		\label{eq:quantumEmpiricalLoss}
	\end{split}
\end{equation}
where $c_{\text{local},i,j}$ is the ground cost given by 
\begin{align}
    \begin{split}
			c_{\text{local},i,j} = &\sqrt{\frac{1}{n}\sum_{k=1}^n (1-p^{(k)}_{i,j})}, \\
	     p^{(k)}_{i,j}=\mathrm{Tr} \left[P_0^k  U^\dagger(\bm{z}_j,\bm{\theta})\ket{\psi_i}\bra{\psi_i}U(\bm{z}_j,\bm{\theta})\right]&, ~~~ 
			P_0^k  = \mathbb{I}_1\otimes \mathbb{I}_2\otimes\cdots \otimes\overbrace{\ket{0}\bra{0}_k}^{k\text{-}\mathrm{th\ bit}}\otimes\cdots\otimes\mathbb{I}_n.
	    \label{eq:localcost_estimator} 
	\end{split}
\end{align}
Note that in practice $c_{\text{local},i,j}$ is estimated with the finite 
number of measurements (shots); 
we denote $\tilde{c}_{\text{local},i,j}^{({N_s})}$ to be the estimator 
with $N_s$ shots for the ideal one $c_{\text{local},i,j}$, and in this case 
the OTL is denoted as $\mathcal{L}_{\tilde{c}_{\text{local}}^{({N_s})}}$.

Based on the OTL \eqref{eq:quantumEmpiricalLoss}, the pseudo-code of proposed 
algorithm is shown in Algorithm~\ref{alg:quantumOTLossLearning}.
The total number of training quantum states $\ket{\psi_i}$ required for 
the parameter update is of the order $O({M}M_g N_s)$ in step~3, and 
$O(\max(M,M_g) N_s N_p)$ in step~5, since the parameter shift rule 
\cite{mitarai2018quantum,schuld2019evaluating} is applicable.

\begin{algorithm}
	\caption{Learning Algorithm with Quantum Optimal Transport Loss 
	\eqref{eq:quantumEmpiricalLoss}}
	\label{alg:quantumOTLossLearning}
	\begin{algorithmic}[1]
		\renewcommand{\algorithmicrequire}{\textbf{Input:}}
		\renewcommand{\algorithmicensure}{\textbf{Output:}}
		\REQUIRE Quantum circuit model $U(\bm{z},\bm{\theta})$ with initial parameters $\bm{\theta}$, learning rate $\varepsilon$, 
		ensemble $\{\ket{\psi_i}\}$
		\ENSURE  A quantum circuit that outputs an ensemble approximating 
		the input ensemble
		\REPEAT
		\STATE Generate latent variables $\{\bm{z}_j\}_{j=1}^{{M_g}}$ sampled 
		from the latent distribution.
		\STATE Estimate the ground costs $\left\{\tilde{c}^{({N_s})}_{\text{local},i,j}\right\}_{i,j=1}^{{M},{M_g}}$ from $\{\ket{\psi_i}\}_{i=1}^{M}$ and $\{U(\bm{z}_j,\bm{\theta})\}_{j=1}^{M_g}$ with ${N_s}$ shots, using  Eq.~\eqref{eq:localcost}.
		\STATE Calculate the optimal transport plan $\pi_{i,j}$ by solving the linear programming \eqref{eq:quantumEmpiricalLoss}.
		\STATE Calculate the gradients $\left\{\frac{\partial}{\partial \theta_k}\mathcal{L}_{c_{\text{local}}}\right\}_{k=1}^{N_p}$ from $\pi_{i,j}$ and $\left\{\frac{\partial}{\partial \theta_k} \tilde{c}^{({N_s})}_{\text{local},i,j}\right\}_{i,j,k=1}^{{M},{M_g},N_p}$ using the parameter shift rule.
		\STATE Update $\{\theta_k\}_{k=1}^{N_p}$ by using the gradients $\left\{\frac{\partial}{\partial \theta_k}\mathcal{L}_{c_{\text{local}}}\right\}_{k=1}^{N_p}$ with 
		learning rate $\varepsilon$. 
		\UNTIL{convergence}
	\end{algorithmic}
\end{algorithm}

\section{Performance analysis of the cost and its gradient}
\label{sec:PerformanceAnalysis}

In this section, we analyze the performance of the proposed OTL 
\eqref{eq:quantumEmpiricalLoss} and its gradient vector. 
First, in Sec.~\ref{sec:dependenceDataDimension}, we numerically study the 
approximation error of the loss with the focus on its dependence on the 
intrinsic dimension of data and the number of qubits, to see if the similar 
results to Theorem~\ref{th:ConvergenceSpeedWasserstein} and Corollaries~\ref{co:ConvergenceSpeedWassSame} and \ref{co:ConvergenceSpeedWassDiff} would hold even despite that 
the OTL is now a divergence rather than distance. 
Then, in Sec.~\ref{sec:dependenceShotNumber}, we provide numerical and 
theoretical analyses on the approximation error as a function of the number 
of measurement (shots). 
Finally, in Sec.~\ref{sec:barrenplateau}, we numerically show that the 
OTL certainly avoids the vanishing gradient issue; i.e., thanks to the 
locality of the cost, the its gradient does not decay exponentially fast. 
All the analysis in this section is focused on the property of cost at 
a certain point of learning process (say, at the initial time); the 
performance analysis on the training process will be discussed in the 
next section.

We employ the parameterized unitary matrix 
$U(\bm{z},\bm{\theta})$ shown in Fig.~\ref{fig:ansatzcircuit} to construct 
the implicit generative model \eqref{eq:quantumImplicitGenerativeModel}, 
which is similar to that given in Ref.~\cite{mcclean2018barren} except that 
our model contains the latent variables $\bm{z}$. 
That is, the model is composed of the following $N_L$ repeated unitaries 
(we call each unitary the $\ell$-th layer): 
\begin{equation}
	U_{N_L,\bm{\xi},\bm{\eta}} (\bm{z},\bm{\theta}) = \prod_{\ell=1}^{N_L} W V_{\bm{\xi}_\ell,\bm{\eta}_\ell}(\bm{z},\bm{\theta}_\ell),
\end{equation}
where $\bm{\theta}_\ell=\{\theta_{\ell,j}\}_{j=1}^n$, 
$\bm{\xi}_\ell=\{\xi_{\ell,j}\}_{j=1}^n$, and 
$\bm{\eta}_\ell=\{\eta_{\ell,j}\}_{j=1}^n$ are $n$-dimensional parameter 
vectors in the $\ell$-th layer. 
We summarize these vectors to 
$\bm{\theta}=\{\bm{\theta}_\ell\}_{\ell=1}^{N_L}$, 
$\bm{\xi}=\{\bm{\xi}_\ell\}_{\ell=1}^{N_L}$, and 
$\bm{\eta}=\{\bm{\eta}_\ell\}_{\ell=1}^{N_L}$. 
Here $\bm{\theta}$ are trainable parameters and $\bm{z}$ are latent variables. 
$W$ is a fixed entangling unitary gate composed of the ladder-structured 
controlled-$Z$ gates; that is, $W$ operates the two-qubit controlled-$Z$ 
gate on all adjacent qubits;
\begin{equation}
	W =\prod_{i=1}^{n-1} CZ_{i,i+1},
\end{equation}
where $CZ_{i,i+1}$ is the controlled-$Z$ gate acting on the $i$-th and 
$(i+1)$-th qubits. 
The operator $V_{\bm{\xi}_\ell,\bm{\eta}_\ell}(\bm{z},\bm{\theta}_\ell)$ 
consists of the single-qubit rotation operators:
\begin{equation}
	V_{\bm{\xi}_\ell,\bm{\eta}_\ell}(\bm{z},\bm{\theta}_\ell) =\prod_{i=1}^n R_{\xi_{\ell,i}}(\theta_{\ell,i}z_{\eta_{\ell,i}}),
\end{equation}
where $R_{\xi_{\ell,i}}(\theta_{\ell,i}z_{\eta_{\ell,i}})$ is the single-qubit 
Pauli rotation operator with angle $\theta_{\ell,i}z_{\eta_{\ell,i}}$ and 
direction ${\xi_{\ell,i}}\in\{X,Y,Z\}$ in the $\ell$-th layer, such as 
$R_X(\theta_{\ell,3} z_5)={\rm exp}(-i \theta_{\ell,3} z_5 \sigma_x)$. 
The component of the latent variables, $\eta_{\ell,i}\in\{0,1,2,\ldots,N_z\}$, 
and ${\xi_{\ell,i}}\in\{X,Y,Z\}$ are randomly chosen at the beginning of 
learning and never changed during learning. 
Also, we have introduced a constant bias term $z_0=1$ so that the ansatz 
$\ket{\phi_{\bm{\theta}}(\bm{z})}=U(\bm{z},\bm{\theta})\ket{0}^{\otimes n}$ can take 
a variety of states (for instance, if $z_1, \ldots, z_{N_z}\sim 0$ in the 
absence of the bias, then $\ket{\phi_{\bm{\theta}}(\bm{z})}$ is limited 
to around $\ket{0}^{\otimes n}$). 
We used Qiskit \cite{Qiskit} in all the simulation studies.

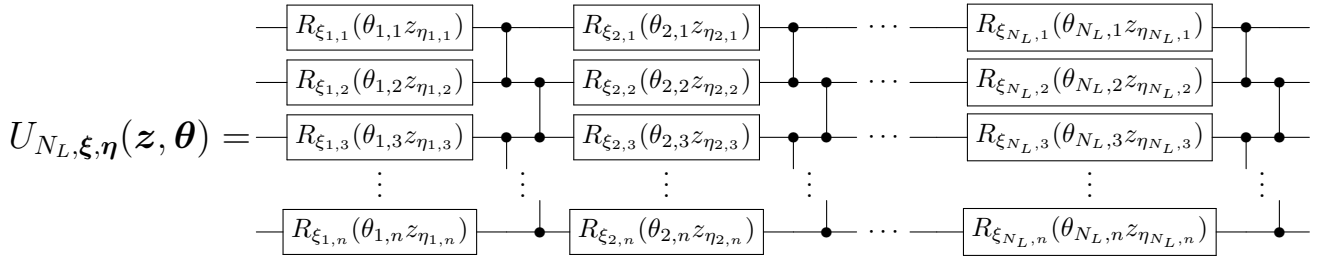
\begin{figure}[tb]
	\centering
	\leavevmode
	\Qcircuit @C=1em @R=.3em {
	&\nghost{\hspace{5em}}&\lstick{} & \gate{R_{\xi_{1,1}}(\theta_{1,1}z_{\eta_{1,1}})} & \ctrl{1}    & \qw             & \gate{R_{\xi_{2,1}}(\theta_{2,1}z_{\eta_{2,1}})} & \ctrl{1}    & \qw             & \qw & \cdots &  & \qw & \gate{R_{\xi_{{N_L},1}}(\theta_{{N_L},1}z_{\eta_{{N_L},1}})} & \ctrl{1}    & \qw             & \qw\\
	&&\lstick{} & \gate{R_{\xi_{1,2}}(\theta_{1,2}z_{\eta_{1,2}})} & \control\qw & \ctrl{1}        & \gate{R_{\xi_{2,2}}(\theta_{2,2}z_{\eta_{2,2}})} & \control\qw & \ctrl{1}        & \qw & \cdots &  & \qw & \gate{R_{\xi_{{N_L},2}}(\theta_{{N_L},2}z_{\eta_{{N_L},2}})} & \control\qw & \ctrl{1}        & \qw\\
	&&\lstick{} & \gate{R_{\xi_{1,3}}(\theta_{1,3}z_{\eta_{1,3}})} & \ctrl{1}    & \control\qw     & \gate{R_{\xi_{2,3}}(\theta_{2,3}z_{\eta_{2,3}})} & \ctrl{1}    & \control\qw     & \qw & \cdots &  & \qw & \gate{R_{\xi_{{N_L},3}}(\theta_{{N_L},3}z_{\eta_{{N_L},3}})} & \ctrl{1}    & \control\qw     & \qw\\
	&&          &                                                  &             &                 &                                                  &             &                 &     &        &  &     &                                                              &             &                 & \\
	&&          & \vdots                                           &             & \lstick{\vdots} & \vdots                                           &             & \lstick{\vdots} &     &        &  &     & \vdots                                                       &             & \lstick{\vdots} & \\\\\\
	&&          &                                                  &             &                 &                                                  &             &                 &     &        &  &     &                                                              &             &                 & \\
	&&\lstick{} & \gate{R_{\xi_{1,n}}(\theta_{1,n}z_{\eta_{1,n}})} & \qw         & \ctrl{-1}       & \gate{R_{\xi_{2,n}}(\theta_{2,n}z_{\eta_{2,n}})} & \qw         & \ctrl{-1}       & \qw & \cdots &  & \qw & \gate{R_{\xi_{{N_L},n}}(\theta_{{N_L},n}z_{\eta_{{N_L},n}})} & \qw         & \ctrl{-1}       & \qw
	\inputgrouph{1}{9}{4.2em}{\mbox{\Large $U_{N_L,\bm{\xi},\bm{\eta}} (\bm{z},\bm{\theta})=$ }}{-3.5em}
	}
	\caption{
		The structure of parameterized quantum circuit (ansatz) used in the performance analysis in Section~\ref{sec:PerformanceAnalysis} and {sec:demonstration}. 
		This ansatz consists of the repeated layers with a similar structure.
		In the $\ell$-th layer, the single qubit Pauli rotation operators with angles $\{\theta_{\ell,j}\times z_{\eta_{\ell,j}}\}_{j=1}^n$ and 
		directions $\{\xi_{\ell,j}\}_{j=1}^n$ are applied to each qubit 
		followed by a ladder-structured controlled-$Z$ gate. 
		The rotation angles $\bm{\xi}=\{\xi_{\ell,j}\}_{\ell,j=1}^{N_L,n}$ and 
		the components of the latent variables $\bm{\eta}=\{\eta_{\ell,j}\}_{\ell,j=1}^{N_L,n}$ are randomly chosen 
		at the beginning of the learning process and never changed during 
		the learning. 
	}
	\label{fig:ansatzcircuit}
\end{figure}

\subsection{Approximation error as a function of the number of training data}
\label{sec:dependenceDataDimension}

We have seen in Sec.~\ref{sec:OptimalTransportLoss} that, for the p-Wasserstein 
distance, the convergence rate of its approximation error is of the order 
$\mathcal{O}(M^{-1/s})$ rather than $\mathcal{O}(M^{-1/N_x})$, where 
$M$ is the number of training samples, $N_x$ is the dimension of the sample 
space $\mathcal{X}$, and $s$ can be interpreted as the intrinsic dimension of 
the data or the latent space dimension. 
This is desirable, because in general $s \ll N_x$. 
However, the OTL \eqref{eq:quantumEmpiricalLoss} is not a distance but 
a divergence. 
Hence it is not clear if it could theoretically enjoy a similar error 
scaling, but in this subsection we give a numerical evidence to positively 
support this conjecture.

We present the following two types of numerical simulations. 
The aim of the first one (\textbf{Experiment A}) is to see if our OTL 
would satisfy a similar property to 
Eq.~\eqref{eq:dimensiondependence_identicaldistribution}, which describes 
the difference of two empirical distributions sampled from the common 
hidden distribution. 
The second one (\textbf{Experiment B}) is studied for the case of 
Eq.~\eqref{eq:dimensiondependence_differentdistribution}, which describes 
the difference of the ideal OTL assuming the infinite samples available 
and an empirical distributions sampled from the common the ideal one. 
We used the statevector simulator \cite{Qiskit} to calculate the OTL 
assuming the infinite number of measurement; in the next subsection, we 
will study the influence of the finite number of shots on the performance. 
Throughout all the numerical experiments, we randomly chose the parameters 
$\bm{\xi},\bm{\eta},\bm{\theta},\tilde{\bm{\xi}},\tilde{\bm{\eta}},
\tilde{\bm{\theta}}$ and did not change these values.

{\bf Experiment A.} 
As an analogous quantity appearing in 
Eq.~\eqref{eq:dimensiondependence_identicaldistribution}, we here focus on 
the following expected value of the empirical OTL defined in 
Eq.~\eqref{eq:quantumEmpiricalLoss}: 
\begin{equation}
	\begin{split}
		\mathbb{E}_{\tilde{\bm{z}},\bm{z}\sim U(0,1)^{N_z}}\left[
		J^{c_{\text{local}}}_{\bm{{\xi}},{\bm{\eta}};\bm{\xi},\bm{\eta}}
		(\tilde{\bm{z}},\tilde{\bm{\theta}}:\bm{z},\bm{\theta};M)
		\right],
		\label{eq:numericalExperiment_self}
	\end{split}
\end{equation}
where 
\begin{equation}
	\begin{split}
		J^{c_{\text{local}}}_{\bm{\tilde{\xi}},\tilde{\bm{\eta}};\bm{\xi},\bm{\eta}}(\tilde{\bm{z}},\tilde{\bm{\theta}}:\bm{z},\bm{\theta};M)	=\mathcal{L}_{c_{\text{local}}}\left(\{ U_{N_L,\tilde{\bm{\xi}},\tilde{\bm{\eta}}}(\tilde{\bm{z}}_i,\tilde{\bm{\theta}})\ket{0}^{\otimes n}\}_{i=1}^{M}, ~ 
		\{ U_{N_L,\bm{\xi},\bm{\eta}}(\bm{z}_j,\bm{\theta})\ket{0}^{\otimes n}\}_{j=1}^{M}\right).
		\label{eq:definitionNumericalEmpiricalLoss}
	\end{split}
\end{equation}
In Eq.~\eqref{eq:numericalExperiment_self}, we set 
$\bm{\xi}=\tilde{\bm{\xi}}$ and $\bm{\eta}=\tilde{\bm{\eta}}$ for the 
two unitary operators that appear in the argument of 
$\mathcal{L}_{c_{\text{local}}}$ in Eq.~\eqref{eq:definitionNumericalEmpiricalLoss}. 
This indicates that $J^{c_{\text{local}}}_{\bm{{\xi}},{\bm{\eta}};\bm{\xi},\bm{\eta}}(\tilde{\bm{z}},{\bm{\theta}}:\bm{z},\bm{\theta};M)$ in 
Eq.~\eqref{eq:numericalExperiment_self} would become zero in the limit of infinite 
number of training data ($M\to\infty$). 
The expectation in Eq.~\eqref{eq:numericalExperiment_self} is taken with respect 
to the latent variables $\tilde{\bm{z}}_i$ and $\bm{z}_j$ subjected to the  
$N_z$-dimensional uniform distribution $U(0,1)^{N_z}$, but we numerically 
approximate it by $N_{\text{Monte}}$ Monte Carlo samplings. 
Other conditions in the numerical simulation are shown in 
Table~\ref{tab:ExperimentalParamters_dimensionDependence}.

Figure~\ref{fig:numericalExperiment_self} plots the values of 
Eq.~\eqref{eq:numericalExperiment_self} with several $N_z$ (the dimension of the 
latent variables) and $n$ (the number of qubits). 
In the figures, the dotted lines show the scaling curve $M^{-1/N_z}$. 
Notably, in the range of a large number of training data, the points and dotted 
lines are almost consistent, regardless of the number of qubits. 
This implies that the OTL for two different ensembles given by 
Eq.~\eqref{eq:numericalExperiment_self} is almost independent of the number 
of qubits $n$ and depends mainly on the latent dimension $N_z$, likewise the 
case of distance-based loss function proven in Corollary~\ref{co:ConvergenceSpeedWassSame}.

\begin{table}[tb]
	\begin{center}
		\caption{List of parameters for numerical simulation of Sec.~\ref{sec:dependenceDataDimension}}
		\begin{tabular}{cll} \hline
			number of measurements        & $N_{s}$            & $\infty$ (statevector)              \\
			number of training data             & $M$                & $\{2^i|i\in \{0,1,2,\ldots,10\}\}$ \\
			number of Monte Carlo samples & $N_{\text{Monte}}$ & $100$                              \\
			number of qubits              & $n$                &
			\begin{tabular}{l}
				$\{1,2,4,6,8,10\}$ for Experiment A \\
				$\{1,2,4,6,8\}$ for Experiment B
			\end{tabular}                                                              \\
			dimension of latent space     & $N_z$              &
			\begin{tabular}{l}
				$\{1,2,4,6,8\}$ for Experiment A \\
				$\{1,2,4,6,10,14\}$ for Experiment B
			\end{tabular}                                                              \\
			number of layers              & $N_L$              & $3+\lfloor N_z/n\rfloor$           \\
			\hline
		\end{tabular}
		\label{tab:ExperimentalParamters_dimensionDependence}
	\end{center}
\end{table}

\begin{figure}[hbtp]
	\centering
	\subfigure[$n=1$]{\includegraphics[width=.49\linewidth]{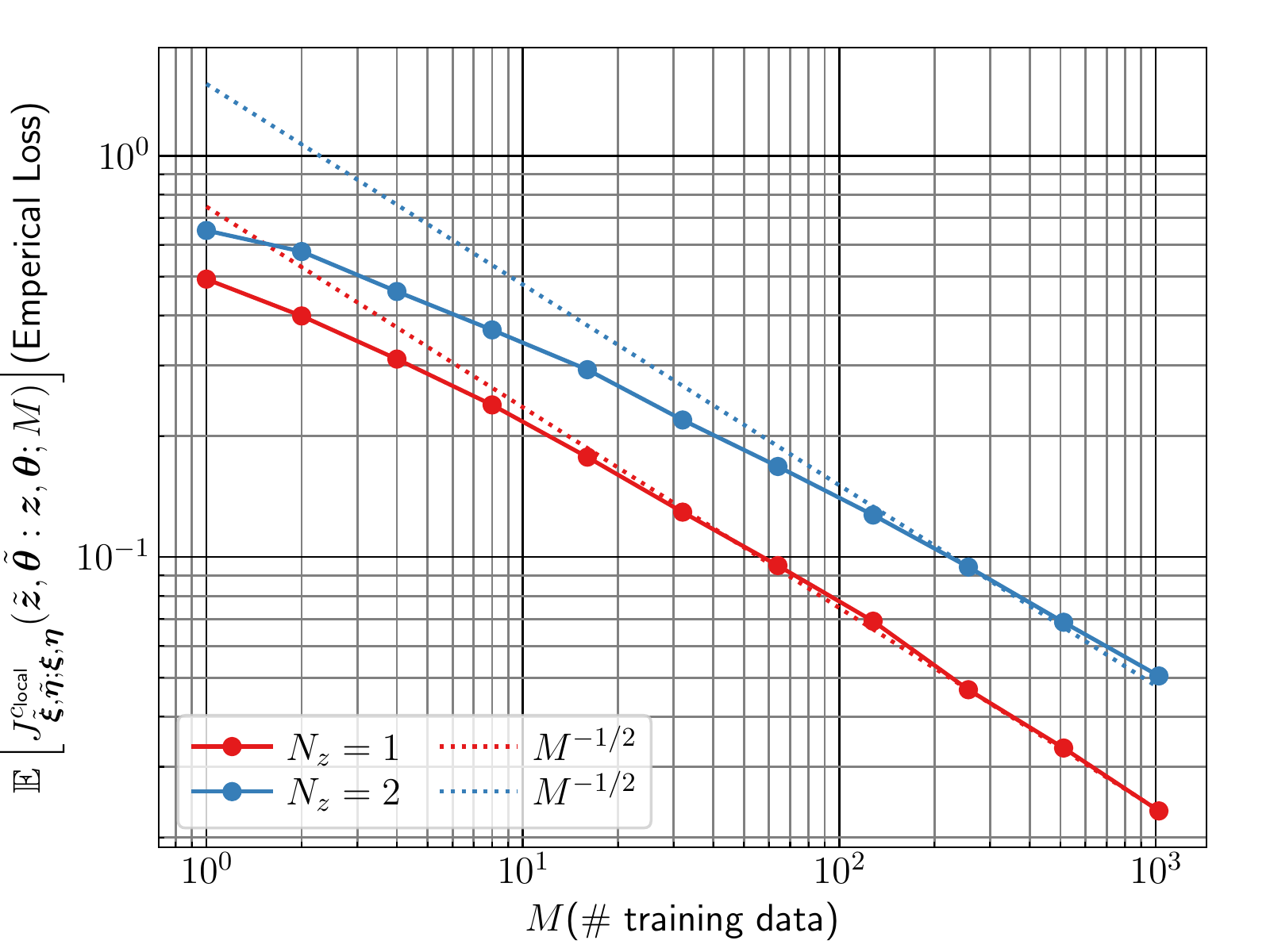}}
	\subfigure[$n=2$]{\includegraphics[width=.49\linewidth]{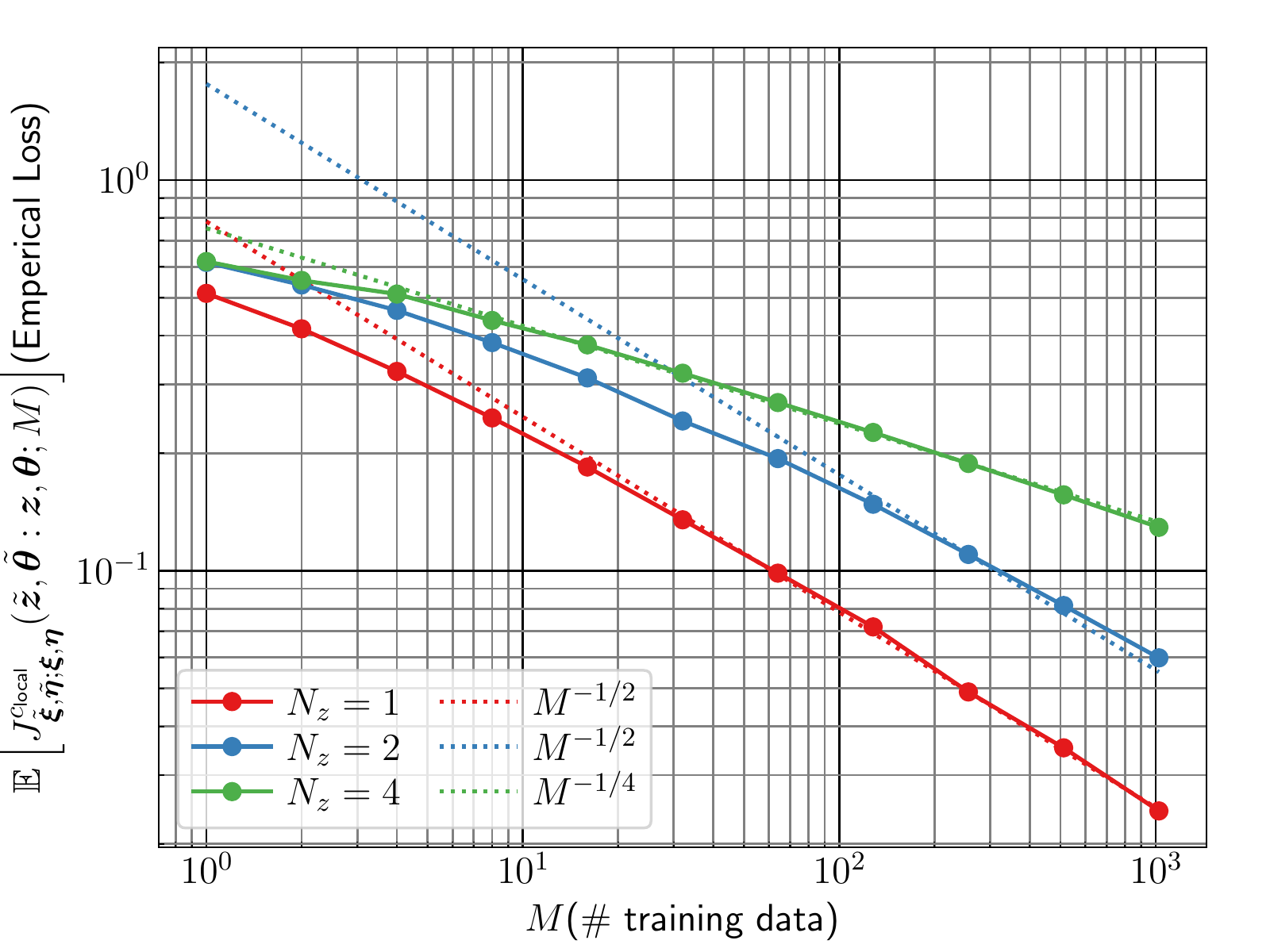}}
	\subfigure[$n=4$]{\includegraphics[width=.49\linewidth]{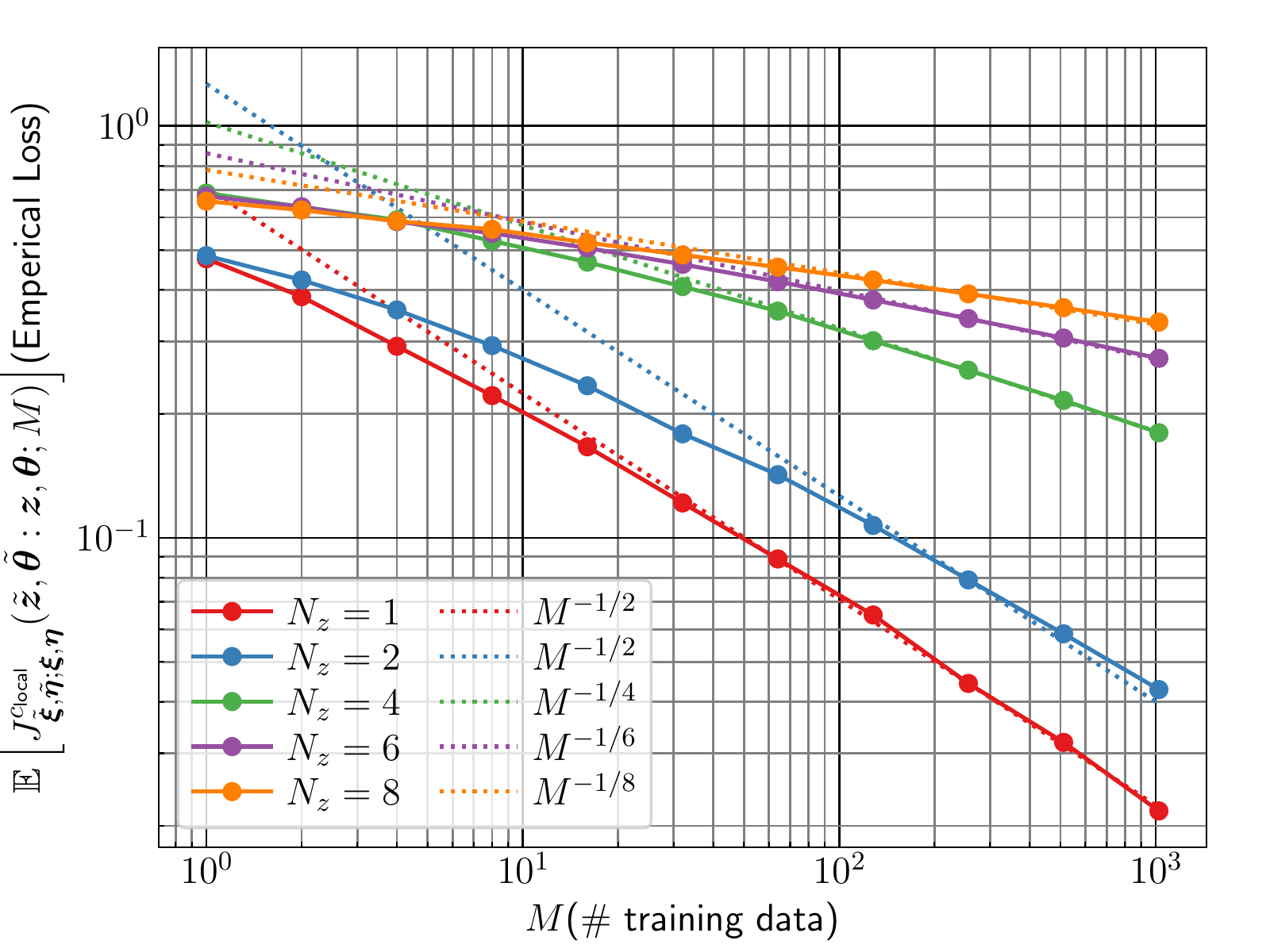}}
	\subfigure[$n=6$]{\includegraphics[width=.49\linewidth]{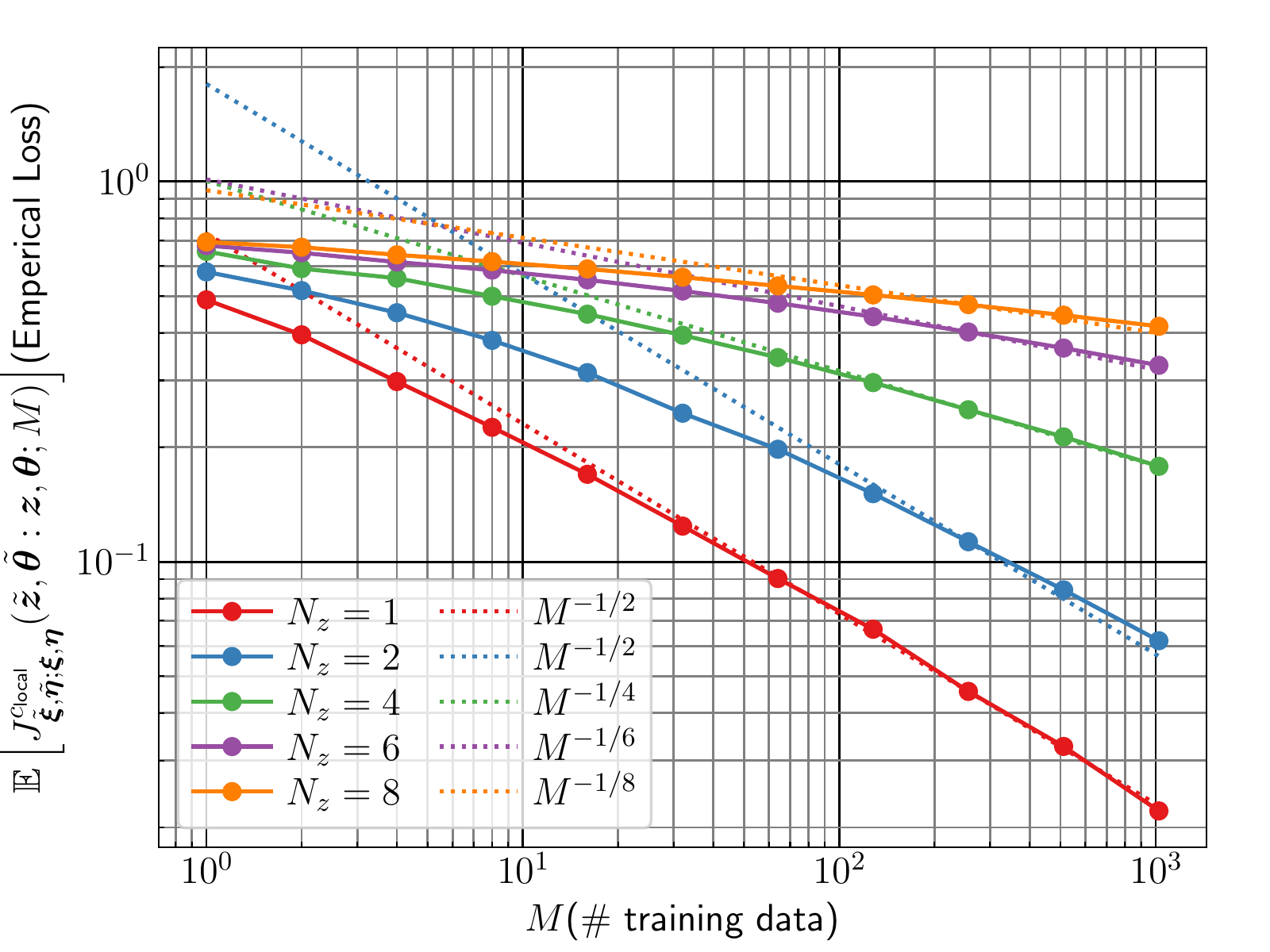}}
	\subfigure[$n=8$]{\includegraphics[width=.49\linewidth]{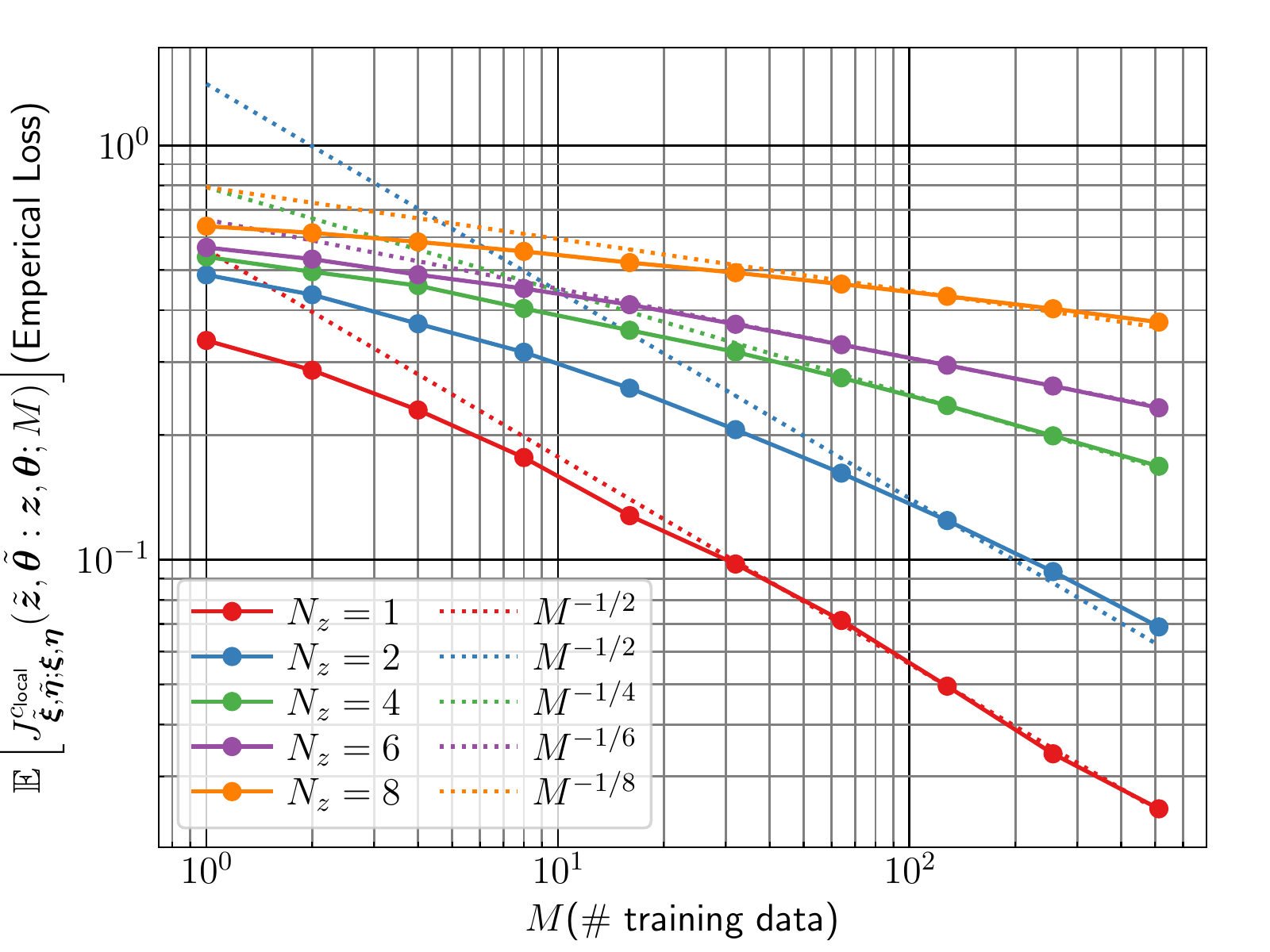}}
	\subfigure[$n=10$]{\includegraphics[width=.49\linewidth]{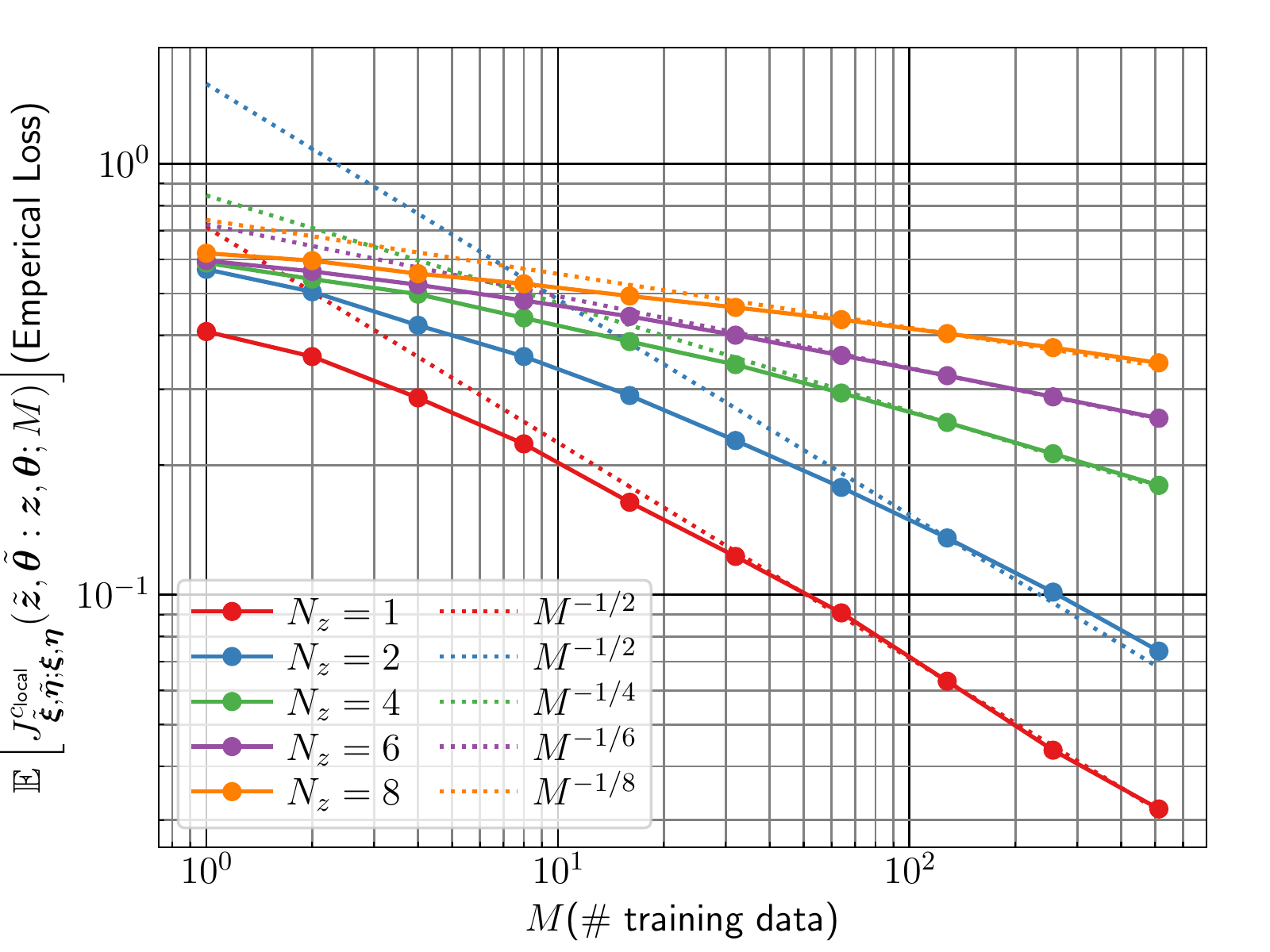}}
	\caption{
		Results of numerical simulations on the relationship between the number 
		of training data and the OTL given by  Eq.~\eqref{eq:numericalExperiment_self}, with several qubit number $n$. 
		Each subgraph shows the results for various latent dimensions $N_z$.
		For reference, the scaling curves $M^{-1/N_z}$ are added as dotted lines. 
		These graphs show that Eq.~\eqref{eq:numericalExperiment_self} mainly scale as $M^{-1/N_z}$ and are almost independent to the number of 
		qubits, $n$. 
	}\label{fig:numericalExperiment_self}
\end{figure}

{\bf Experiment B.} 
We next turn to the second experiment to confirm that the approximation error 
of the proposed OTL scales similar to 
Eq.~\eqref{eq:dimensiondependence_differentdistribution}. 
Specifically, we numerically show the dependence of the following expectation 
value on the number of training data $M$: 
\begin{equation}
	\begin{split}
		\mathbb{E}_{\tilde{\bm{z}},\bm{z}\sim U(0,1)^{N_z}}
		&\left[
			\lim_{K\to\infty}
			J^{c_{\text{local}}}_{\bm{\tilde{\xi}},\tilde{\bm{\eta}};\bm{\xi},\bm{\eta}}(\tilde{\bm{z}},\tilde{\bm{\theta}}:\bm{z},\bm{\theta};K)
			-
			J^{c_{\text{local}}}_{\bm{\tilde{\xi}},\tilde{\bm{\eta}};\bm{\xi},\bm{\eta}}(\tilde{\bm{z}},\tilde{\bm{\theta}}:\bm{z},\bm{\theta};M)
			\right],
		\label{eq:numericalExperiment_differentdistribution}
	\end{split}
\end{equation}
where $J^{c_{\text{local}}}$ is defined in Eq.~\eqref{eq:definitionNumericalEmpiricalLoss}.
In this case, we set different fixed parameters $\bm{\xi}\neq\tilde{\bm{\xi}}$ 
and $\bm{\eta}\neq\tilde{\bm{\eta}}$ for the two unitary operators in 
Eq.~\eqref{eq:definitionNumericalEmpiricalLoss}. 
The parameters used in the numerical simulation are shown in 
Table~\ref{tab:ExperimentalParamters_dimensionDependence}.

The first term of Eq.~\eqref{eq:numericalExperiment_differentdistribution} 
is an ideal quantity assuming that an infinite number of training data is available. 
Since the first term is independent of the number of training data, the second 
term is expected to take the following form, as suggested from  
Eq.~\eqref{eq:dimensiondependence_differentdistribution}:
\begin{equation}
	\mathbb{E}_{\tilde{\bm{z}},\bm{z}\sim U(0,1)^{N_z}}
	\left[
		J^{c_{\text{local}}}_{\bm{\tilde{\xi}},\tilde{\bm{\eta}};\bm{\xi},\bm{\eta}}(\tilde{\bm{z}},\tilde{\bm{\theta}}:\bm{z},\bm{\theta};M)\right]
	= a M^{-1/b} + c.
	\label{eq:fitting}
\end{equation}
To identify the parameter $b$, we use the Monte Carlo method to calculate the 
left hand side of Eq.~\eqref{eq:fitting} as a function of $M$ and then execute 
the curve-fitting via $a M^{-1/b} + c$. 
We repeat this procedure with several values of the number of qubits $n$ and 
the latent dimension $N_z$; see Appendix~\ref{sec:fittingresult} for a more 
detailed discussion. 
The result of parameter identification is depicted in Fig.~\ref{fig:result_fittingparameter}, which shows that the fitting parameter 
$b$ is almost independent to the number of qubits $n$ and linearly scales with 
respect to the latent dimension $N_z$. 
This result is indeed consistent with 
Eq.~\eqref{eq:dimensiondependence_differentdistribution}.

\begin{figure}[hbtp]
	\centering
	\subfigure[]{\includegraphics[width=.49\linewidth]{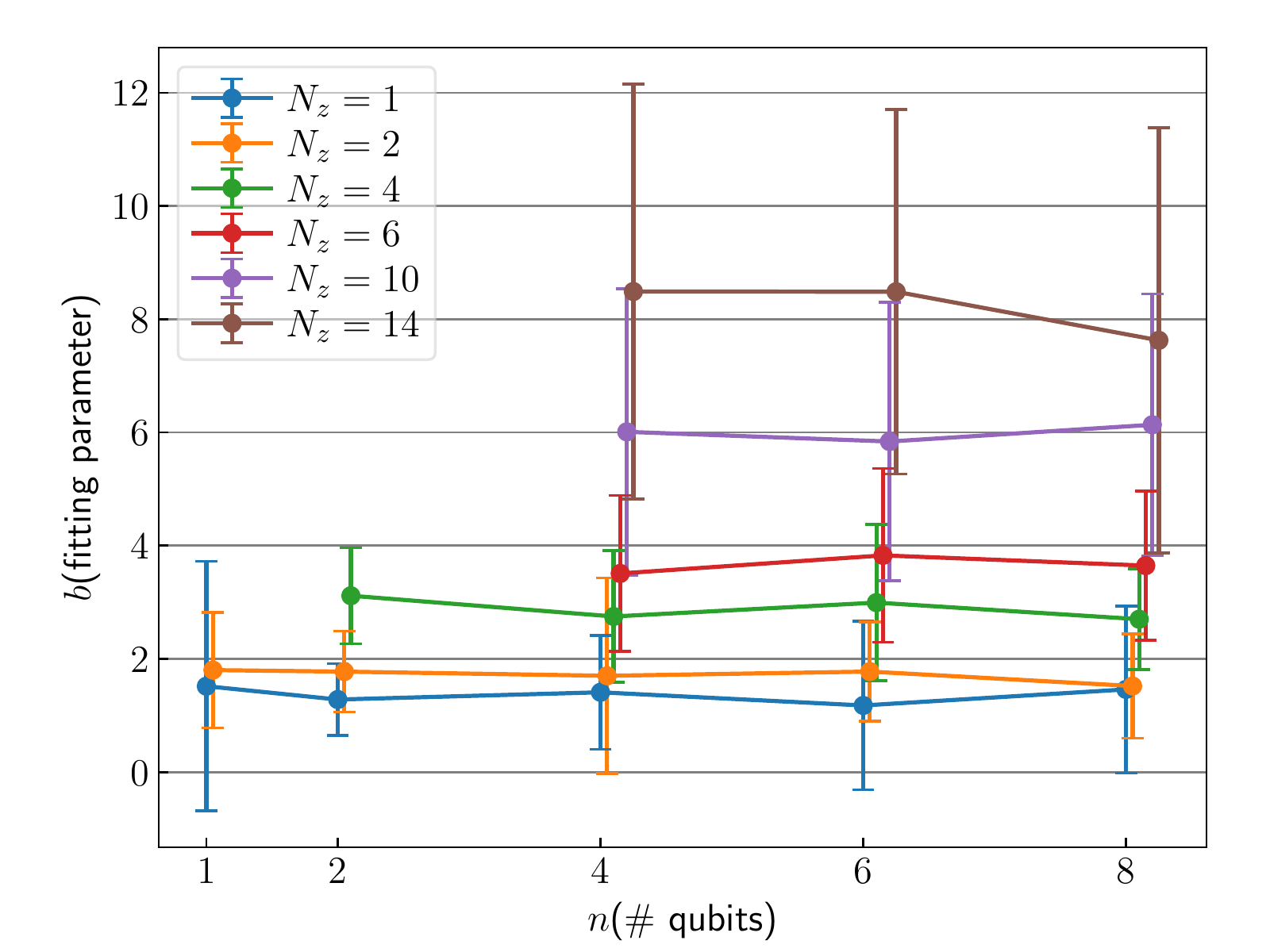}}
	\subfigure[]{\includegraphics[width=.49\linewidth]{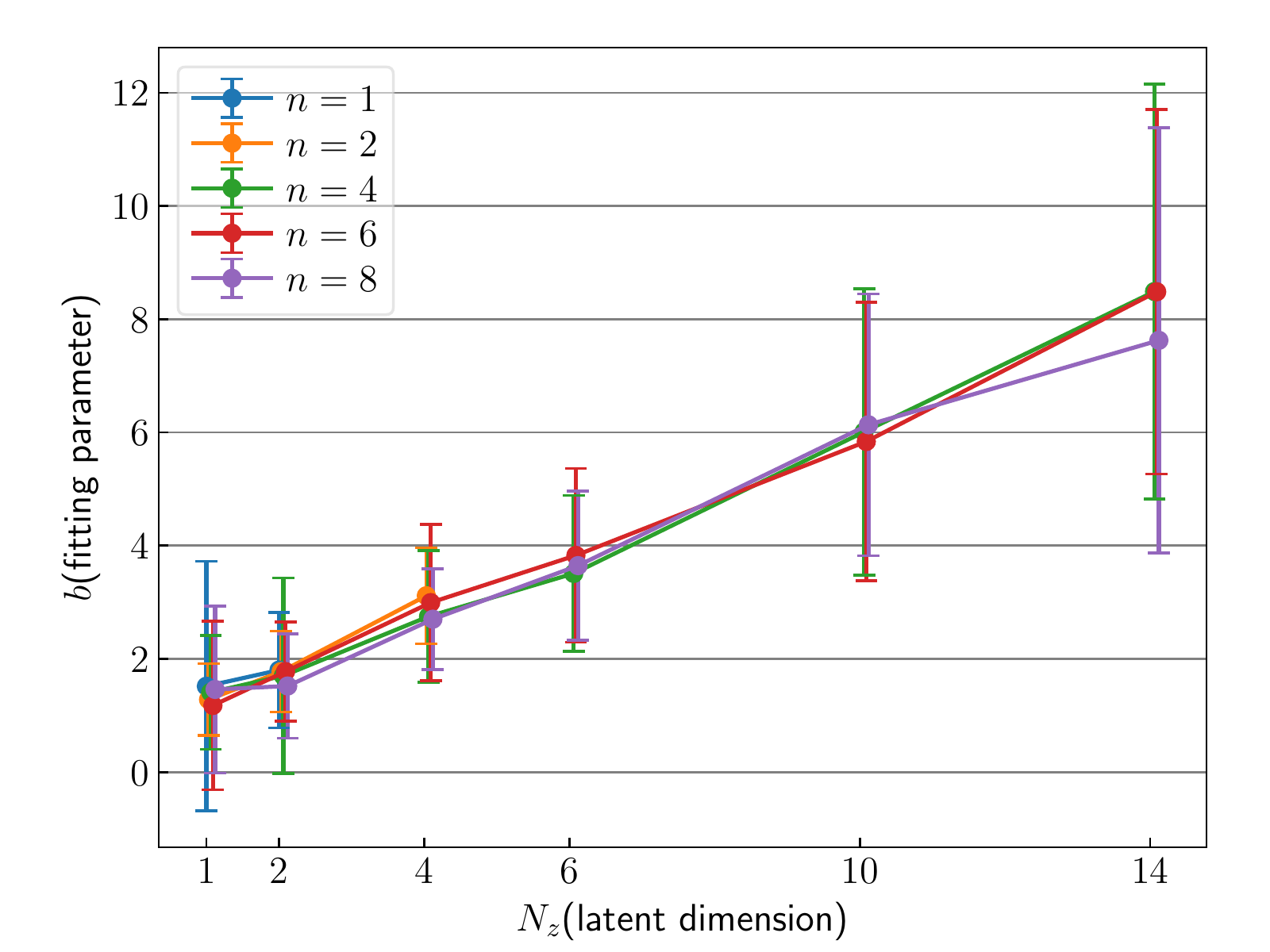}}
	\caption{
		The simulation results of the fitting parameter $b$ as a function of 
		(a) the number of qubits $n$ and (b) the latent dimension $N_z$. 
		The fitting parameter $b$ is obtained by fitting the second term of 
		Eq.~\eqref{eq:numericalExperiment_differentdistribution} by using 
		Eq.~\eqref{eq:fitting}.
		The subfigure (a) shows that $b$ is almost independent to the number 
		of qubits $n$, while the subgraph (b) shows that $b$ linearly scales 
		with the latent dimension $N_z$.
	}\label{fig:result_fittingparameter}
\end{figure}

The results obtained in Experiments~A and B suggest us to have the following 
conjecture: 
\begin{conjecture}
The scaling of the approximation error of the optimal transport loss 
\eqref{eq:quantumEmpiricalLoss} with respect to the number of training data $M$ 
is determined via the latent dimension $N_z$ as follows: 
    \begin{equation}
			\mathbb{E}_{\tilde{\bm{z}},\bm{z}\sim U(0,1)^{N_z}}
			\left[ \mathcal{L}_{c_{\rm{local}}}\left(\{ U(\tilde{\bm{z}}_i,\bm{\theta})\ket{0}^{\otimes n}\}_{i=1}^{\infty},\{ U(\bm{z}_j,\bm{\theta})\ket{0}^{\otimes n}\}_{j=1}^{M}\right) \right] 
			\lesssim O({M}^{-1/N_z}),
	\end{equation}
without respect to the number of qubits. 
\label{obs:zdependence}
\end{conjecture}
This conjecture means that the proposed OTL can be efficiently computed 
via sampling, when the intrinsic dimension of the data and the latent 
dimension are sufficiently small. 
Proving this conjecture would be challenging and is the subject of future work.

\subsection{Approximation error as a function of the number of shots}
\label{sec:dependenceShotNumber}

\begin{figure}[hbtp]
	\centering
	\subfigure[$N_z=1$]{\includegraphics[width=.49\linewidth]{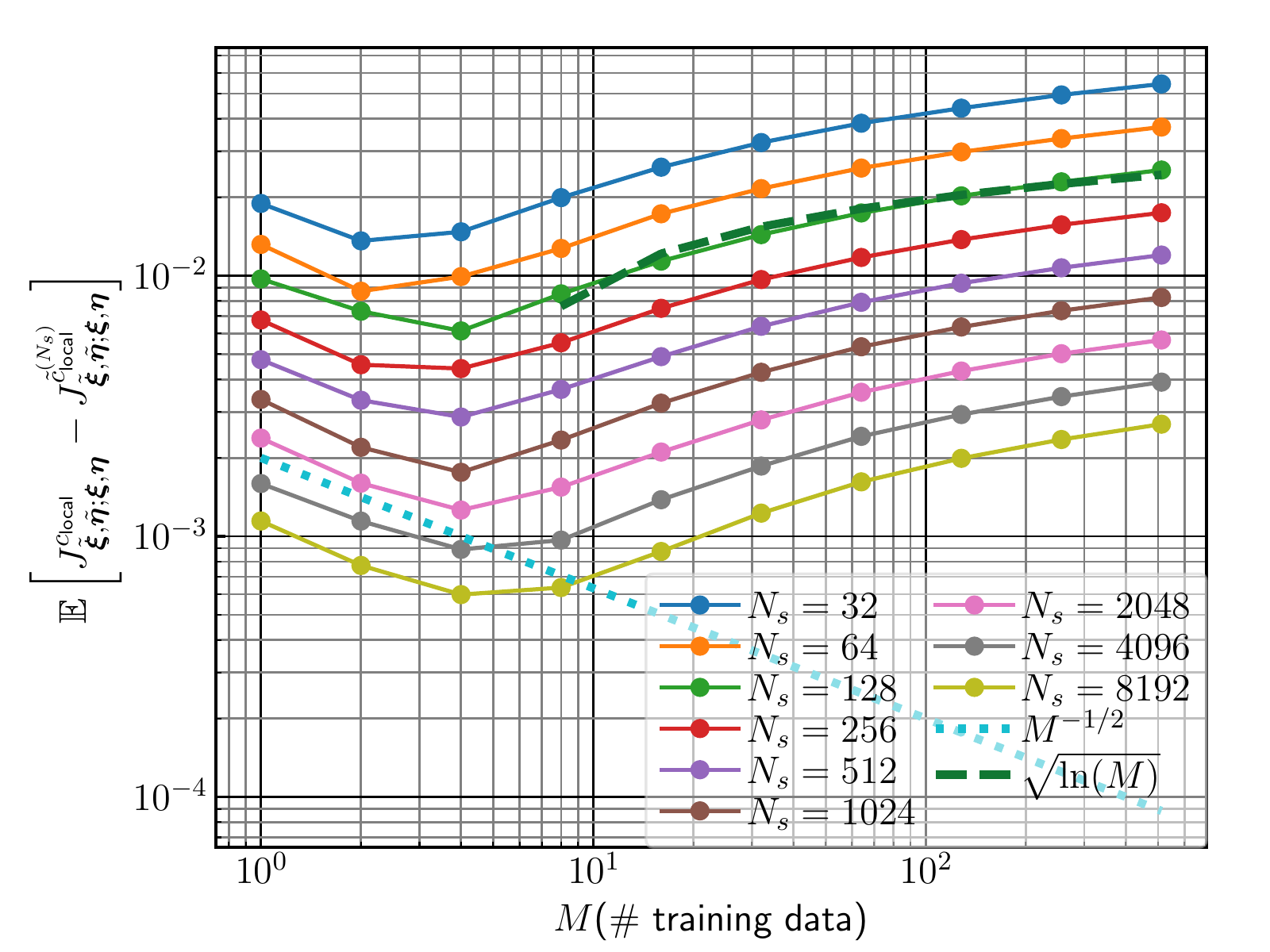}
		\includegraphics[width=.49\linewidth]{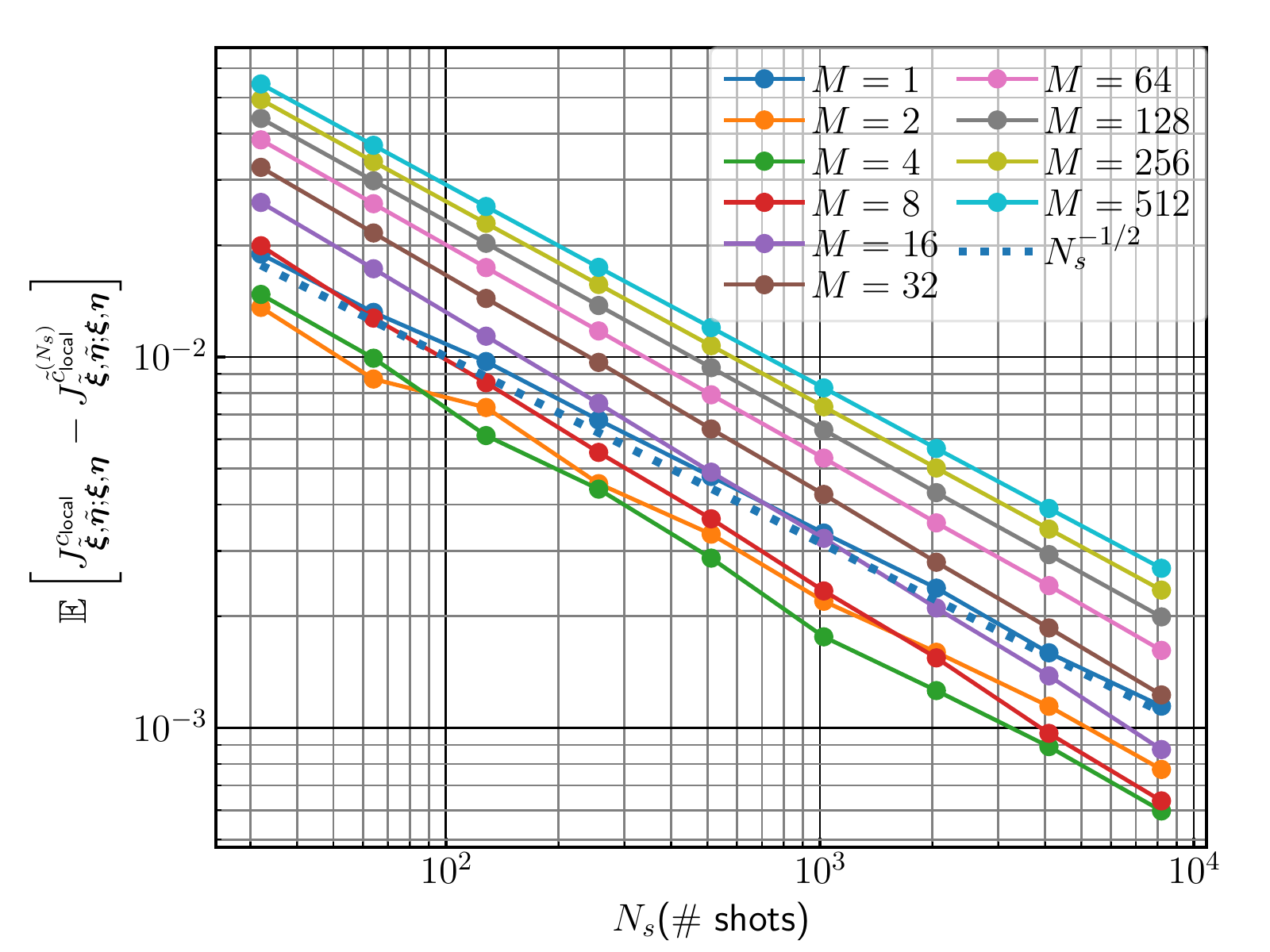}}
	\subfigure[$N_z=2$]{\includegraphics[width=.49\linewidth]{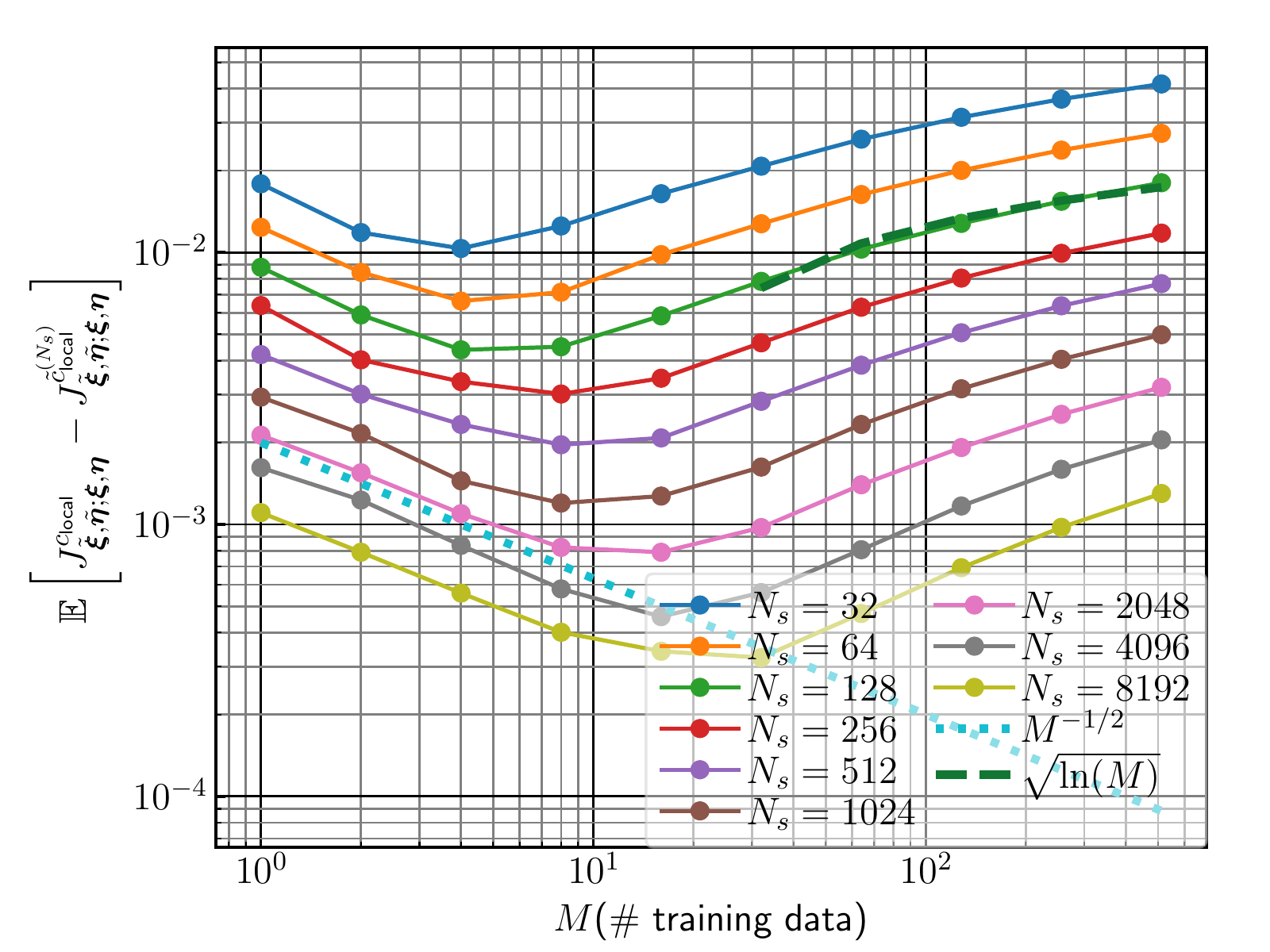}
		\includegraphics[width=.49\linewidth]{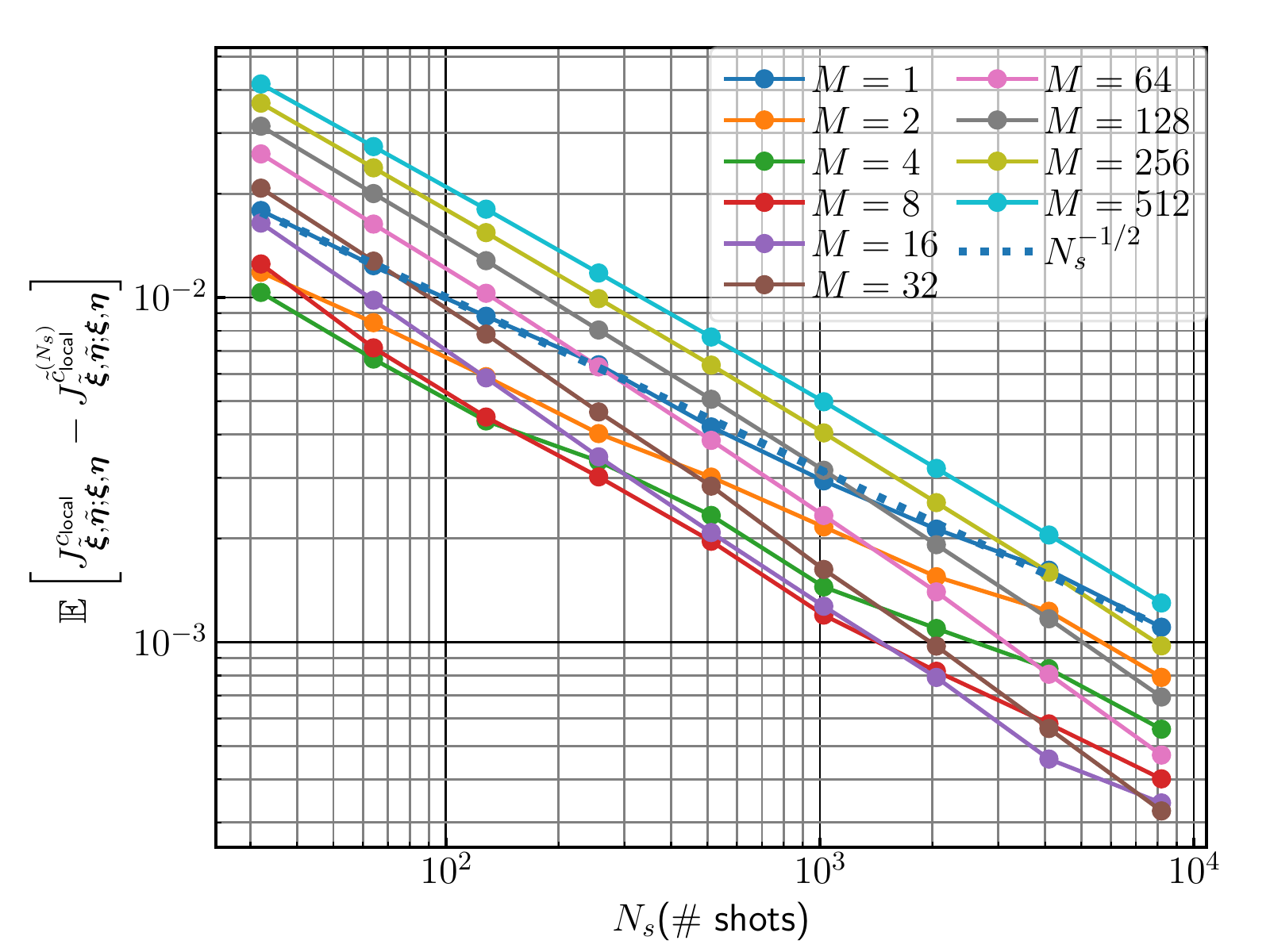}}
	\subfigure[$N_z=4$]{\includegraphics[width=.49\linewidth]{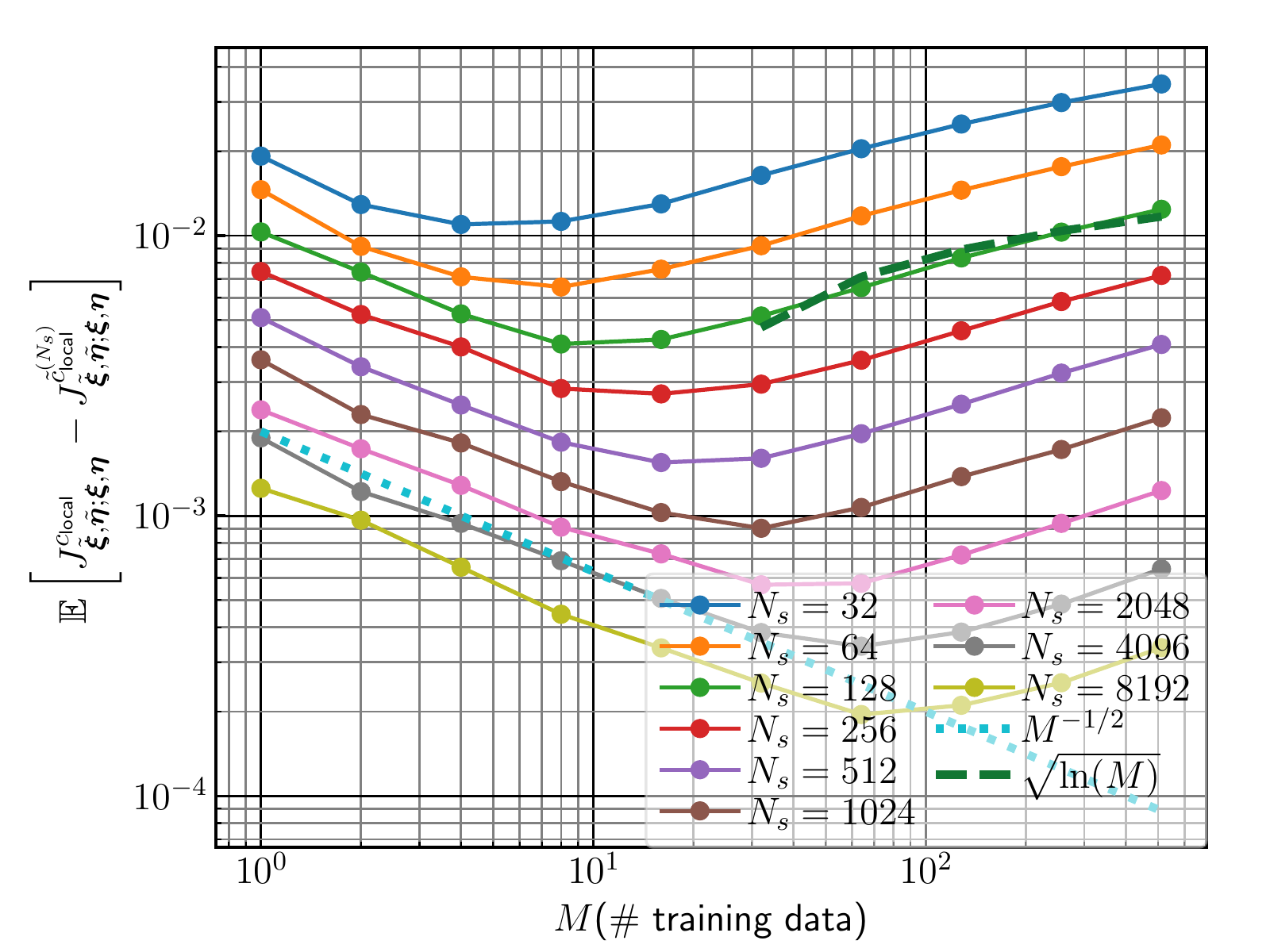}
		\includegraphics[width=.49\linewidth]{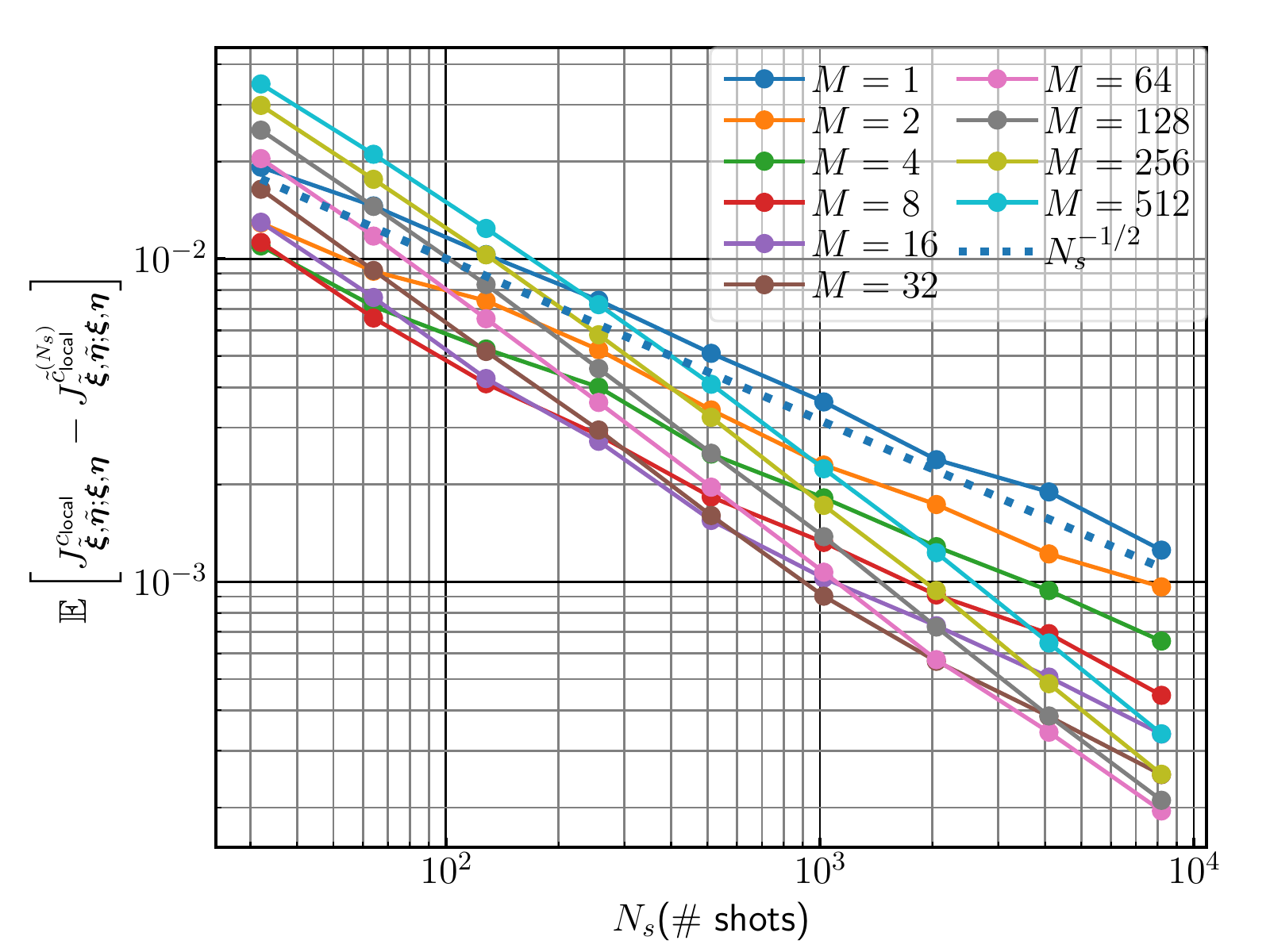}}
	\caption{
		Simulation results of the approximation error of the OTL due to 
		the number of shots defined in Eq.~\eqref{eq:shotDependence}. 
		The left and right panels show the dependence of the error on the number 
		of training data $M$ and shots $N_s$, respectively.
		The results for different latent dimension $N_z$ are shown in the figures 
		(a), (b), and (c). 
		For reference, we added the curve of $M^{-1/2}$ with the dotted line 
		in the left panels. 
		The fitting result of the points $N_s =128$ with the curve $\sqrt{c_1\ln(M)+c_2}$ is added as the dashed line in the left panels. 
		Also, we added the curve of $N_s^{-1/2}$ as the dotted line in the 
		right panels. 
	}\label{fig:numericalExperiment_shotDependence}
\end{figure}

The error analysis in Sec.~\ref{sec:dependenceDataDimension} assumes that 
the number of shot is infinite and thereby the ground cost between quantum 
states can be perfectly determined. 
Here, we analyze the effect of the finiteness of the number of shots on the 
approximation error.
The following proposition serves as a basis of the analysis.

\begin{proposition}
	Let $\tilde{c}^{(N_s)}_{\rm{local}}$ be an estimator of the ground cost $c_{\rm{local}}$ of Eq.~\eqref{eq:localcost} using $N_s$ samples.
	Suppose that the support of two different probability distributions are strictly separated; as a result, there exists a lower bound $g>0$ to 
	the ground cost for any $i,j\in\{1,2,\ldots,M\}$, 
	i.e., $(c_{\rm{local}}(\ket{\psi_i},U(\bm{z}_j,\bm{\theta})\ket{0}^{\otimes n})>g, 
	\forall i,j)$.
	Then, for any positive constant $\delta$, the following inequality holds: 
	\begin{equation}
		\begin{split}
			P&\left( \Big|
				\mathcal{L}_{c_{\rm{local}}}\left(\{\ket{\psi}_i\}_{i=1}^{M},\{ U(\bm{z}_j,\bm{\theta})\ket{0}^{\otimes n}\}_{j=1}^{M}\right)
			-
			\mathcal{L}_{\tilde{c}^{(N_s)}_{\rm{local}}}\left(\{\ket{\psi}_i\}_{i=1}^{M},\{ U(\bm{z}_j,\bm{\theta})\ket{0}^{\otimes n}\}_{j=1}^{M}\right)
			\Big|\right. \\
			&\hspace{4cm} \left.\ge  \sqrt{\frac{2M}{\delta}}\sqrt{ \frac{1-g}{N_s} + \frac{(1-g)^2}{4N_s^2 g}}+\frac{1-g}{2 N_s \sqrt{g}}\right)  \le \delta.
		\end{split}
	\end{equation}
	\label{prop:err_shotdependence}
\end{proposition}

The proof is shown in Appendix~\ref{sec:proof_shotdependence}. 
Proposition~\ref{prop:err_shotdependence} states that the approximation error 
of the OTL is upper bounded by a constant of the order 
$O(\sqrt{M/N_s})$, under the condition $M\gg 1$ and $N_s \gg 1$  . 
Therefore, if Observation~\ref{obs:zdependence} is true, the approximation 
error due to the finiteness of $N_s$ and $M$ is upper bounded by  
$O(M^{-1/N_z})+O\left(\sqrt{{M}/N_s}\right)$, where $N_z$ is the latent 
dimension.

Then, we provide a numerical simulation to show the averaged approximation error 
as a function of $M$ as well as $N_s$. 
Again, we employ the hardware efficient ansatz shown in 
Fig.~\ref{fig:ansatzcircuit} of Sec.~\ref{sec:dependenceDataDimension}. 
The purpose of the numerical simulation is to see the dependence of the 
following expectation value on $M$ and $N_s$. 
\begin{equation}
	\mathbb{E}_{\tilde{\bm{z}},\bm{z}\sim U(0,1)^{N_z}}\left[
	\left|
	J^{\tilde{c}^{(N_s)}_{\text{local}}}_{\bm{\tilde{\xi}},\tilde{\bm{\eta}};\bm{\xi},\bm{\eta}}(\tilde{\bm{z}},\tilde{\bm{\theta}}:\bm{z},\bm{\theta};M)
	-
	J^{c_{\text{local}}}_{\bm{\tilde{\xi}},\tilde{\bm{\eta}};\bm{\xi},\bm{\eta}}(\tilde{\bm{z}},\tilde{\bm{\theta}}:\bm{z},\bm{\theta};M)
	\right|
	\right],
	\label{eq:shotDependence}
\end{equation}
where $J^{\tilde{c}^{(N_s)}_{\text{local}}}$ can be computed via 
Eq.~\eqref{eq:definitionNumericalEmpiricalLoss}. 
As in the numerical simulation in Sec.~\ref{sec:dependenceDataDimension}, we 
approximate the expectation 
$\mathbb{E}_{\tilde{\bm{z}},\bm{z}\sim U(0,1)^{N_z}}[\bullet]$ by 
Monte Carlo sampling of $\tilde{\bm{z}}$ and $\bm{z}$ from the uniform 
distribution $U(0,1)^{N_z}$. 
Also, we randomly choose the fixed parameters  
$\bm{\xi},\bm{\eta},\bm{\theta},\tilde{\bm{\xi}},\tilde{\bm{\eta}},\tilde{\bm{\theta}}$ prior to the simulation.
The other simulation parameters are given in 
Table~\ref{tab:ExperimentalParamters_shotDependence}. 
Simulation results are depicted in Fig.~\ref{fig:numericalExperiment_shotDependence}, where the notable points are summarized as follows. 
\begin{itemize}
	\item In the range of small number of training data $M$, the approximation 
	error is roughly proportional to $M^{-1/2}$.
	\item In the range of large number of training data $M$, the approximation 
	error takes $\sqrt{c_1 \ln M+c_2}$ with constants $c_1$ and $c_2$.
	
	\item The dependence of the error on the number of shots $N_s$ is roughly proportional to $N_s^{-1/2}$.
\end{itemize}
Appendix~\ref{sec:explain_shotDependence} provides an intuitive explanation 
of these results. 
In particular, it is important to know that we need to choose a proper number 
of samples to reduce the approximation error.

\begin{table}[tb]
	\begin{center}
		\caption{List of parameters for numerical simulation of Sec.~\ref{sec:dependenceShotNumber}}
		\begin{tabular}{cll} \hline
			number of measurements        & $N_{s}$            & $\{2^{i+7}|\ i\in \{0,1,2,\ldots,7\}\}$, $\infty$(statevector) \\
			number of training data             & $M$                & $\{2^i|\ i\in \{0,1,2,\ldots,10\}\}$                           \\
			number of Monte Carlo samples & $N_{\text{Monte}}$ & $256$                                                          \\
			number of qubits              & $n$                & $8$                                                            \\
			dimension of latent space     & $N_z$              & $\{1,2,4\}$                                                    \\
			number of layers              & $N_L$              & $3+\lfloor N_z/n\rfloor$                                       \\
			\hline
		\end{tabular}
		\label{tab:ExperimentalParamters_shotDependence}
	\end{center}
\end{table}

\subsection{Avoidance of the vanishing gradient issue}
\label{sec:barrenplateau}

In Sec.~\ref{sec:vanishingGradientVQA} we chose the cost function composed 
of the local measurements, as a least condition to avoid the vanishing 
gradient issue. 
Note that, however, employing a local cost is not enough to avoid this issue; 
for instance, Ref.~\cite{cerezo2021cost} proposed the method of using a special type of 
parametrized quantum circuit called the {\it alternating layered ansatz (ALA)} 
in addition to using the local cost, which is actually proven to avoid the 
issue. 
Nevertheless, we here numerically demonstrate that our method can certainly 
mitigate the decrease of gradient even without such additional condition, 
compared to the case with global cost.

More specifically, we calculated the expectation of the variance of the partial 
derivative of the OTL \eqref{eq:quantumEmpiricalLoss}, based on the training 
ensemble $\{\ket{\psi_i}\}^{M}_{i=1}$ and the sampled data from the generative 
model $\{U(\bm{z}_j,\bm{\theta})\ket{0}^{\otimes n}\}^{M}_{j=1}$; 
\begin{equation}
\label{eq:barren plateau}
	\mathbb{V}_{\bm{\xi},\bm{\eta},\bm{\theta},\bm{z}}\left[\frac{\partial}{\partial \theta}
	\mathcal{L}_{c_{\text{local}}}\left(\{ \ket{\psi}_i\}_{i=1}^{M},\{ U_{N_L,\bm{\xi},\bm{\eta}}(\bm{z}_j,\bm{\theta})\ket{0}^{\otimes n}\}_{j=1}^{M}\right)\right].
\end{equation}
The partial derivative is calculated using the parameter shift rule~\cite{mitarai2018quantum}. 
The expectation of the variance is approximated by Monte Carlo calculations 
with respect to $\bm{z}$, $\bm{\xi}$, $\bm{\eta}$, and $\bm{\theta}$, where 
$\bm{z}$ is sampled from the uniform distribution $U(0,1)$ and 
$\bm{\xi},\bm{\eta},\bm{\theta}$ are randomly chosen from 
${\bm{\xi}}\in\{X,Y,Z\}^{n N_L}$, $\bm{\eta}\in\{0,1,2,\ldots,N_z\}^{n N_L}$, 
and $\bm{\theta} \in [0,2\pi]^{n N_L}$. 
The structure of the generative model $U(\bm{z},\theta)\ket{0}^{\otimes n}$ 
is the same as that shown in Fig.~\ref{fig:ansatzcircuit}. 
The derivative is taken with respect to $\theta_{1,1}$. 
The training ensemble $\{\ket{\psi_i}\}_{i=1}^{M}$ is prepared as follows;
\begin{align*}
    \ket{\psi}_i 
       = W' V'_2(\bm{\zeta}^i_2)W' V'_1(\bm{\zeta}^i_1)\ket{0}^{\otimes n}, 
\end{align*}
where $\bm{\zeta}^i_1=\{\zeta^i_{1,j}\}_{j=1}^n$ and 
$\bm{\zeta}^i_2=\{\zeta^i_{2,j}\}_{j=1}^n$ are randomly chosen from the uniform 
distribution on $[0,2\pi]$ and fixed during Monte Carlo calculation. 
The operators $W'$, $V'_1$, and $V'_2$ are defined as follows:
\begin{equation}
	\begin{split}
		W' =\prod_{i=1}^{n-1} CX_{j,j+1}, ~~ 
		V'_1(\bm{\zeta}^i_1) =\prod_{j=1}^n R_{j,Y}(\zeta^i_{1,j}), ~~ 
		V'_2(\bm{\zeta}^i_2) =\prod_{j=1}^n R_{j,Z}(\zeta^i_{2,j}),
	\end{split}
\end{equation}
where $CX_{j,k}$ denotes the controlled-$X$ gate, which operates $X$ gate on 
the $k$-th qubit with the $j$-th control qubit. 
$R_{j,Y}$ and $R_{j,Z}$ denote single qubit Pauli rotations around the $x$ and 
$y$ axes, respectively. 
The other simulation parameters are given in 
Table~\ref{tab:ExperimentalParamters_barrenplateau}.

Figure~\ref{fig:grad_n} shows the numerical simulation result of the 
variance \eqref{eq:barren plateau}, for the cases where (a) the global cost 
\eqref{eq:tracedistance} is used and (b) the local cost \eqref{eq:localcost} 
is used. 
The number of training data is fixed to $M=8$. 
The clear exponential decays in variance of gradient are observed for the global 
cost, regardless of $N_L$. 
In contrast, for the case of local cost, the relatively shallow circuits with $N_L=10, 25$ 
exhibit approximately constant scaling with respect to $n \geq 10$, while the deep 
circuits with $N_L\geq 50$ also exhibit slower scaling than the global one and 
keep larger variance even when $n \geq 8$. 
This result implies that the OTL with local ground cost can avoid the gradient 
vanishing issue, despite that the circuit is not specifically designed for 
this purpose. 
Note also that the result is consistent with that reported in 
\cite{holmes2022connecting}, which studied the cost function composed of single ground cost of our setting.

In addition, Fig.~\ref{fig:grad_n}(c) shows the variance 
\eqref{eq:barren plateau} as a function of the number of training data, $M$, 
in the case of $n=14$. 
In the figure, the points represent the Monte-Carlo numerical results and the 
dotted lines represent the scaling curves $M^{-x}$ where the value $x$ is 
determined via fitting. 
This fitting result implies that the gradient obeys the simple statistical 
scaling with respect to $M$ and thus the proposed algorithm would enjoy 
efficient learning even for a large training ensemble.

\begin{table}[tb]
	\begin{center}
		\caption{List of parameters for numerical simulation of Sec.~\ref{sec:barrenplateau}}
		\begin{tabular}{cll} \hline
			number of measurements        & $N_{s}$            & $\infty$ (statevector)          \\
			number of training data       & $M$                & $\{2^i|\ i\in \{1,2,3,4\}\}$   \\
			number of Monte Carlo samples & $N_{\text{Monte}}$ & $300$                          \\
			number of qubits              & $n$                & $\{2,4,6,8,10,12,14\}$         \\
			dimension of latent space     & $N_z$              & $\{1\}$                      \\
			number of layers              & $N_L$              & $\{10,25,50,75,100,200\}$      \\
			\hline
		\end{tabular}
		\label{tab:ExperimentalParamters_barrenplateau}
	\end{center}
\end{table}

\begin{figure}[hbtp]
	\centering
	\subfigure[Global cost]{\includegraphics[width=.32\linewidth]{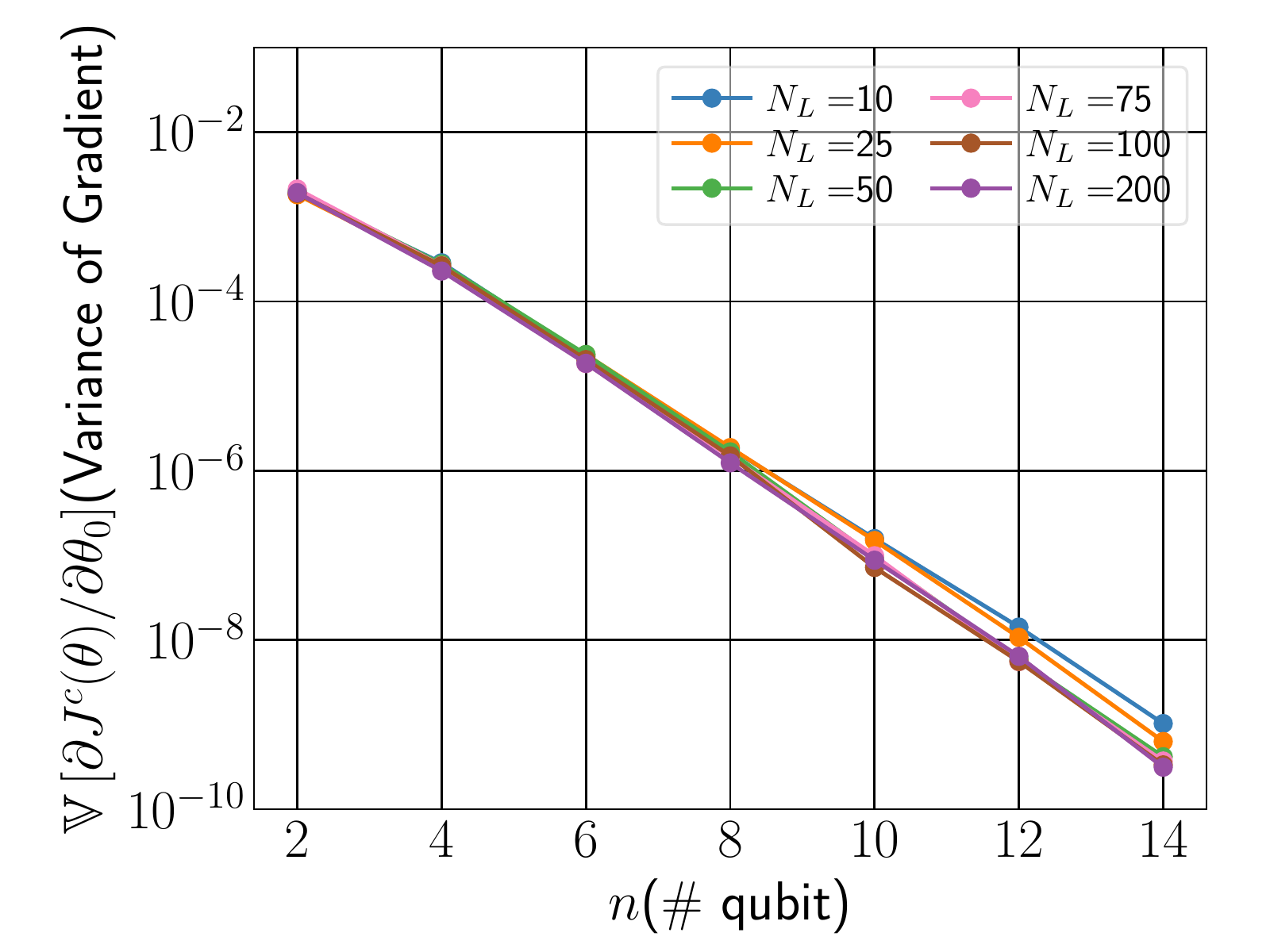}}
	\subfigure[Local cost]{\includegraphics[width=.32\linewidth]{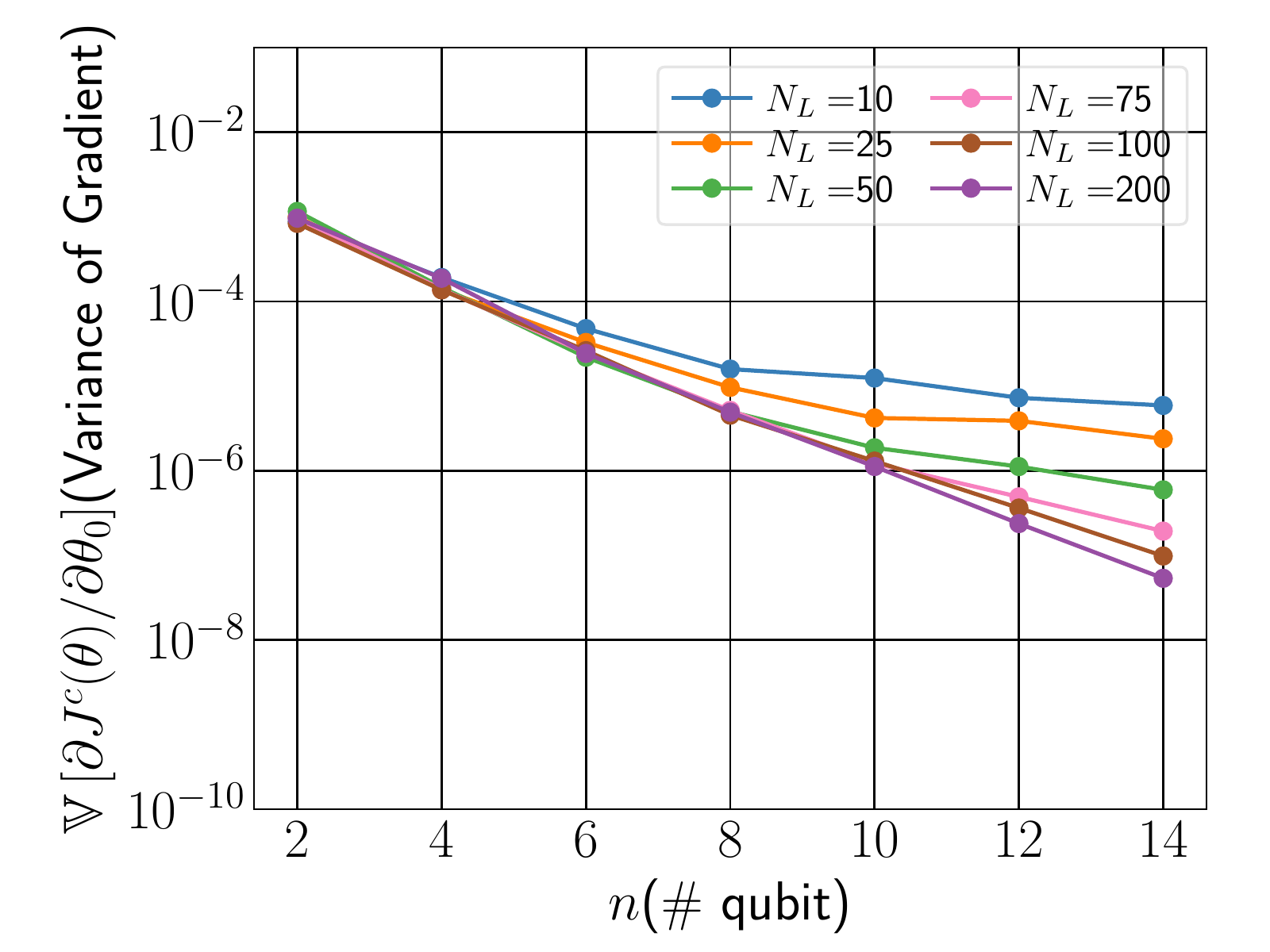}}
	\subfigure[Training number dependence(Local cost)]{\includegraphics[width=.32\linewidth]{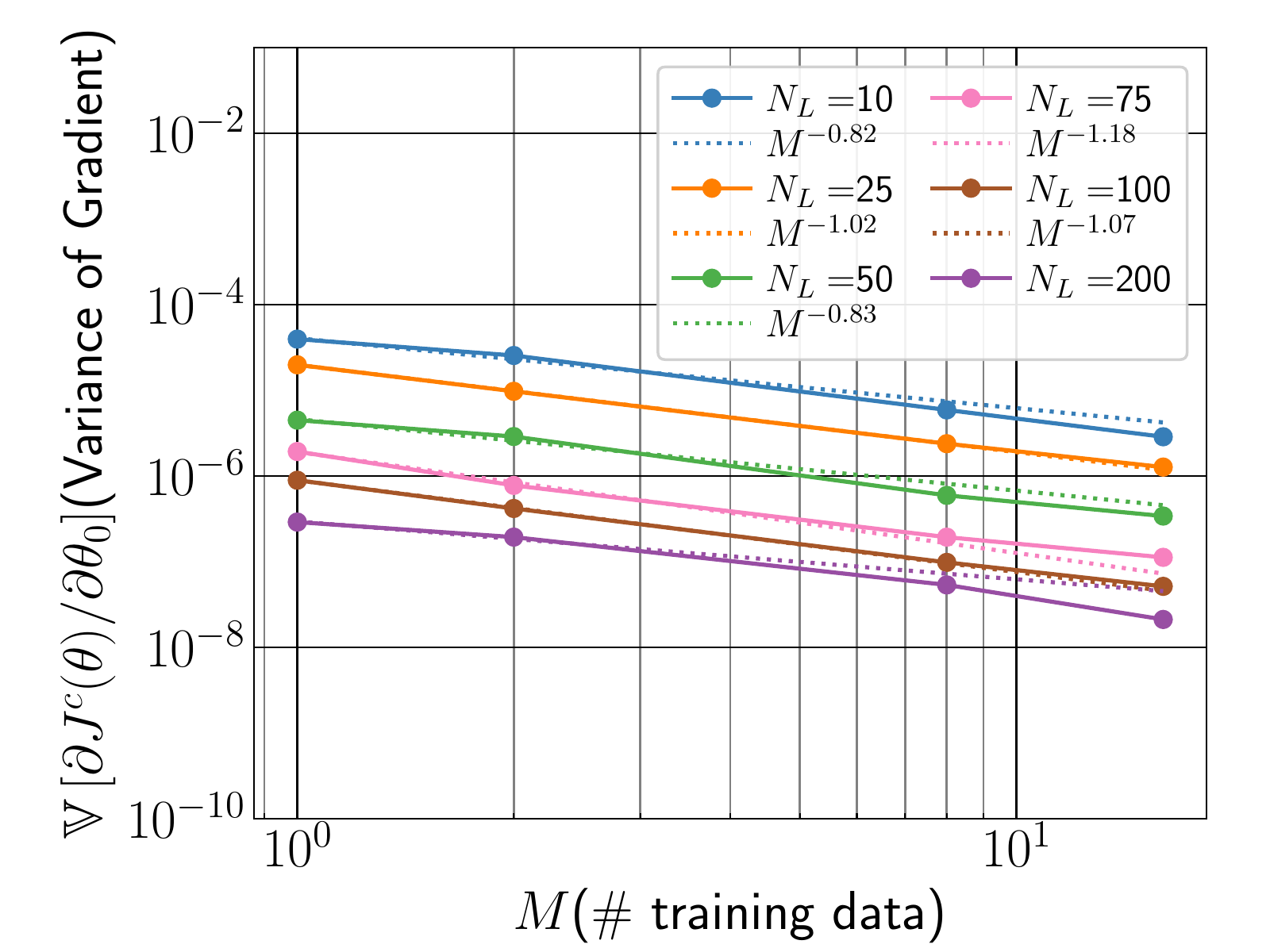}}
	\caption{
		The gradient of the cost function as a function of the number of qubit, 
		$n$, for the case of (a) the global cost \eqref{eq:tracedistance} and 
		(b) the local cost \eqref{eq:localcost}. 
		Figure (c) shows the dependence on the number of training data, $M$, for 
		the case of local cost. 
		Several curves are depicted for different number of layers $N_L$. 
		The clear exponential decay is observed in (a), but is avoided in (b).
		The polynomial decay ($\simeq M^{-1}$) is observed in (c), implying the 
		simple statistical scaling.
	}\label{fig:grad_n}
\end{figure}

\section{Demonstration of the generative model training and application to 
anomaly detection}
\label{sec:demonstration}

In this section, we present anomaly detection based on the cost function defined in Eq.~\eqref{eq:optimaltransportloss} as a proof-of-concept of the proposed loss function.

\subsection{Quantum Anomaly Detection}
\label{sec:anomalydetection}

Anomaly detection is a task to judge whether a given test data $\bm{x}^{(t)}$ 
is anomalous (rare) data or not, based on the knowledge learned from the past 
training data $\bm{x}_i,(i=1,2,\ldots,M)$, i.e., the generative model. 
Unlike typical  classification tasks, this problem deals with a large 
imbalance in the number of normal data and that of anomalous data; actually, 
the former is usually much bigger than the latter. 
Therefore, typical classification methods are not suitable to solve this task, 
and some specialized schemes have been widely developed 
\cite{chandola2009anomaly}.

The anomaly detection problem is important in the field of quantum technology. 
That is, to realize accurate state preparation and control, we are required 
to detect contaminated quantum states and remove those states as quickly as 
possible. 
Previous quantum anomaly detection schemes rely on the measurement-based 
data processing \cite{hara2014anomaly,hara2016quantum}, which however 
require a large number of measurement as in the case of quantum state 
tomography. 
In contrast, our anomaly detection scheme directly inputs quantum states to 
the constructed generative model and then diagnoses the anormality with 
much fewer measurements.

The following is the procedure for constructing the conventional anomaly 
detector based on the generative model \cite{ide2015introduction}, which 
we apply to our quantum case. 
\begin{enumerate}
	\item(\textit{Distribution estimation}): 
	Construct a model probability distribution from the normal dataset.  \label{item:procedure1_DistributionEstiomation}
	\item(\textit{Anomaly score design}): 
	Define an Anomaly Score (AS), based on the model distribution of normal data. 
	\label{item:procedure2_DefineAnomalyScore}
	\item(\textit{Threshold determination}): 
	Set a threshold of AS for diagnosing the anormality.  \label{item:procedure3_SetThreshold}
\end{enumerate}

Of these steps, the model probability distribution in Step~\ref{item:procedure1_DistributionEstiomation} is constructed by the 
learning algorithm presented in Sec.~\ref{sec:LearningAlgorithm}. 
To designing the AS in Step~\ref{item:procedure2_DefineAnomalyScore}, we 
refer to AnoGAN~\cite{schlegl2017unsupervised} in classical machine learning. 
Namely, we define a loss function $\mathcal{L}(U(\bm{z},\bm{\theta})\ket{0}^{\otimes n},\ket{\psi^{(t)})}$ 
for a test data $\ket{\psi^{(t)}}$ and the generative model  
$U(\bm{z},\bar{\bm{\theta}})$ constructed from the training dataset with 
learned parameter $\bar{\bm{\theta}}$; then take the minimum with respect 
to the latent variables $\bm{z}$ to calculate AS: 
\begin{equation}
	(\textrm{Anomaly Score})
	=\min_{\bm{z}} \mathcal{L}
	   (U(\bm{z},\bar{\bm{\theta}})\ket{0}^{\otimes n},\ket{\psi^{(t)}}).
\end{equation}
As the loss function $\mathcal{L}$, we use the local ground cost 
$c_{\rm local}(\ket{\psi^{(t)}},U(\bm{z},\bar{\bm{\theta}})\ket{0}^{\otimes n})$ 
defined in Eq.~\eqref{eq:localcost}. 
The above minimization is executed via the gradient descent with respect 
to $\bm{z}$, which is obtained via the parameter shift rule similar to the 
derivative in $\theta$. 
Algorithm~\ref{alg:anomalydetection} summarizes the procedure. 
\begin{algorithm}
	\caption{Algorithm to calculate Anomaly Score}
	\label{alg:anomalydetection}
	\begin{algorithmic}[1]
		\renewcommand{\algorithmicrequire}{\textbf{Input:}}
		\renewcommand{\algorithmicensure}{\textbf{Output:}}
		\REQUIRE A trained quantum circuit $U(\bm{z},\bm{\bar{\theta}})$, 
		test data $\{\ket{\psi^{(t)}}\}$
		\ENSURE  Anomaly Score
		\STATE Initialize $\bm{z}$
		\REPEAT
		\STATE Calculate the ground cost $\mathcal{L}$ from $\{U(\bm{z},\bm{\bar{\theta}})\}$ and $\ket{\psi^{(t)}}$ according to  Eq.~\eqref{eq:localcost}
		\STATE Calculate the gradients 
		$\left\{\frac{\partial \mathcal{L}}{\partial z_k}\right\}_{k=1}^{N_z}$ 
		using the parameter shift rule.
		\STATE Update $\bm{z}$ by using the gradient descent method 
		via $\left\{\frac{\partial \mathcal{L}}{\partial z_k}\right\}_{k=1}^{N_z}$
		\UNTIL{convergence}
	\end{algorithmic}
\end{algorithm}

\subsection{Distributed dataset}
\label{sec:distributed dataset}

The first demonstration is to construct a generative model that learns 
a quantum ensemble distributed on the equator of the generalized Bloch 
sphere. 
That is, the training ensemble (i.e., the normal dataset) 
$\{\ket{\psi_j}\}_{j=1}^{M}$ to be learned is set as follows: 
\begin{align}
    \ket{\psi_j} = \mathrm{cos}(\pi/4)\ket{0} +e^{2\pi i \phi^j}\mathrm{sin}(\pi/4)\ket{2^n -1},
    \label{eq:ad_trainingdata_sekido}
\end{align}
where $\phi^j$ is randomly generated from the uniform distribution on 
$[0,1]$ and $\ket{x}$ denotes the $x$-th basis in the $2^n$-dimensional 
Hilbert space. 
Note that the configuration of this ensemble cannot be learned by the 
existing mixed-state-based quantum anomaly detection scheme 
\cite{hara2014anomaly,hara2016quantum}, because the mixed state 
corresponding to this ensemble is nearly the maximally mixed state, the 
learning of which thus does not give us a generative model recovering 
the original ensemble.

We employ the same ansatz as that given in Sec.~\ref{sec:PerformanceAnalysis} 
with the parameters shown in Table~\ref{tab:ExperimentalParamters_ad} and 
construct the generative model according to Algorithm~\ref{alg:quantumOTLossLearning}. 
As the optimizer, we take Adam \cite{Kingma2014-pa} with learning rate 0.01. 
The number of learning iterations (i.e., the number of the updates of the 
parameters) is set to 1500 for $n=2$ and 10000 for $n=10$.

Once the model for normal ensemble is constructed, it is then used to 
anomaly detection. 
Here the set of test data $\{\ket{\psi^{(t)}}\}$ is given by 
\begin{align}
    \ket{\psi^{(t)}} 
      = \mathrm{cos}(\frac{\pi}{2}\theta^{(t)})\ket{0} 
      +e^{2\pi i \phi^{(t)}}\mathrm{sin}(\frac{\pi}{2}\theta^{(t)})\ket{2^n -1},
    \label{eq:ad_TestData}
\end{align}
where $\theta^{(t)}, \phi^{(t)} \in \{0, 0.1, 0.2, \ldots ,2\}$. 
We calculate the AS using Algorithm~\ref{alg:anomalydetection}. 
The other simulation parameters are shown in 
Table~\ref{tab:ExperimentalParamters_ad}.

\begin{table}[tb]
	\begin{center}
		\caption{List of parameters for numerical simulation in  Sec.~\ref{sec:distributed dataset}}
		\begin{tabular}{cll} \hline
			number of measurements~(training)          & $N_{s}$  & $100$~(for~$n=2$),~$\infty$~(for~$n=10$) \\
			number of measurements~(anomaly detection) & $N_{ad}$ & $50$~(for~$n=2$),~$100$~(for~$n=10$) \\
			number of training data             & $M$                & $30$                \\
			number of qubits              & $n$                & $2,10$       \\
			dimension of latent space     & $N_z$              & $1$  \\
			number of layers              & $N_L$              & $10$            \\
			\hline
		\end{tabular}
		\label{tab:ExperimentalParamters_ad}
	\end{center}
\end{table}

The numerical simulation result in the case of $n=2, 10$ are presented in  
Figs.~\ref{fig:anomalyDetection_n2} and \ref{fig:anomalyDetection_n10}, 
respectively. 
The training ensemble is shown in the figure (a), where each blue point 
corresponds to the generalized Bloch vector.
Some output states of the constructed generative model, corresponding to 
different value of $z \in [0,1]$, are shown in the figure (b). 
Both of the red and blue points in Fig.~(c) represent the test data state 
\eqref{eq:ad_TestData}. 
The figure (d) shows the calculated AS, where the blue and red plots 
correspond to the blue and red points in (c), respectively. 
The dotted line in (d) illustrates the theoretical expected values assuming 
that the model completely learns the training data.

Firstly, we see the clear correlation between the AS and the theoretical curve 
in Fig.~\ref{fig:anomalyDetection_n2}, implying that AS is appropriately 
calculated via the proposed method. 
In the practical usecase, a user defines a threshold of AS depending on the 
task and then compare the calculated AS with the threshold for identifying the anomaly quantum states. 
For instance, if we set the threshold as $\mathrm{AS}=0.3$, the test states  
conditioned in $0.3 \le \theta^{(t)}/\pi \le 0.7$ in 
Fig.~\ref{fig:anomalyDetection_n2} are judged as normal while others are 
anomaly. 
Moreover, the output states of the learned generative model and the result 
of anomaly detection in the case of $n=10$ are shown in 
Fig.~\ref{fig:anomalyDetection_n10}. 
Although, the result displayed in (b) would suggest that the learning fails, 
the output states show correlation with the training states displayed in (a); 
actually, the output states live on the $xy$-plane in the generalized Bloch 
sphere spanned by $\ket{0}^{\otimes 10}$ and $\ket{1}^{\otimes 10}$. 
In addition, it is notable that only $N_s=100$ is enough even for the case 
of $n=10$ to perform anomaly detection, provided that we obtain an appropriate 
generative model from the training normal states. 
This is an advantage for practical situation, as this indicates that the 
proposed method may scale up withe respect to the number of qubits.

\begin{figure}[hbtp]
	\centering
	\subfigure[Training data]{\includegraphics[width=.23\linewidth]{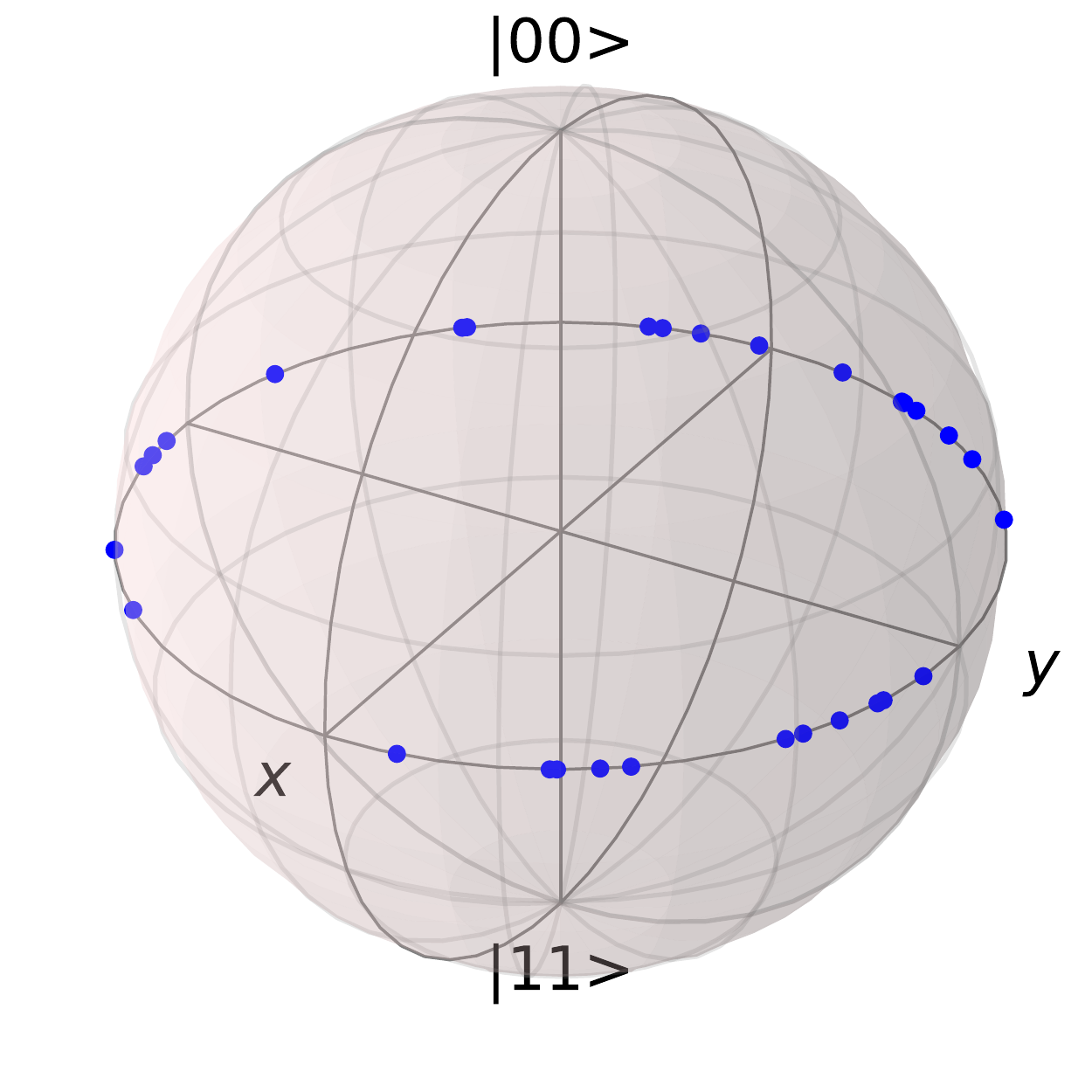}}
	\subfigure[Model output (after learn; $z \in \lbrack 0,1\rbrack $)]{\includegraphics[width=.23\linewidth]{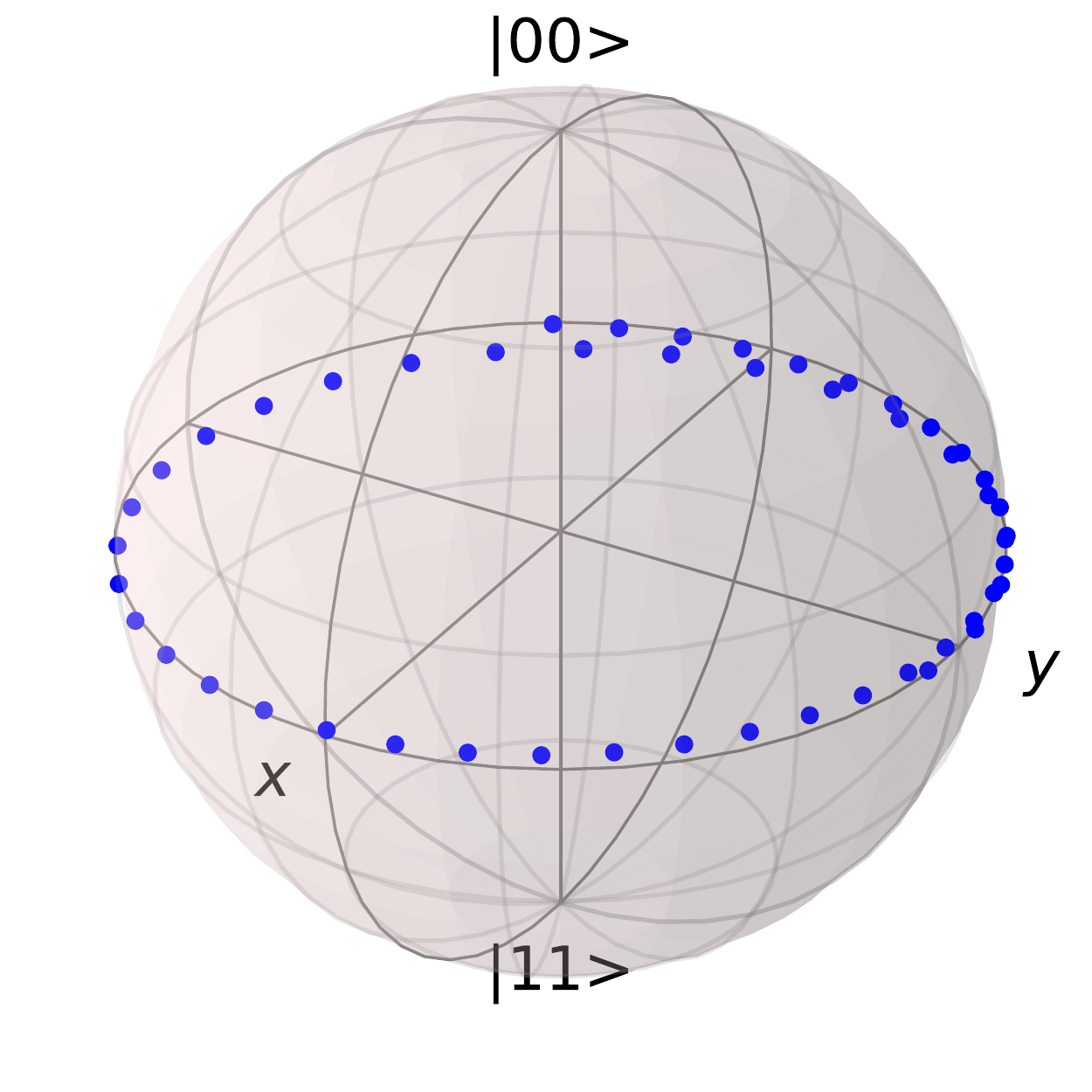}}
	\subfigure[Test data]{\includegraphics[width=.23\linewidth]{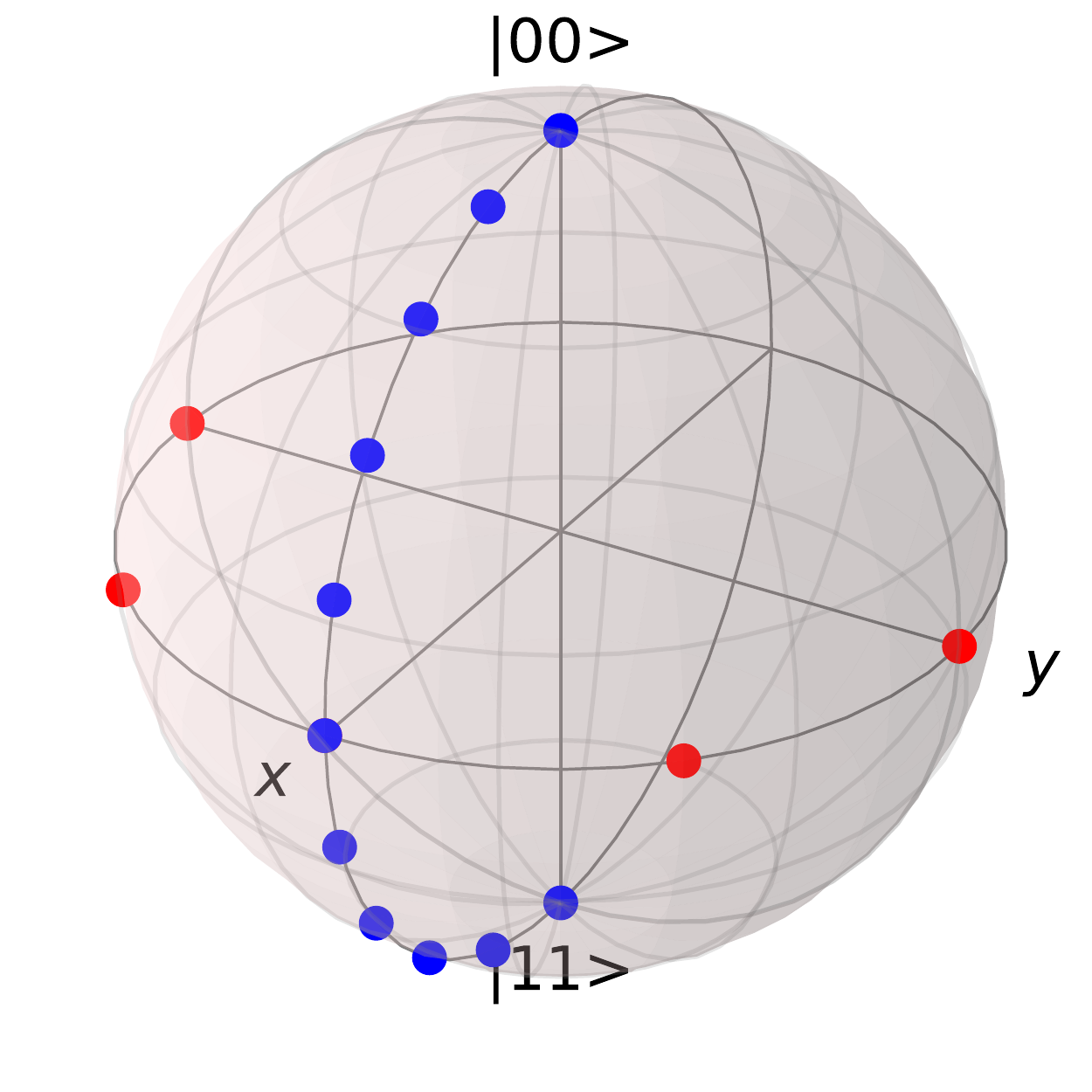}}
	\subfigure[Anomaly score]{\includegraphics[width=.27\linewidth]{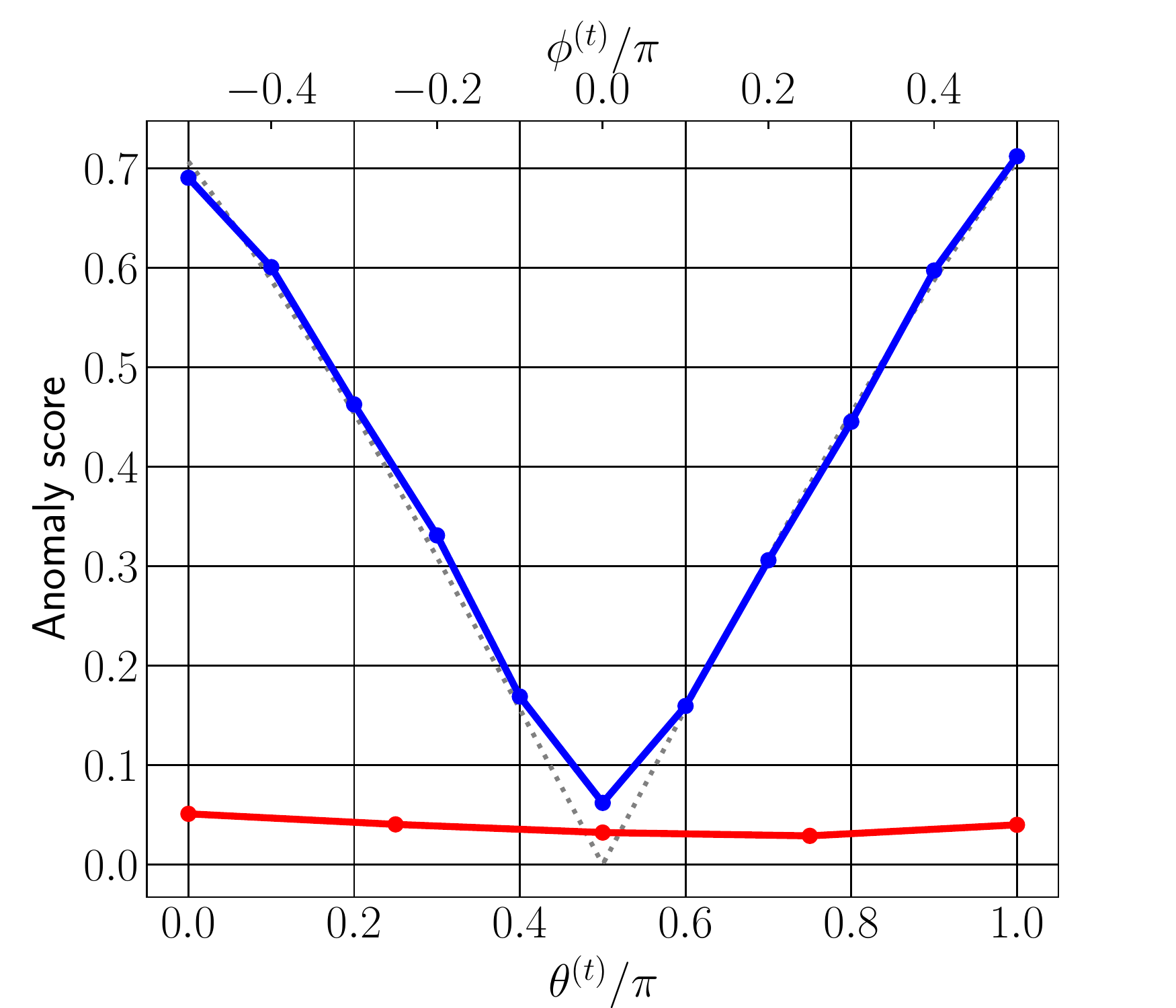}}
	\caption{
		The simulation results in the case of $n=2$, for the distributed 
		training ensemble. 
		(a) Generalized Bloch vector representation of training ensemble. 
		(b) Output states from the generative model with latent variables 
		$\bm{z} \in [0,1]$. 
		(c) Test data (both red and blue points). 
		(d) Anomaly scores for different test data. 
		The blue and red plots correspond to the blue and red points in (c), 
		respectively. 
		The values of angle $\theta^{(t)}$ and $\phi^{(t)}$ are given 
		in the bottom and upper horizontal axis, respectively. 
		The dotted line represents the theoretical value of AS.
	}\label{fig:anomalyDetection_n2}
\end{figure}

\begin{figure}[hbtp]
	\centering
	\subfigure[Training data]{\includegraphics[width=.23\linewidth]{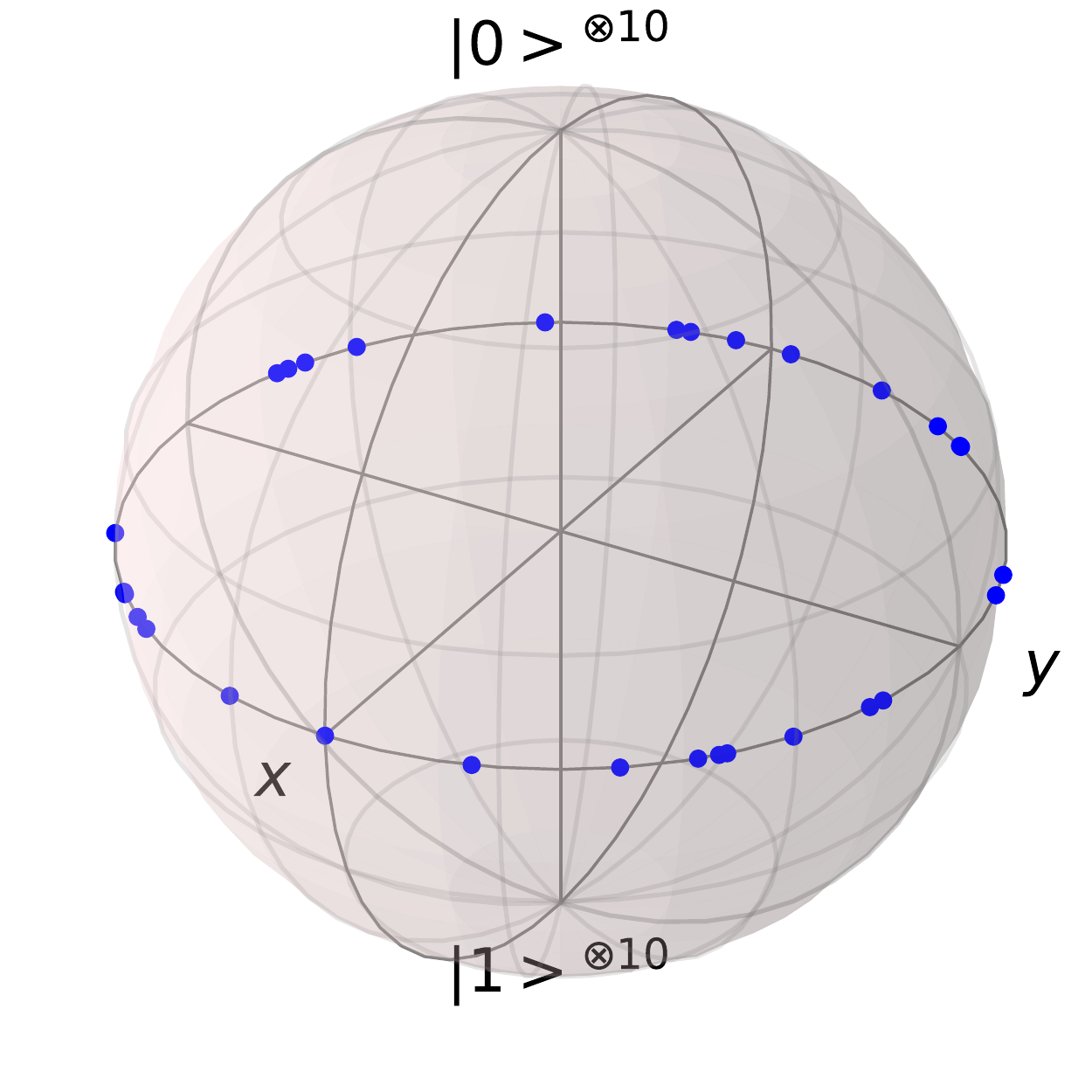}}
	\subfigure[Model output (after learn; $z \in \lbrack 0,1\rbrack $)]{\includegraphics[width=.23\linewidth]{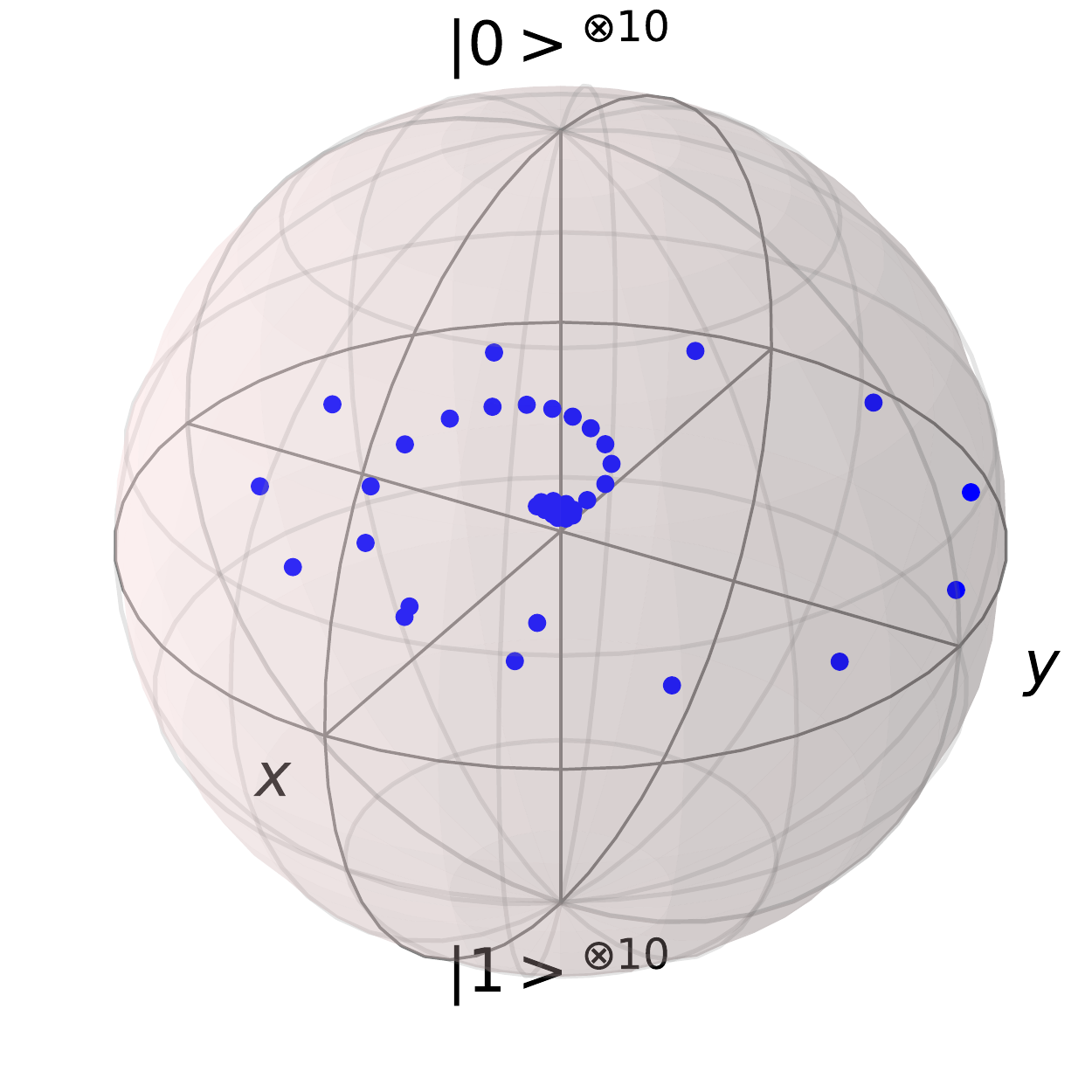}}
	\subfigure[Test data]{\includegraphics[width=.23\linewidth]{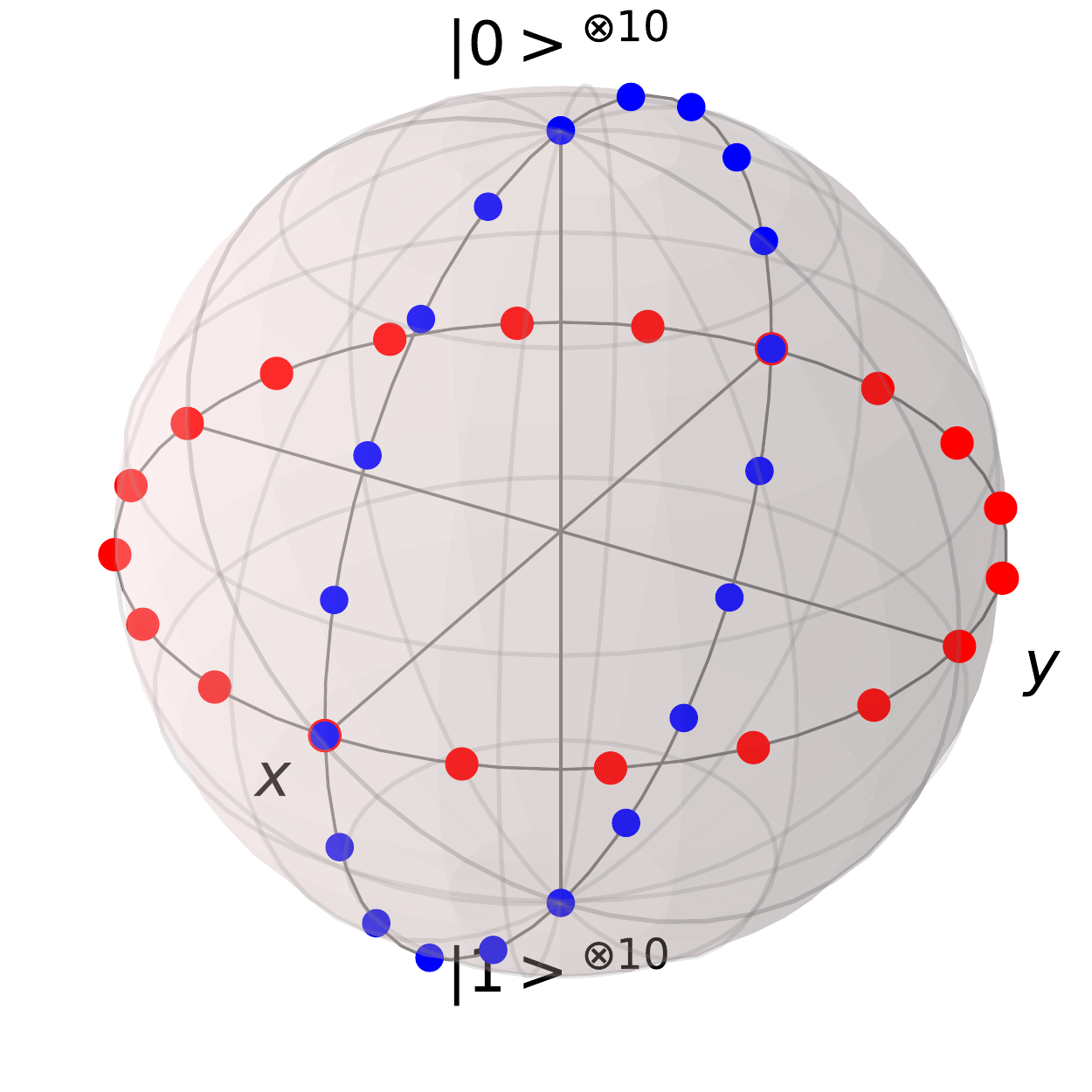}}
	\subfigure[Anomaly score]{\includegraphics[width=.27\linewidth]{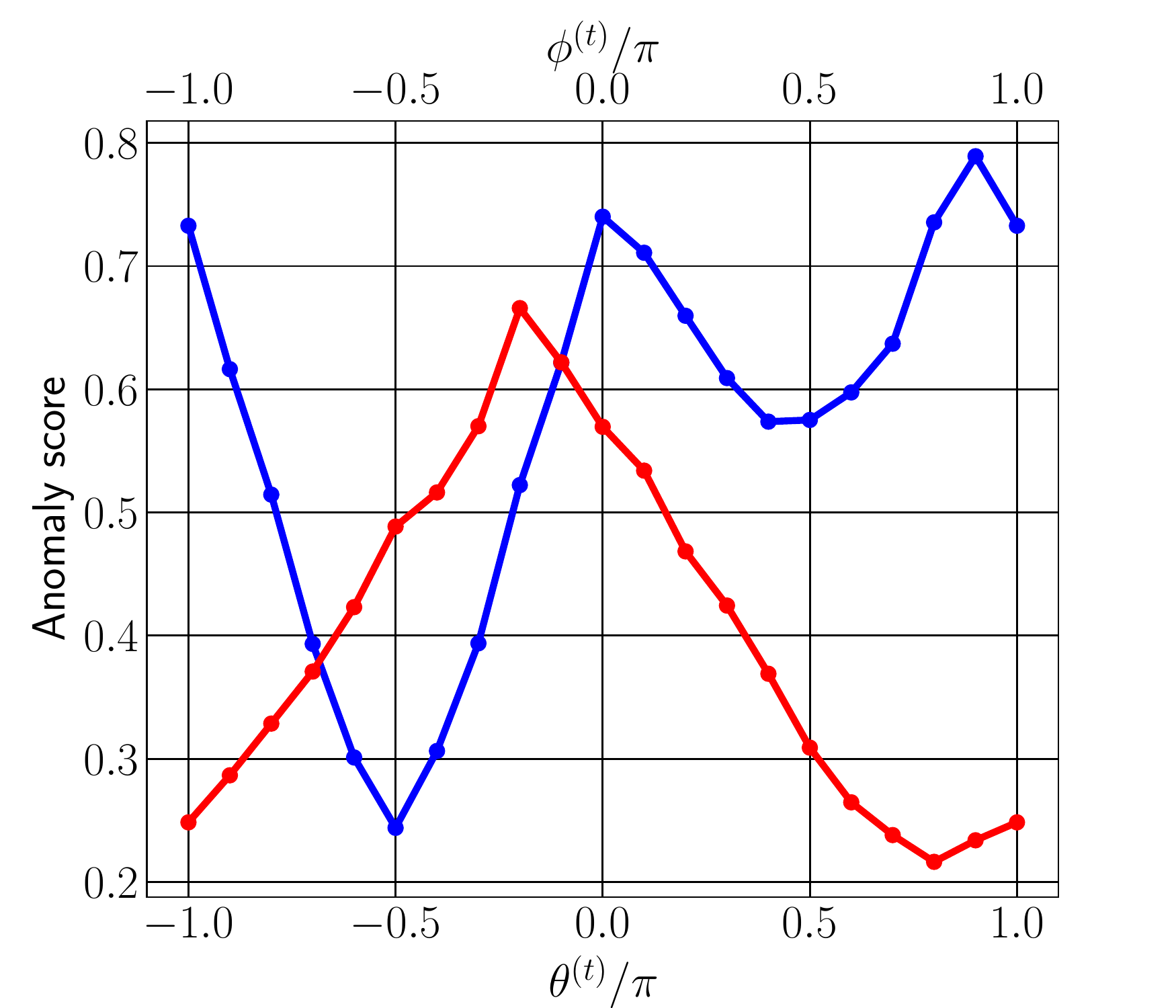}}
	\caption{
		The simulation results in the case of $n=10$, for the distributed 
		training ensemble. 
		(a) Generalized Bloch vector representation of training states. 
		(b) Output states from the generative model with latent variables 
		$\bm{z} \in [0,1]$. 
		(c) Test data (both red and blue points). 
		(d) Anomaly scores for different test data. 
		The blue and red plots correspond to the blue and red points in (c), 
		respectively. 
		The values of angle $\theta^{(t)}$ and $\phi^{(t)}$ are given 
		in the bottom and upper horizontal axis, respectively. 
	}\label{fig:anomalyDetection_n10}
\end{figure}

\subsection{Localized dataset}
\label{sec:appresult}

Next let us consider localized quantum ensemble. 
That is, the state of training ensemble $\{\ket{\psi_j}\}_{j=1}^{M}$ 
corresponding to the normal dataset is given by 
\begin{align}
    \ket{\psi_j} = \mathrm{cos}\left(\frac{\pi}{2}\Delta\theta_j\right)\ket{0} +e^{2\pi i \Delta\phi_j}\mathrm{sin}\left(\frac{\pi}{2}\Delta\theta_j\right)\ket{2^n -1},
    \label{eq:ad_trainingdata}
\end{align}
where $n$ is the number of qubits. 
$\Delta\theta_j$ and $\Delta\phi_j$ are sampled from the normal distribution 
$N(\mu, \sigma)$ and the uniform distribution $U(a, b)$, respectively 
($\mu$ and $\sigma$ represent the mean and the variance, respectively). 
We will consider the two cases $(\mu, \sigma, a, b) = (0, 0.02, 0, 0.1)$ for 
$n=6$ and $(\mu, \sigma, a, b) = (0, 0.02, 0, 0.2)$ for $n=10$. 
Note that in this choice of parameters, the ensemble $\{\ket{\psi_j}\}_{j=1}^{M}$ 
is nearly two-dimensionally distributed on the generalized Bloch sphere, 
as illustrated in Fig.~\ref{fig:AnomalyDetection_appendix} (a, d). 
The other simulation parameters are shown in 
Table~\ref{tab:ExperimentalParamters_ad_app}.

\begin{table}[tb]
	\begin{center}
		\caption{List of parameters for numerical simulation of Sec.~\ref{sec:appresult}}
		\begin{tabular}{cll} \hline
			number of measurements~(training)          & $N_{s}$  & 1000 \\
			number of measurements~(anomaly detection) & $N_{ad}$ & 50 \\
			number of training data             & $M$                & $10$                \\
			number of qubits              & $n$                & $6,10$       \\
			dimension of latent space     & $N_z$              & $2$ \\
			number of layers              & $N_L$              & $10$            \\
			\hline
		\end{tabular}
		\label{tab:ExperimentalParamters_ad_app}
	\end{center}
\end{table}

To construct a generative model via learning this two dimensional distribution, 
we set the dimension of latent variable as $N_z=2$ from the above-mentioned 
observation on the dimensionality of $\{\ket{\psi_j}\}_{j=1}^{M}$. 
Also, we here take the so-called alternating layered ansataz~(ALA), which, 
together with the use of local cost, is guaranteed to mitigate the gradient 
vanishing issue \cite{cerezo2021cost, nakaji2021expressibility}. 
This ansatz is more favorable than the previous one which we here call the 
hardware efficient ansatz~(HEA), in view of the possibility to avoid the 
gradient vanishing issue. 
Therefore it is worth comparing their learning curves. 
Typical learning curves are shown in Fig.~\ref{fig:learningCurve}.
The blue plots, which are labeled "local", represent the case where the 
cost is the local one \eqref{eq:localcost}) and the ansatz is ALA; thus 
we denote this case as L-ALA. 
On the other hand, the orange plots, which are labeled "global", represent 
the case where the cost is the global one \eqref{eq:tracedistance}) and 
the ansatz is HEA; thus we denote this case as G-HEA. 
Not that both displayed costs are calculated as the global one, to directly 
compare them; 
that is, the "local" represents the cost calculated at each iteration based 
on the global cost with the parameters that optimize the local cost. 
We observe that L-ALA has a clear advantage over G-HEA in terms of the 
convergence speed. 
This result coincides with that of Sec.~\ref{sec:barrenplateau}, indicating 
the advantage of the local cost. 
In addition to the convergence speed, the final cost of L-ALA is lower than 
that of G-HEA. 
Note that the learning performance heavily depends on the initial random seed, 
yet it was indeed difficult to find a successful setting of G-HEA; actually 
in all cases we tried the trajectory seemed to be trapped in a local minima, 
presumably because the variance of G-HEA is much smaller than that of L-ALA.

\begin{figure}[hbtp]
	\centering
	\subfigure[Learning curve ($n=6$)]{\includegraphics[width=.4\linewidth]{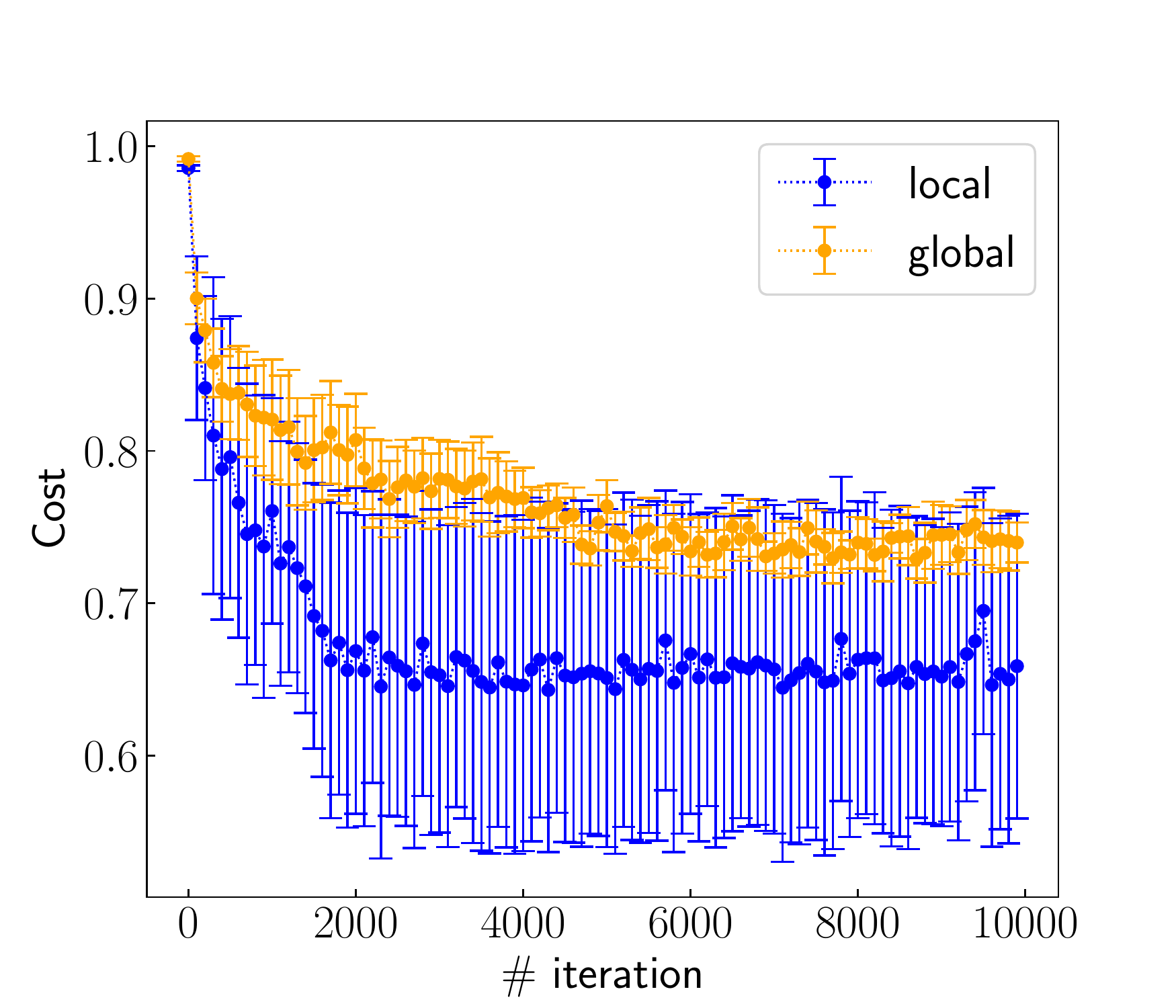}}
    \subfigure[Learning curve ($n=10$)]{\includegraphics[width=.4\linewidth]{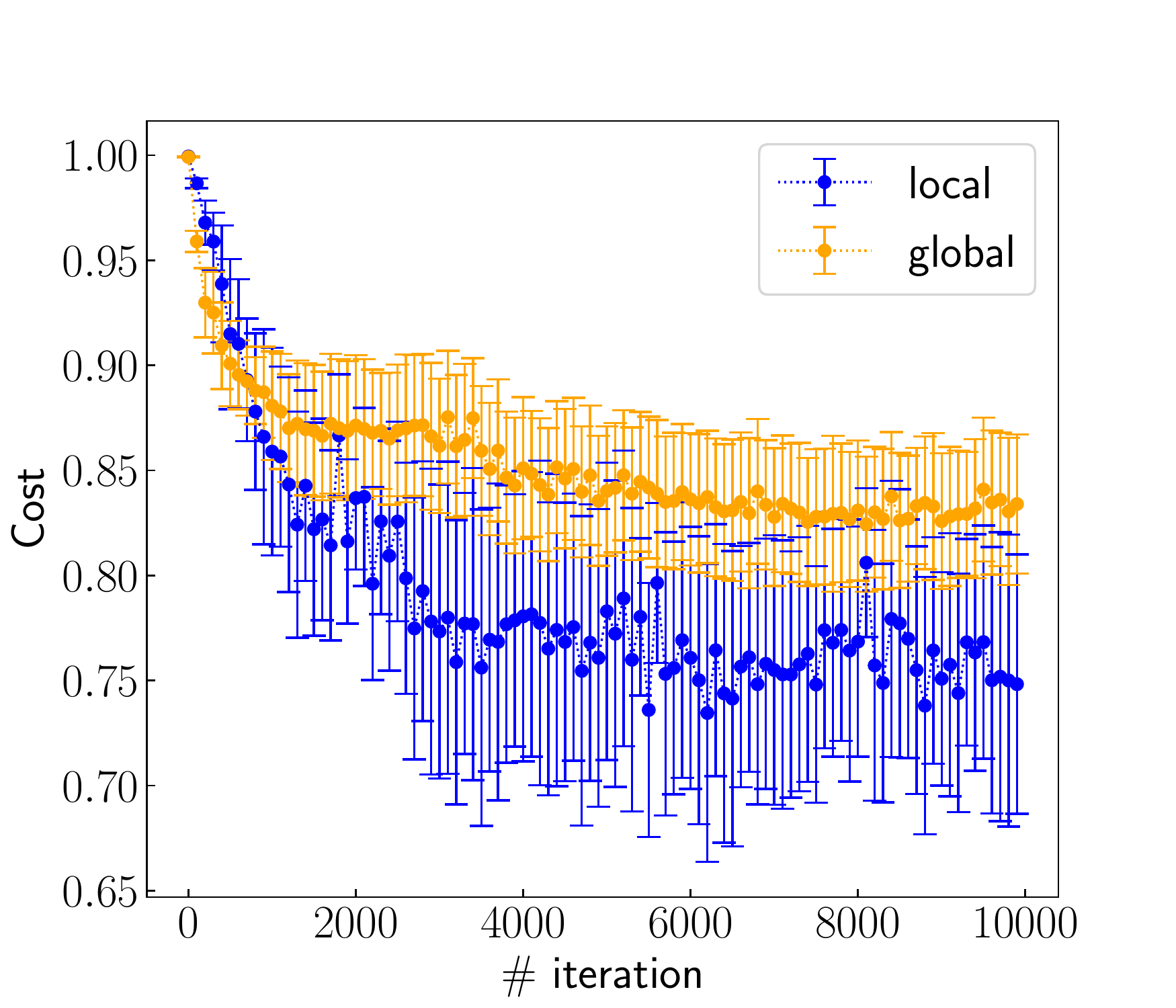}}
	\caption{
		 Learning curves for the case (a) $n=6$ and (b) $n=10$. 
		 The values of OTL is calculated with the parameters at each 
		 iteration trained with local cost \eqref{eq:localcost} or 
		 global cost \eqref{eq:tracedistance}. 
		 Note that the displayed are the cost calculated with the 
		 global one \eqref{eq:tracedistance} for both blue and orange 
		 cases, for fair comparison. 
		 The range of error bar is 1 standard deviation. 
	}\label{fig:learningCurve}
\end{figure}

We apply the constructed generative model to the anomaly detection problem. 
In Fig.~\ref{fig:AnomalyDetection_appendix} (b) and (e), the test quantum 
states are displayed for the case of $n=6$ and $n=10$, respectively. 
The resultant anomaly score for each test data are shown in 
Fig.~\ref{fig:AnomalyDetection_appendix} (c, f). 
In both cases, we can say that the models are trained appropriately. 
In particular, the variance of the distribution of $\{\ket{\psi_j}\}_{j=1}^{M}$, 
i.e., the distribution of the blue points in (a, d), is well captured by the 
width of the dip of red lines in (c, f). 
Finally note that, in this section, we use QASM simulator for the numerical 
simulation; the number of shot is 1000 for each measurement in the learning 
process, and 50 for the anomaly detection task, even for the case of $n=10$. 
Compared to the state tomography, these numbers of shot are clearly too small. 
Nonetheless, the proposed method enabled the model to learn the training 
ensemble with such small number of shots.

\begin{figure}[hbtp]
	\centering
	\subfigure[Training data]{\includegraphics[width=.3\linewidth]{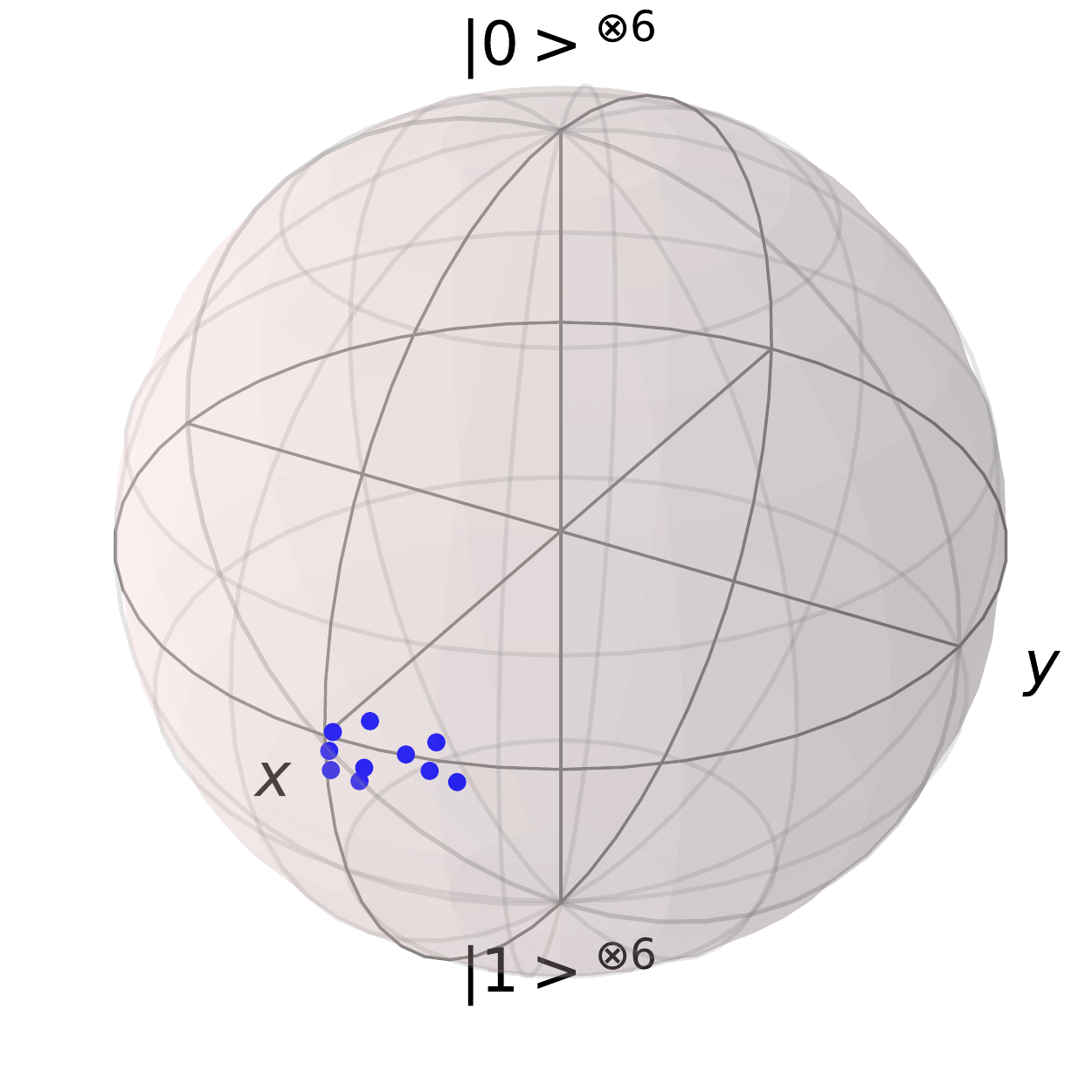}}
	\subfigure[Test data]{\includegraphics[width=.3\linewidth]{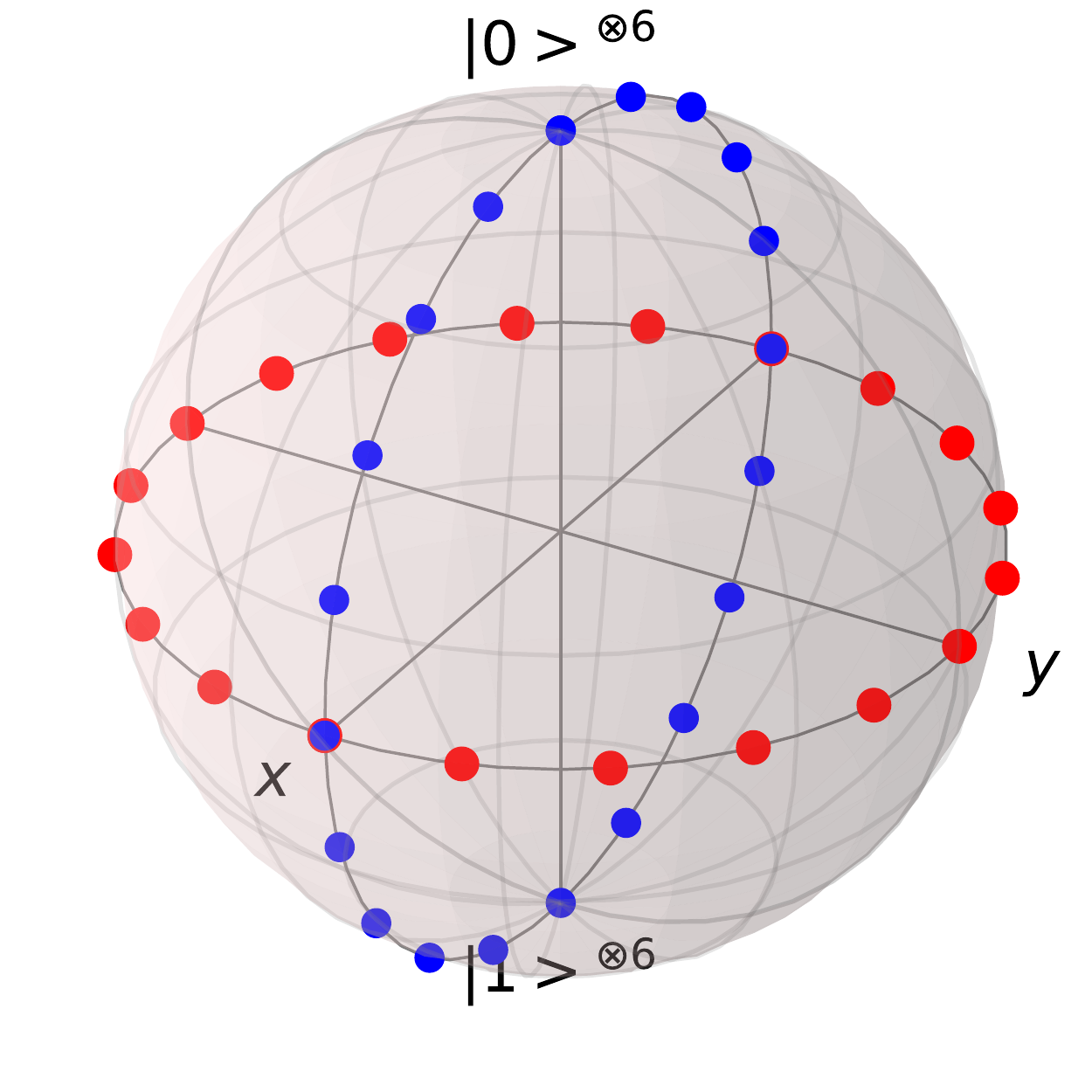}}
	\subfigure[Anomaly score]{\includegraphics[width=.35\linewidth]{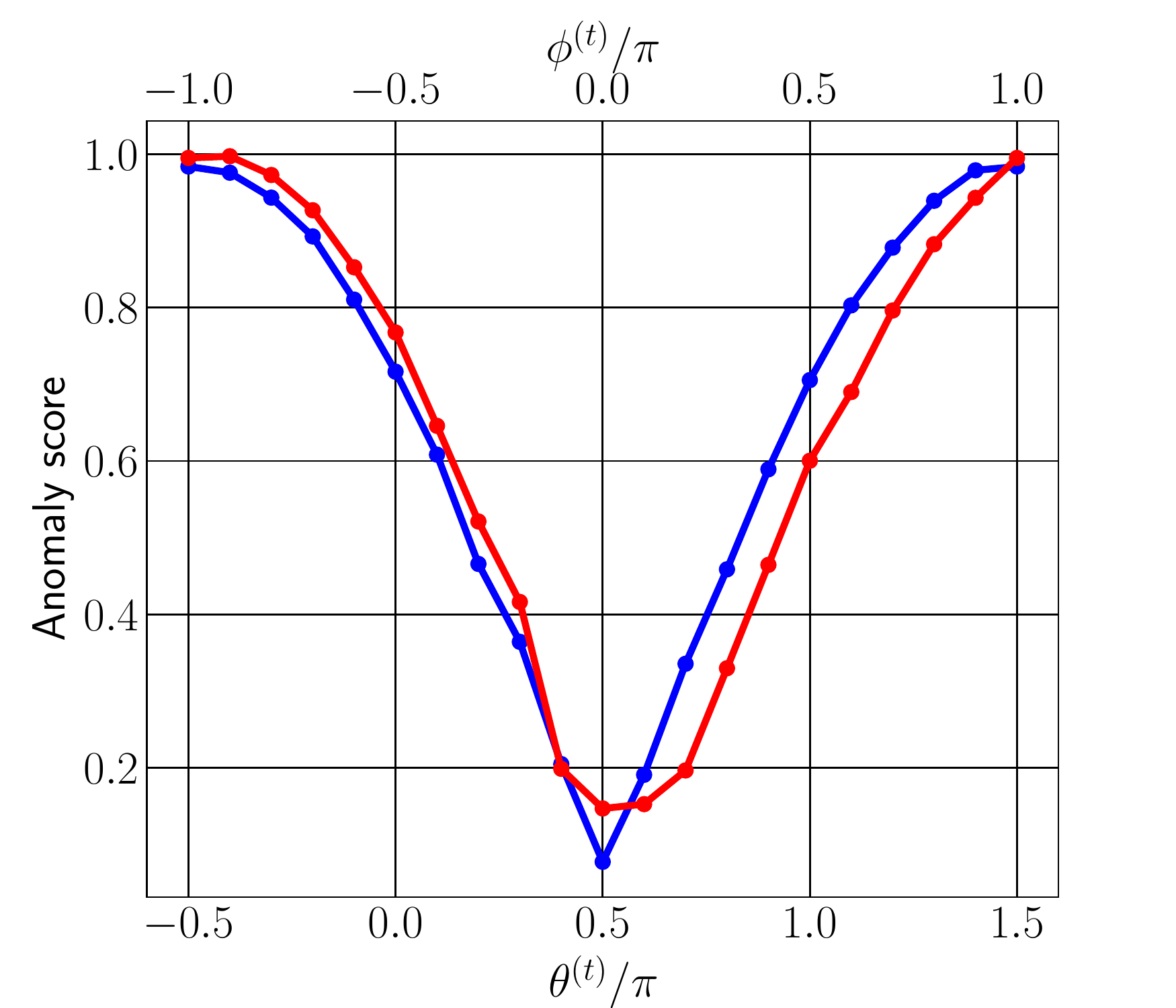}}
	\subfigure[Training data]{\includegraphics[width=.3\linewidth]{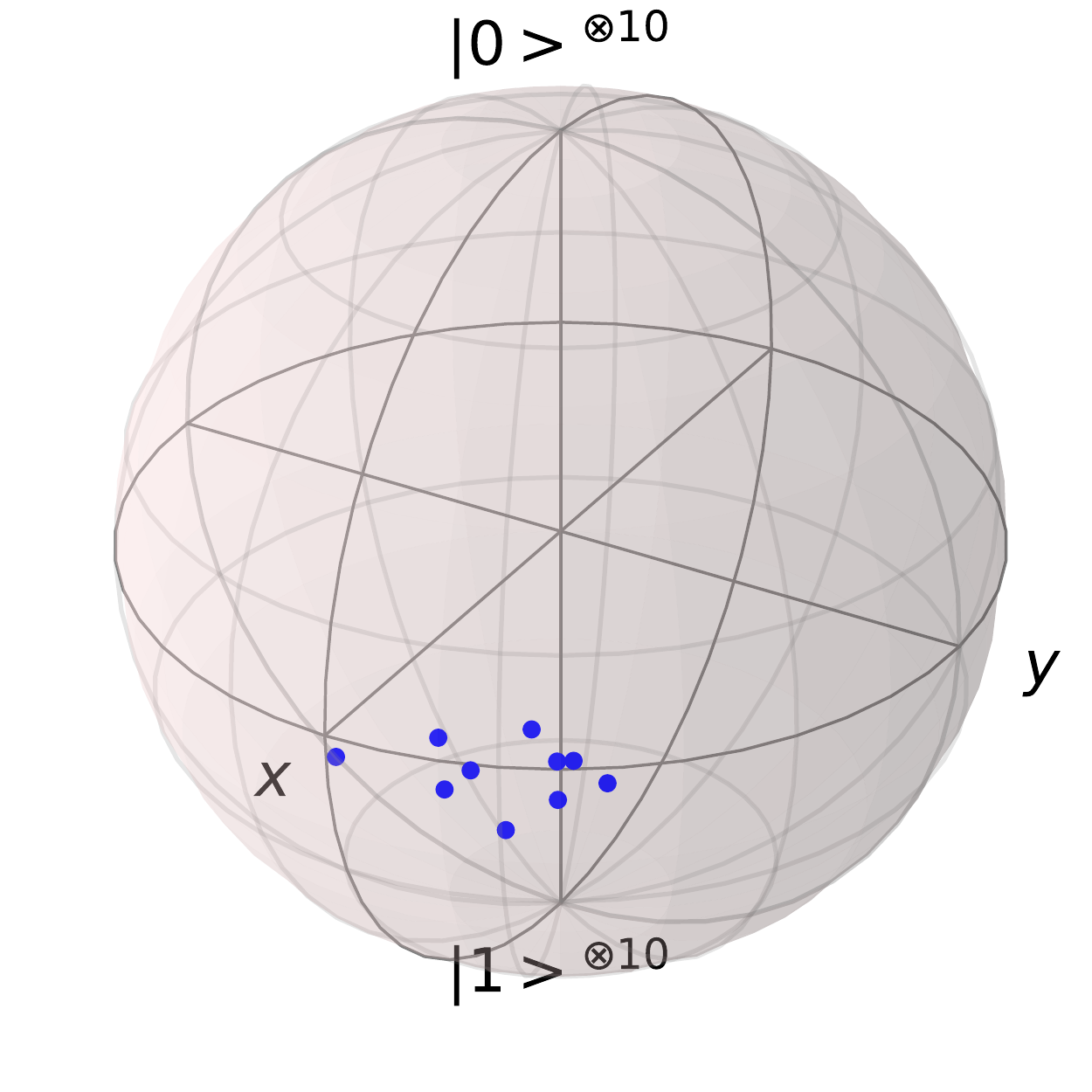}}
	\subfigure[Test data]{\includegraphics[width=.3\linewidth]{fig/ad_n10_test_data.pdf}}
	\subfigure[Anomaly score]{\includegraphics[width=.35\linewidth]{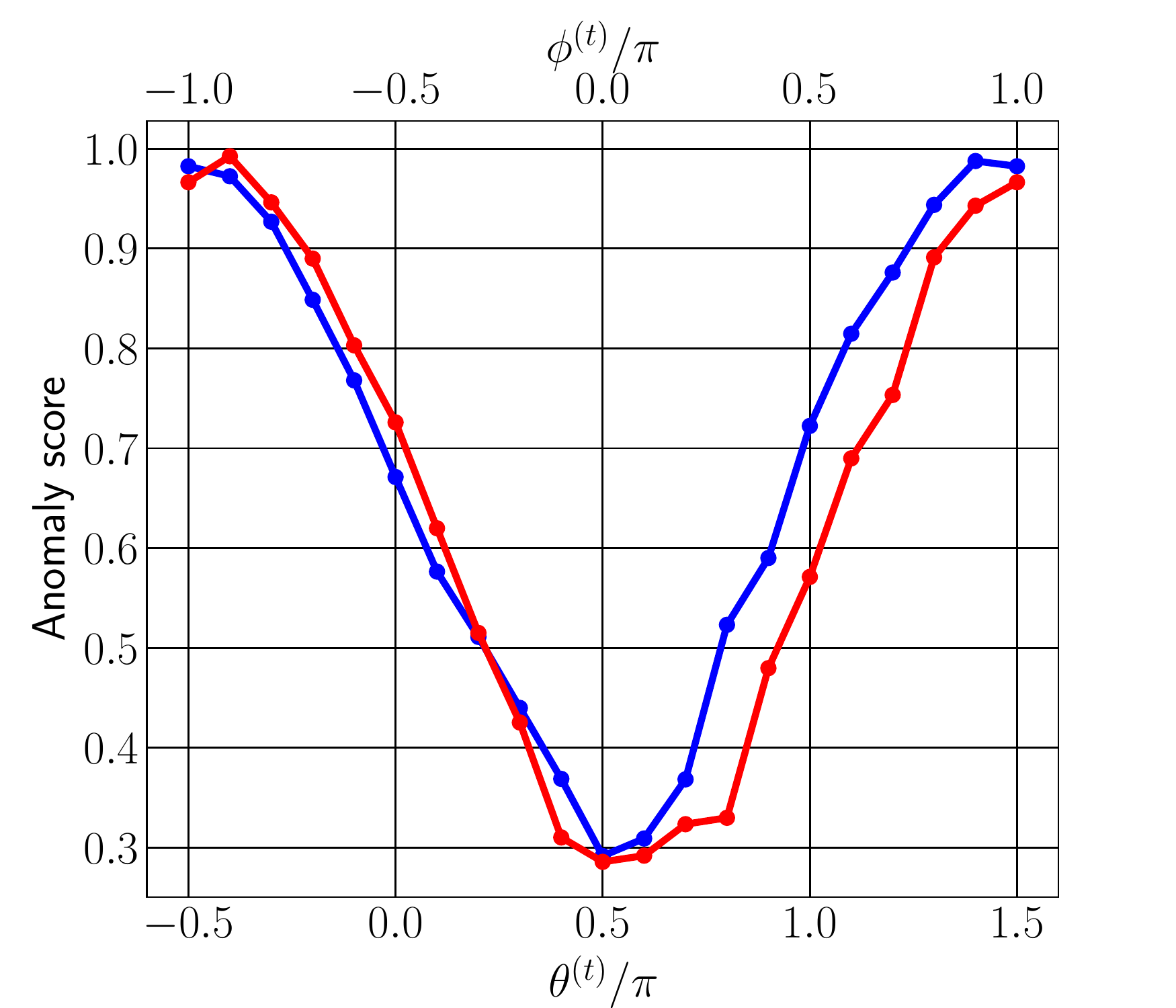}}
	\caption{
		The simulation results for the localized training ensemble, 
		in the case of (a,b,c) $n=6$ and (d,e,f) $n=10$. 
		(a, d) Generalized Bloch vector representation of training states. 
		(b, e) Test data (both red and blue points). 
		(c, f) Anomaly scores for different test data. 
		The blue and red plots correspond to the blue and red points in (b, e), 
		respectively. 
		The values of angle $\theta^{(t)}$ and $\phi^{(t)}$ are given 
		in the bottom and upper horizontal axis, respectively. 
		The dotted line represents the theoretical value of AS.}
		\label{fig:AnomalyDetection_appendix}
\end{figure}

\section{Conclusion}
\label{sec:conclusion}

In classical machine learning, many generative models are vigorously studied, but there are only a few studies on quantum generative models for quantum data.
This paper offers a new approach for building such a quantum generative model, i.e., learning method for quantum ensemble with unsupervised machine learning setting.
For that purpose, we proposed a loss function based on optimal transport loss (OTL), which can be effectively used even when both the target and model probability distributions are hard to characterize. 
We also adopted the previously proposed local cost as the ground cost of OTL, to avoid the vanishing gradient problem and thereby increase learnability, but this makes OTL being no longer a distance.
Hence we have shown that OTL with the local cost satisfies the properties of divergence between probability distributions and confirmed that the proposed OTL is suitable as a cost for generative model. 
We then theoretically and numerically analyzed the properties of the proposed OTL.
Our analysis indicates that the proposed OTL is a good cost function when the quantum data ensemble has a certain structure, i.e., it is confined in a relatively low dimensional manifold.
In addition, We numerically showed that OTL can avoid the vanishing gradient issue thanks to the locality of the cost. 
Finally, we demonstrated the anomaly detection problems of quantum states, that cannot be handled via existing methods, to show a validity of our method.

There is still a lot of works to be done in the direction of this paper.
First, this paper assumed that a quantum processing is performed individually for each quantum state, but it would be interesting to consider another setting such as allowing coherent access~\cite{aharonov2022quantum} or quantum memory~\cite{huang2021quantum}. 
In this scenario, the anomaly detection technique could be used for processing quantum states which come from an external quantum processor, such as quantum sensor.
Also, it would be interesting to extend the cost to the entropy regularized Wasserstein distance~\cite{cuturi2013sinkhorn,feydy2018interpolating,genevay2019sample,amari2016geometry,amari2018information}, which is effective for dealing with higher dimensional data in classical generative models.
In this case, however, the localized cost proposed in this paper does not satisfy the axiom of distance likewise the case demonstrated in this paper; this is yet surely an interesting future work. 
Lastly, note that classical data can also be within the scope of our method, provided it can be effectively embedded into a quantum data; then the anomaly detection of financial data, which usually requires high dimensionality for precise detection, might be a suitable target.

\section*{Acknowledgement}
This work was supported by MEXT Quantum Leap Flagship Program Grant Number JPMXS0118067285 and JPMXS0120319794

\bibliographystyle{unsrt_mod}
\bibliography{main}

\afterpage{\clearpage}

\appendix
\section*{Appendices}
\addcontentsline{toc}{section}{Appendices}
\renewcommand{\thesubsection}{\Alph{subsection}}
\section{Proof of Proposition \ref{pro:optimaltransport_divergence}}\label{sec:proof_optimaltransport_divergence}
Here, we give the proof of Proposition \ref{pro:optimaltransport_divergence}.
\begin{proof}
	First, positivity $\mathcal{L}_c(\alpha,\beta) \ge 0$ is obvious from $c(\bm{x},\bm{y})\ge 0$ and $\pi(\bm{x},\bm{y})\ge 0$.

	Next, we move on the proof of $\alpha=\beta\Rightarrow \mathcal{L}_c(\alpha,\beta)=0$.
	When $\alpha=\beta$, then $\pi(\bm{x},\bm{y}) = \alpha(\bm{x})\delta(\bm{x}-\bm{y})$ is one of the candidates that satisfy the constraint Eq.~\eqref{eq:optimaltransportloss}, which gives the following inequality:
	\begin{equation}
		\mathcal{L}_c(\alpha,\beta)  \le \int c(\bm{x},\bm{y}) \alpha(\bm{x}) \delta(\bm{x}-\bm{y})d\bm{x}d\bm{y} = \int c(\bm{x},\bm{x})\alpha(\bm{x}) d\bm{x} = 0.
	\end{equation}
	Thus we obtain $\mathcal{L}_c(\alpha,\beta)=0$.

	Finally, we turn to $ \mathcal{L}_c(\alpha,\beta)=0\Rightarrow \alpha=\beta$.
	When $\bm{x}\neq \bm{y}$, the ground distance $c(\bm{x},\bm{y})$ is always greater than 0, and if the transport plan $\pi(\bm{x},\bm{y})$ takes a nonzero value , then $\mathcal{L}_c(\alpha,\beta)$ is always greater than 0.
	This indicates that in order for $\mathcal{L}_c(\alpha,\beta)=0$, the transport plan $\pi(\bm{x},\bm{y})$ should be a function that concentrated on $\bm{x}=\bm{y}$, i.e., the transport plans can be written as $\pi(\bm{x},\bm{y})=f(\bm{x})\delta(\bm{x}-\bm{y})$ with Dirac delta function $\delta(\bm{x}-\bm{y})$.
	Then, the constraint $\int \pi(\bm{x},\bm{y}) d\bm{x}=\beta(\bm{x})$ of Eq.~\eqref{eq:optimaltransportloss} leads to  $f(\bm{y})=\beta(\bm{y})$.
	Thus, we have $\pi(\bm{x},\bm{y})=\beta(\bm{x})\delta(\bm{x}-\bm{y})$, and we finally obtain $\alpha(\bm{x})=\beta(\bm{x})$ by taking into account the constraint $\int \pi(\bm{x},\bm{y}) d\bm{y}=\alpha(\bm{x})$.
\end{proof}

\section{Simulation results of Eq.(\ref{eq:fitting})}
\label{sec:fittingresult}
Here, we give the result of numerical simulations for the approximate error of the empirical loss Eq.~\eqref{eq:numericalExperiment_differentdistribution} and the result of fitting them with Eq.~\eqref{eq:fitting}.
The results are shown in Fig.\ref{fig:numericalExperiment_dif}.
The parameters of the simulation are given in Table \ref{tab:ExperimentalParamters_dimensionDependence}.
Fig.\ref{fig:result_fittingparameter} are created based on the result of Fig.\ref{fig:numericalExperiment_dif}.

\begin{figure}[tb]
	\centering
	\subfigure[$n=1$]{\includegraphics[width=.4\linewidth]{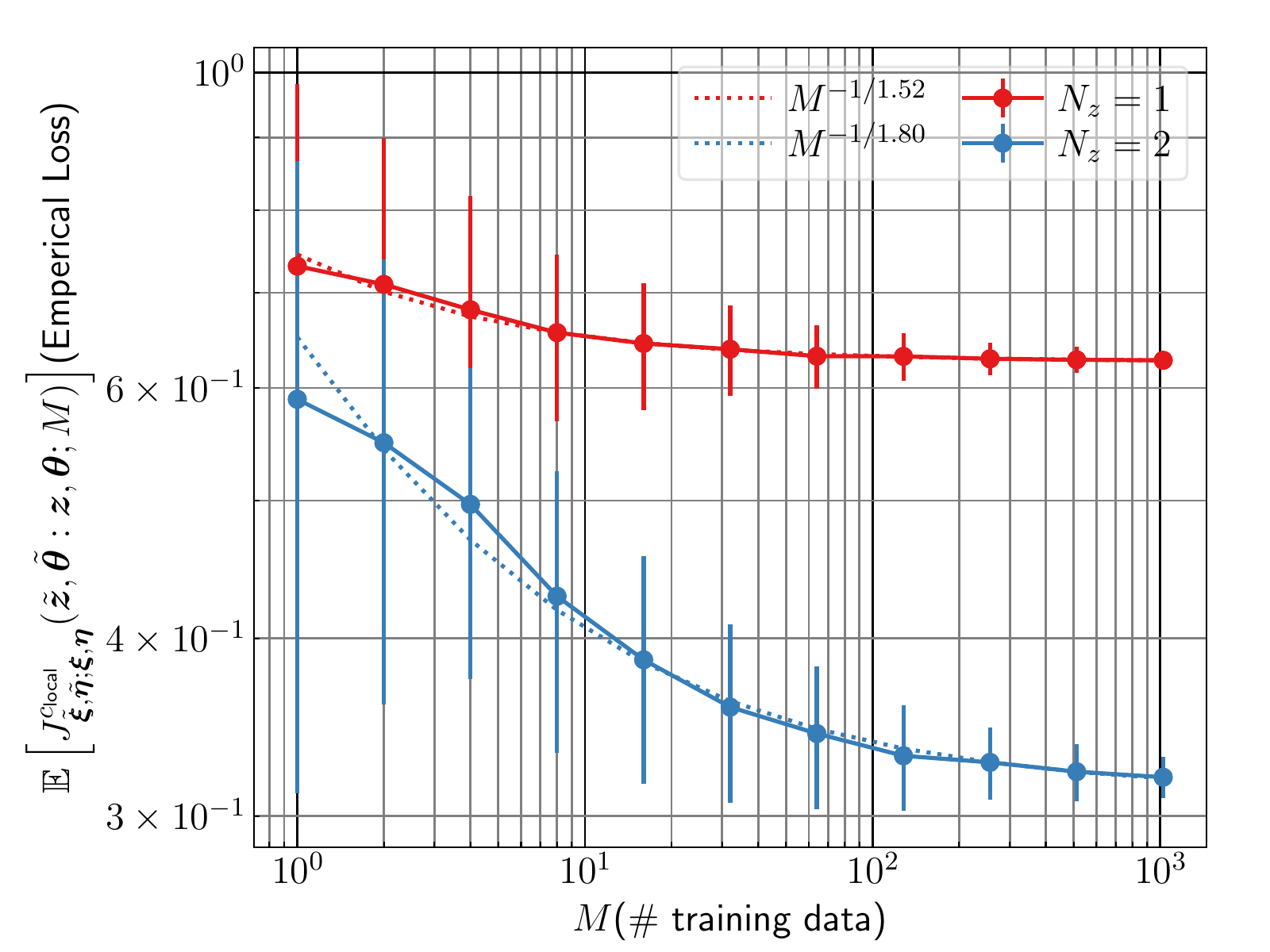}}
	\subfigure[$n=2$]{\includegraphics[width=.4\linewidth]{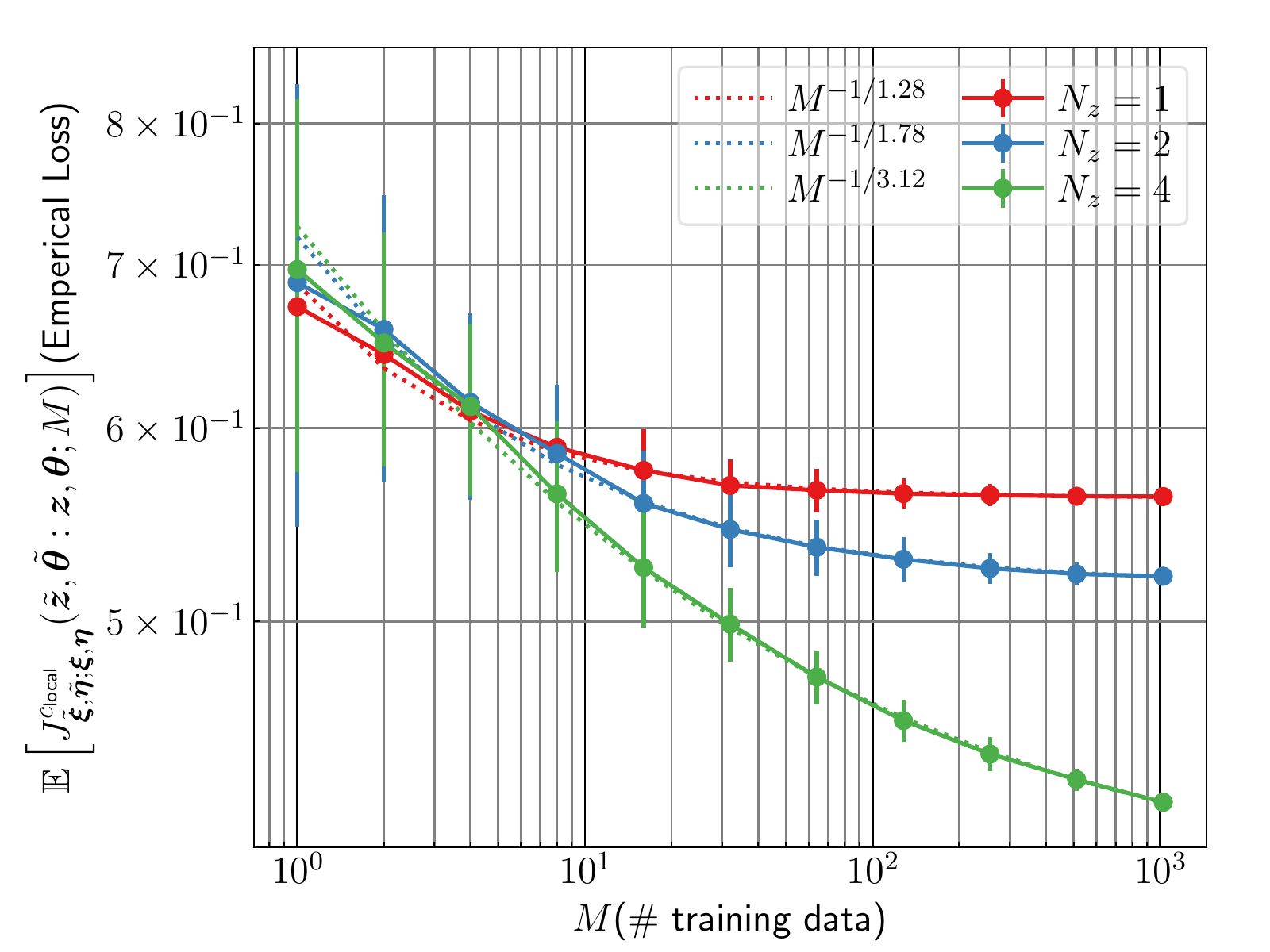}}
	\subfigure[$n=4$]{\includegraphics[width=.4\linewidth]{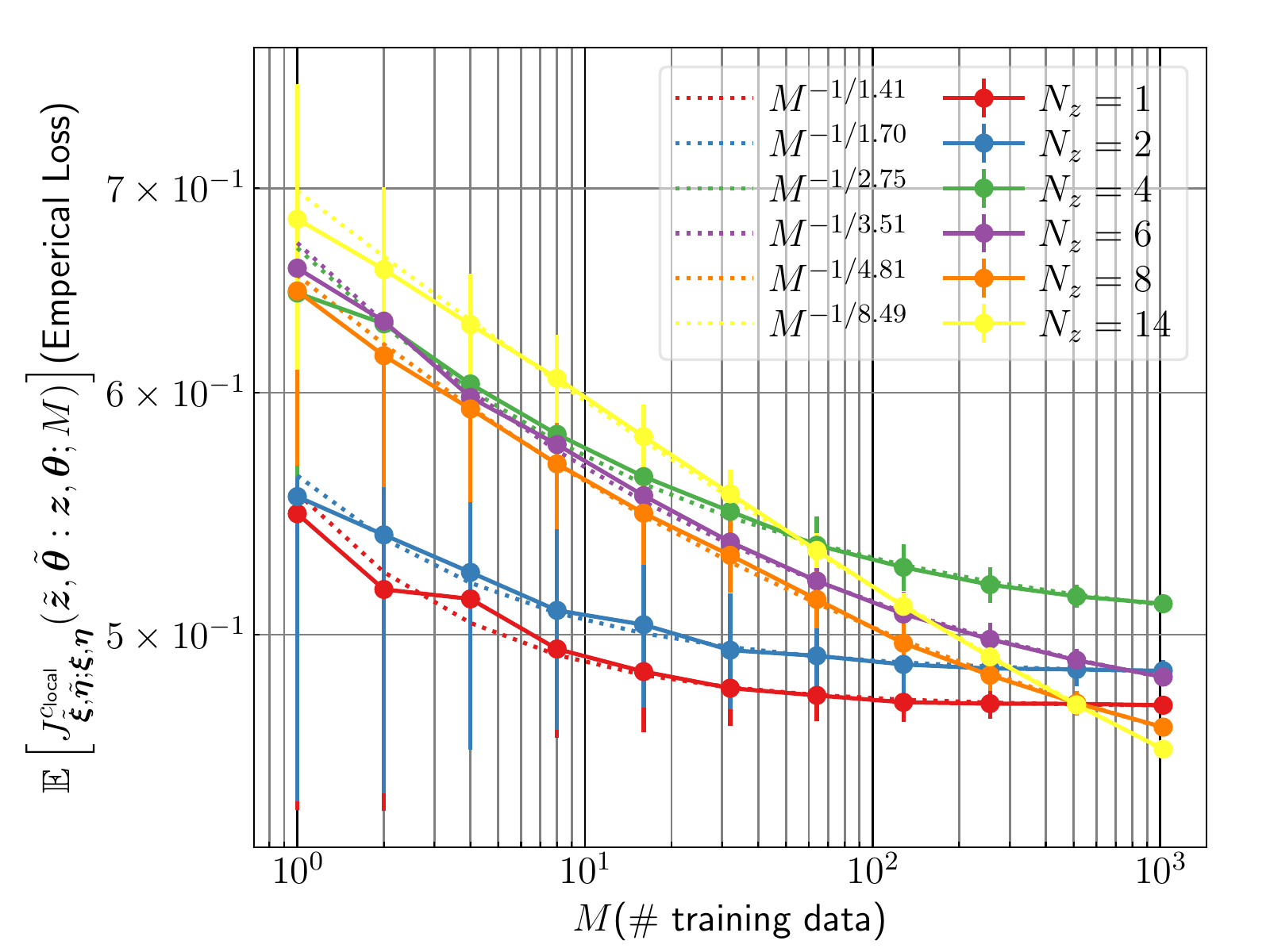}}
	\subfigure[$n=6$]{\includegraphics[width=.4\linewidth]{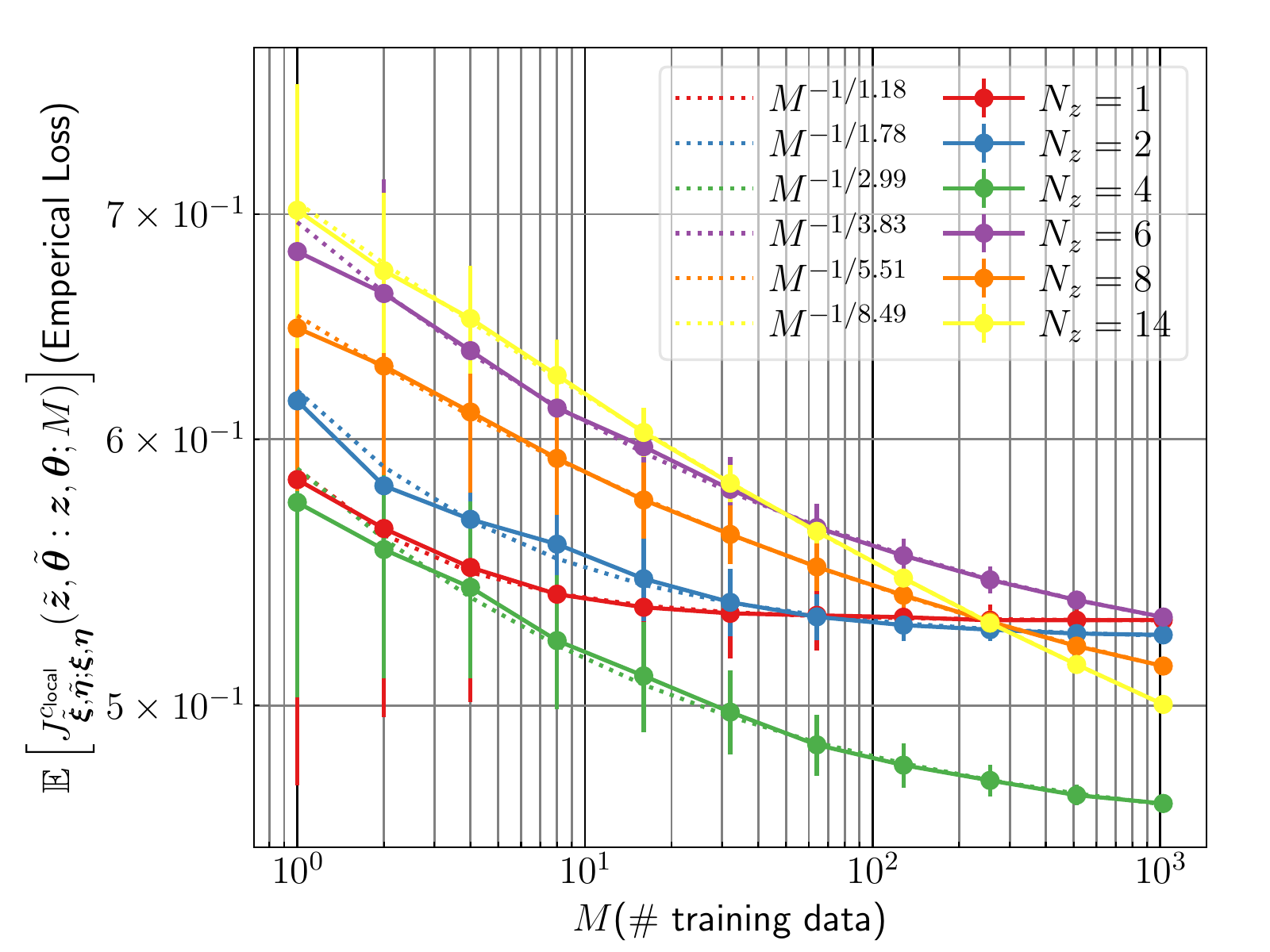}}
	\subfigure[$n=8$]{\includegraphics[width=.4\linewidth]{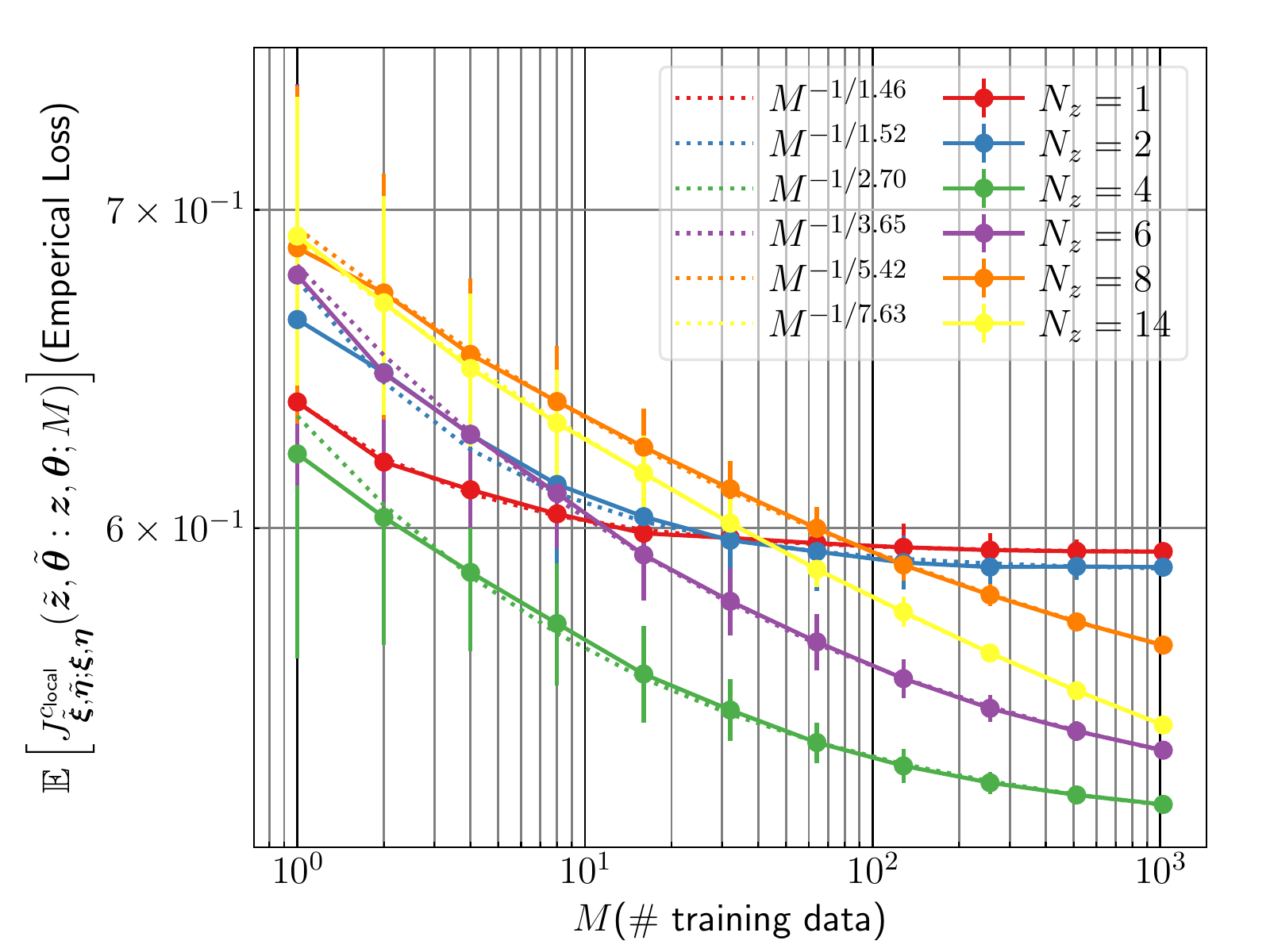}}
	\caption{
		Simulation results of relations ship between the approximation error of the proposed loss Eq.~\eqref{eq:numericalExperiment_differentdistribution} and number of training data $M$.
		Ths points represent the result of the proposed loss, and the dotted lines represent the fitting result with Eq.~\eqref{eq:fitting}.
		The values of parameter $B$ obtained by the fitting is shown in the legend of the dotted lines.
	}\label{fig:numericalExperiment_dif}
\end{figure}
\section{Proof of Proposition~\ref{prop:err_shotdependence}}
\label{sec:proof_shotdependence}
Here, we provide the proof of Proposition~\ref{prop:err_shotdependence} in Sec.~\ref{sec:dependenceShotNumber}.
To prove the proposition, we exploit the following lemma.
\begin{lemma}
	Suppose that $X_i\ (i=1,2,\ldots,N_s)$ are i.i.d. random variables with positive values and expected value $p=\mathbb{E}[X_i]$.
	Then, the following inequalities hold for the expected value of root mean $\mathbb{E}[\sqrt{\bar{X}}]$ and the variance $\mathbb{V}[\sqrt{\bar{X}}]$ of $\bar{X} = \frac{1}{N_s}\sum_{i=1}^{N_s} X_i$:
	\begin{align}
		 & \sqrt{p} - \frac{1-p}{2 N_s \sqrt{p}} \le \mathbb{E}[\sqrt{\bar{X}}]\le \sqrt{p}, \label{eq:inequality_mean} \\
		 & \mathbb{V}[\sqrt{\bar{X}}] \le \frac{1-p}{N_s} + \frac{(1-p)^2}{4N_s^2 p}. \label{eq:inequality_variance}
	\end{align}
	\label{lemm:sqrt_vars}
\end{lemma}
\begin{proof}
	The inequality on the right side of Eq.~\eqref{eq:inequality_mean} is obvious from Jensen's inequality.
	The inequality on the left side of Eq.~\eqref{eq:inequality_mean} can be ontained by substituting $r=\frac{\bar{X}}{\mathbb{E}(\bar{X})}$ into the the inequality $\sqrt{r}\ge \frac{-r^2+3 r}{2}$, which holds for any non-negative number $r$.
	The second inequality for the variance Eq.~\eqref{eq:inequality_variance} be obtained by substituting the first inequality Eq.~\eqref{eq:inequality_mean} into the definition of the variance $\mathbb{V}[\sqrt{\bar{X}}]=\mathbb{E}\left[\left(\sqrt{\bar{X}-\mathbb{E}(\sqrt{\bar{X}})}\right)^2\right]$.
\end{proof}
Here, the random variable of this lemma $X_i$ would be identified with the random variable $X_i=\sum_{k=1}^n X_{i,j,k}^{(s)}$, where $X_{i,j,k}^{(s)}$ is defined just above Eq.~\eqref{eq:localcost_estimator}.

Next, we consider the following lemma, which indicates that if the difference between two ground costs $c_1,c_2$ is sufficiently small, the difference between corresponding two optimal transport losses $\mathcal{L}_{c^1},\mathcal{L}_{c^2}$ is also sufficiently small.
\begin{lemma}
	Denote the optimal transport plans of Eq.~\eqref{eq:optimaltransportloss} between the empirical distributions $\hat{\alpha}_M(x)=\frac{1}{M}\sum_{i=1}^M\delta(\bm{x}-\bm{x}_i)$ and $\hat{\beta}_M(y)=\frac{1}{M}\sum_{j=1}^M\delta(\bm{y}-\bm{y}_j)$ for two different ground costs $c_1$ and $c_2$ as $\pi^1_{i,j}=\pi^1(\bm{x}_i,\bm{y}_j)$ and $\pi^2_{i,j}=\pi^2(\bm{x}_i,\bm{y}_j)$, respectively.
	Also, denote the non-zero components of the optimal transport plans  $\pi^1_{i,j}$ and $\pi^2_{i,j}$ as $A^1=\{(i,j)| \pi^1_{i,j}>0\}$ and $A^2=\{(i,j)| \pi^2_{i,j}>0\}$, respectively.
	Suppose that the difference between the ground costs $c_1$ and $c_2$ for any $(i,j) \in A^1 \cup A^2$ is upper bounded with a constant $t$ as follows:
	\begin{equation}
		|c^1(\bm{x}_i,\bm{y}_j)-c^2(\bm{x}_i,\bm{y}_j)| < t.
	\end{equation}
	Then, the difference between the corresponding optimal transport losses $\mathcal{L}_{c^1}$ and $\mathcal{L}_{c^2}$ follows 
	\begin{equation}
		|\mathcal{L}_{c^1}(\hat{\alpha}_M,\hat{\beta}_M) -	\mathcal{L}_{c^2}(\hat{\alpha}_M,\hat{\beta}_M)|  < t.
	\end{equation}
	\label{lemma:costdiff_vs_Wassersteindiff}
\end{lemma}
\begin{proof}
	Here, we only show the proof of $\mathcal{L}_{c^1}(\hat{\alpha}_M,\hat{\beta}_M) -	\mathcal{L}_{c^2}(\hat{\alpha}_M,\hat{\beta}_M) < t$. The other inequality $\mathcal{L}_{c^2}(\hat{\alpha}_M,\hat{\beta}_M) -	\mathcal{L}_{c^1}(\hat{\alpha}_M,\hat{\beta}_M) < t$ can be proved in the same way.
	\begin{align*}
		\mathcal{L}_{c^1}(\hat{\alpha}_M,\hat{\beta}_M) -\mathcal{L}_{c^2}(\hat{\alpha}_M,\hat{\beta}_M) & = \min_{\pi} \int c^1(\bm{x},\bm{y})d\pi(\bm{x},\bm{y})-\min_{\pi} \int c^2(\bm{x},\bm{y})d\pi(\bm{x},\bm{y})                   \\
		                                                                                                 & = \min_{\pi} \int c^1(\bm{x},\bm{y})d\pi(\bm{x},\bm{y})-\sum_{i,j=1}^M c^2(\bm{x}_i,\bm{y}_j)\pi^2(\bm{x}_i,\bm{y}_j)           \\
		                                                                                                 & \le \sum_{i,j=1}^M c^1(\bm{x}_i,\bm{y}_j)\pi^2(\bm{x}_i,\bm{y}_j)-\sum_{i,j=1}^M c^2(\bm{x}_i,\bm{y}_j)\pi^2(\bm{x}_i,\bm{y}_j) \\
		                                                                                                 & \le \sum_{i,j=1}^M t\pi^2(\bm{x}_i,\bm{y}_j) = t
	\end{align*}
\end{proof}
Lemma~\ref{lemma:costdiff_vs_Wassersteindiff} immediately leads the following lemma.
\begin{lemma}
	\label{lemma:upperbound_otloss}
Given a set of positive constant $t$ and $\delta$, which satisfies 
	\begin{equation}
		P(|c^1(\bm{x}_i,\bm{y}_j) - c^2(\bm{x}_i,\bm{y}_j)| \ge t) \le \frac{\delta }{2M}
		\label{eq:inequality_prob}
	\end{equation}
	for any component $i,j\in \{1,2,\ldots,M\}$, where the notations are same as Lemma~\ref{lemma:costdiff_vs_Wassersteindiff}.
	Then, the following inequality holds for the corresponding optimal transport losses:
	\begin{equation}
		P(|\mathcal{L}_{c^1}(\hat{\alpha}_M,\hat{\beta}_M) - \mathcal{L}_{c^2}(\hat{\alpha}_M,\hat{\beta}_M)| \ge t) \le \delta.
	\end{equation}
\end{lemma}

\begin{proof}
	For simplicity, we prove only for the the solution of the optimal transport plan that $M$ components take the value $1/M$ and the other components take the value $0$, but proofs for the other solutions can be performed in the same manner.
	For such a solution, the number of non-zero components is at most $|A^1 \cup A^2| \le 2 M$, and using the assumption Eq.~\eqref{eq:inequality_prob}, the probability that those non-zero components have an error within $t$ is upper bounded as 
	\begin{equation}
		P\left(\cap_{(i,j)\in A^1 \cup A^2}|c^1(\bm{x}_i,\bm{y}_j) - c^2(\bm{x}_i,\bm{y}_j)|< t\right) > \left( 1-\frac{\delta}{2M} \right)^{2M} \ge 1-\delta.
	\end{equation}
	Then, Lemma~\ref{lemma:upperbound_otloss} can be straightforwardly proven by taking into account that Lemma~\ref{lemma:costdiff_vs_Wassersteindiff} leads 
	\begin{equation}
		P\left(\cap_{(i,j)\in A^1 \cup A^2}|c^1(\bm{x}_i,\bm{y}_j) - c^2(\bm{x}_i,\bm{y}_j)|< t\right)   \le P(|\mathcal{L}_{c^1}(\hat{\alpha}_M,\hat{\beta}_M) - \mathcal{L}_{c^2}(\hat{\alpha}_M,\hat{\beta}_M)| \le t).
	\end{equation}
\end{proof}

Now, we are ready to prove Proposition~\ref{prop:err_shotdependence}.
\begin{proof}
	Under the same setting as Lemma~\ref{lemm:sqrt_vars}, Chebyshev inequality for the random variable $\sqrt{\bar{X}}$ leads, with some positive constant $k$,
	\begin{equation}
		\begin{split}
			\frac{1}{k^2}
			&\ge P \left( \left|\sqrt{\bar{X}} -\mathbb{E}\left(\sqrt{\bar{X}}\right) \right| \ge k \sqrt{\mathbb{V}\left[\sqrt{\bar{X}}\right]}\right) \\
			&\ge P \left( \left|\sqrt{\bar{X}} -\sqrt{p}\right| - \left| \sqrt{p}-\mathbb{E}\left(\sqrt{\bar{X}}\right) \right| \ge k \sqrt{\mathbb{V}\left[\sqrt{\bar{X}}\right]}\right) \\
			&\ge P \left( \left|\sqrt{\bar{X}} -\sqrt{p}\right|  \ge k \sqrt{ \frac{1-p}{N_s} + \frac{(1-p)^2}{4N_s^2 p}}+\frac{1-p}{2 N_s \sqrt{p}}\right)
			\label{eq:inequality_chebyshev}
		\end{split}
	\end{equation}
	By setting $\sqrt{\bar{X}} = \mathcal{L}_{\tilde{c}}$ and $\sqrt{p}=\mathcal{L}_c$, and comparing Eqs.\eqref{eq:inequality_prob} and \eqref{eq:inequality_chebyshev}, we obtain Proposition~\ref{prop:err_shotdependence} by setting 
	\begin{equation*}
		\begin{split}
			t &= \sqrt{\frac{2n}{\delta}}\sqrt{ \frac{1-p}{N_s} + \frac{(1-p)^2}{4N_s^2 p}}+\frac{1-p}{2 N_s \sqrt{p}}.
		\end{split}
	\end{equation*}
\end{proof}

\section{Intuitive explanation on the shape of Fig.\ref{fig:numericalExperiment_shotDependence}}
\label{sec:explain_shotDependence}
Here, we give the rough explanation of the dependence of the mean approximation error on the number of training data, presented in Fig.\ref{fig:numericalExperiment_dif} of Sec.\ref{sec:dependenceShotNumber}.
Throughout this section, we assume that the number of shots $N_s$ is sufficiently large to hold the asymptotic theory.

We first focus on the case of small number of training data $M$, where the approximation error behaves like $M^{-1/2}$.
In this case, the training data $\{\bm{x}_i\}_{i=1}^M$ would be well separated from each other and the optimal transport plan $\{\pi_{i,j}\}_{i,j=1}^M$ is not expected to be affected by the number of shots $N_s$.
In most cases, the number of non-zero elements of optimal transport plan $A=\{(i,j)| \pi_{i,j}>0\}$ is $M$ and the value of those are $1/M$.

The estimated value of the ground cost $c_{\textrm{local},i,j}$ of Eq.~\eqref{eq:localcost_estimator} with $N_s$ shots is given by $\tilde{c}_{\text{local},i,j}^{({N_s})} = \sqrt{ \frac{1}{n}\sum_{k=1}^n \frac{1}{N_s}\sum_{s=1}^{N_s} X_{i,j,k}^{(s)} }$, where $X_{i,j,k}^{(s)}$ are random variables following the Bernoulli distribution with probability $1-p^{(k)}_{i,j}$.
Due to the central limit theorem, the inside of the root $Y_{i,j}^{(s)}= \frac{1}{n}\sum_{k=1}^n X_{i,j,k}^{(s)}$ asymptotically converge to a normal distribution $\sqrt{N_s}(\sum_{s=1}^{N_s}Y_{i,j}^{(s)}/N_s-\mu_{i,j})\sim \mathcal{N}( 0,\sigma_{i,j}^2 )$ with mean $\mu_{i,j}=\sum_{k=1}^n(1-p^{(k)}_{i,j})$ and variance $\sigma_{i,j}^2=\sum_{k=1}^n(1-p^{(k)}_{i,j})p^{(k)}_{i,j}$.
Thus, delta method tells us that the approximation error of the ground cost follows $\sqrt{N_s}\left(
	\tilde{c}^{(N_s)}_{\textrm{local},i,j}-c_{\textrm{local},i,j}\right)\sim \mathcal{N}\left(0,\frac{\sigma^2_{i,j}}{4\mu_{i,j}}\right)$.

	Assume that the means and variances are almost the same for all the components, i.e., $\mu_{i,j}\approx \mu$, $\sigma_{i,j}\approx\sigma$ $\forall i,j$, the approximation error of the optimal transport loss due to the number of shots can be written as 
\begin{equation}
	\begin{split}
		\mathcal{L}_{\tilde{c}_{\rm{local}}^{(N_s)}} -\mathcal{L}_{c_{\rm{local}}}
		&= \frac{1}{M}\sum_{(i,j)\in A} \left(\tilde{c}^{(N_s)}_{\textrm{local},i,j}-c_{\textrm{local},i,j}\right) \\
		& \sim \mathcal{N}\left(0,\sum_{(i,j)\in A}\frac{\sigma^2_{i,j}}{ 4N_s\mu_{i,j}M^2}\right)\\
		& \approx \mathcal{N}\left(0,\frac{\sigma^2}{ 4N_s\mu M}\right).
	\end{split}
\end{equation}
Thus the approximation error with small number of training data behaves like $(N_s M)^{-1/2}$ under the condition shown here.

On the other hand, behavior of the approximation error in the range of large number of training data $M$ is explained by the theory of the extreme value distribution\cite{de2006extreme}.
In the case of large number of training data, it would be expected that there are many ground costs $c_{\textrm{local},i,j}$ with almost the same value.
As an extreme case, consider the case where all ground costs have a common constant value, $c_{\textrm{local},i,j}=c,\ \forall i,j$.
Again assume that the number of shots $N_s$ is sufficiently large, then the ground cost $\tilde{c}^{(N_s)}_{\textrm{local},i,j}$ follows a normal distribution, which we denote as $\mathcal{N}\left(c,\frac{\sigma^2}{ N_s}\right)$.
Then, using i.i.d random numbers $\{X_{i,j}\}_{i,j=1}^M$ which follow a normal distribution $\mathcal{N}\left(0,\frac{\sigma^2}{N_s}\right)$, the approximation error can be written as 
\begin{equation}
	\begin{split}
		\mathcal{L}_{\tilde{c}_{\rm{local}}^{(N_s)}} -\mathcal{L}_{c_{\rm{local}}}
		&\approx\min_{\{\pi_{i,j}\}_{i,j=1}^{M} } \sum_{i,j=1}^{{M}}X_{i,j}\pi_{i,j}, \\
		\mathrm{subject\ to} \quad    & \sum_{i=1}^{M}\pi_{i,j} = \frac{1}{{M}},\sum_{j=1}^{M_g}\pi_{i,j} = \frac{1}{{M}},\pi_{i,j} \ge 0.
	\end{split}
\end{equation}
Now we approximate this minimization by greedy algorithm, i.e., consider first obtaining the minimum value $X_{i_1,j_1}$ from the $M^2$ components, and then the second minimum value $X_{i_2,j_2}$ from the rest $(M-1)^2$ components other than $i$-th row and $j$-th column, and so on.
Denoting the cumulative distribution function of a random variable $X_{i,j}$ as $F(x)$, the distribution of the minimum value of the $k$ data can be written as 
\begin{equation}
	\begin{split}
		G(x,k) &= 1-(1-F(x))^k, \\
		p(x,k) &= \frac{dG(x,k)}{dx} = k\frac{d F(x)}{dx}(1-F(x))^{k-1}.
	\end{split}
\end{equation}
Then the probability density at which $x_1,x_2,x_3,\ldots,x_M$ are obtained from the greedy algorithm is given as 
\begin{equation}
	\begin{split}
		&p(x_1,x_2,\ldots,x_M) \\
		&=
		p(x_1,M^2)
		\frac{p(x_2,(M-1)^2)\theta(x_2-x_1)}{1-G(x_1,(M-1)^2)}
		\frac{p(x_3,(M-2)^2)\theta(x_3-x_2)}{1-G(x_2,(M-2)^2)}
		\times\cdots \times
		\frac{p(x_M,1^2)\theta(x_M-x_{M-1})}{1-G(x_{M-1},1^2)}\\
		&=\prod_{k=1}^M \frac{k^2}{2k-1}p(x_k,2k-1)\theta(x_k-x_{k-1}),
	\end{split}
\end{equation}
where $\theta(x)$ denotes a step function, and we set $x_0=-\infty$ in the last expression.
Finally, we approximate this expression by the mode.
Then, from the theory of the extreme value distribution, the mode of $p(x,k)$ can be written as $x_{mode}\approx -\sigma\sqrt{2 \ln M/N_s}$ and we reach
\begin{equation}
	\begin{split}
		\mathcal{L}_{\tilde{c}_{\rm{local}}^{(N_s)}} -\mathcal{L}_{c_{\rm{local}}}
		&\approx
		\frac{\sigma}{\sqrt{N_s}}\left(
		\frac{1}{M}\sum_{k=1}^M
		\sqrt{2 \ln (2k-1)}
		\right) \\
		&\approx
		\frac{\sigma}{\sqrt{N_s}}
		\sqrt{2 \ln (2M-1)}.
	\end{split}
\end{equation}
Thus we can roughly understand that the approximation error with large number of training data behaves like $N_s^{-1/2}\sqrt{ \ln (M)}$.

\section{Simulation results of Sec.\ref{sec:barrenplateau}}
\label{sec:bpsresult}
Here, we give the numerical simulation results on the variance in the gradient of the proposed cost function.
As discussed in Sec.\ref{sec:barrenplateau}, we introduce the loss function with the ground cost calculated by the local cost defined in Eq.~\eqref{eq:localcost}, and it avoids the vanishing of the gradient.
The results based on different data numbers are shown in Fig.\ref{fig:cost_barren_plateau_app}.
The parameters of the simulation are given in Table~\ref{tab:ExperimentalParamters_barrenplateau}.
Fig.\ref{fig:grad_n}(c) is created based on the result of Fig.\ref{fig:cost_barren_plateau_app}.

\begin{figure}[hbtp]
	\centering
	\subfigure[$M=1$]{\includegraphics[width=.4\linewidth]{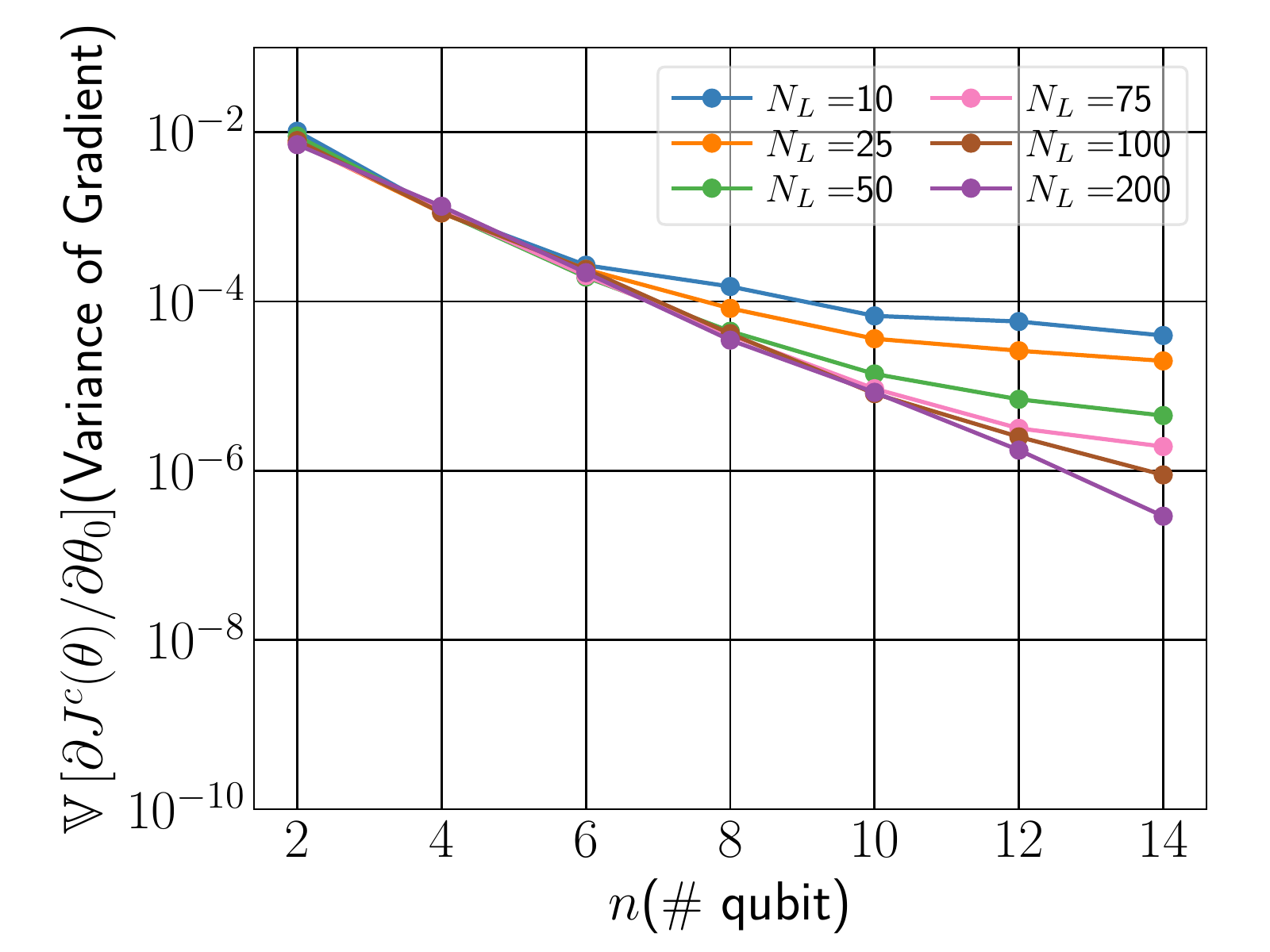}}
	\subfigure[$M=2$]{\includegraphics[width=.4\linewidth]{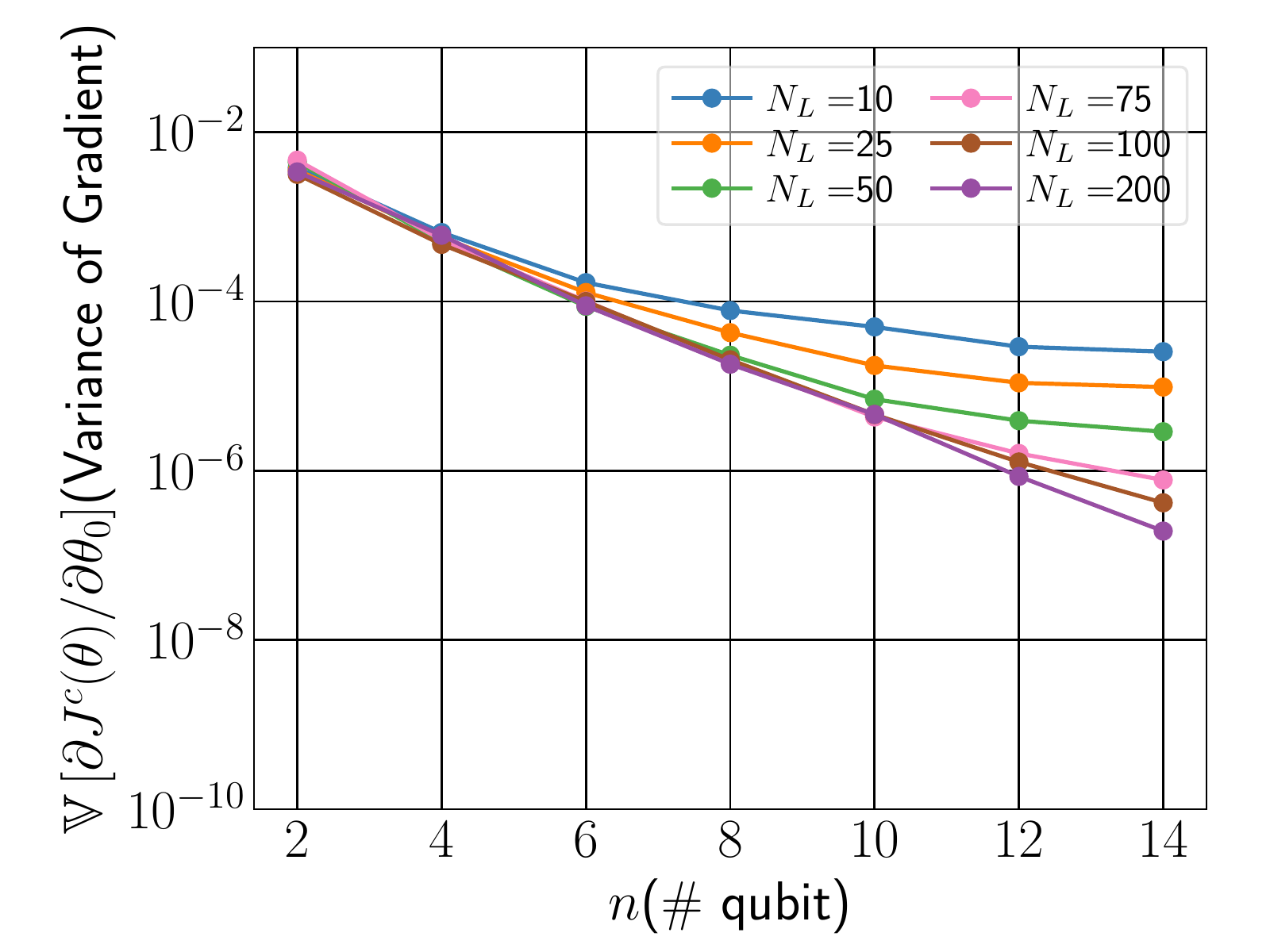}}
	\subfigure[$M=8$]{\includegraphics[width=.4\linewidth]{fig/VarGrad_vs_n_L_HEA_8.pdf}}
	\subfigure[$M=16$]{\includegraphics[width=.4\linewidth]{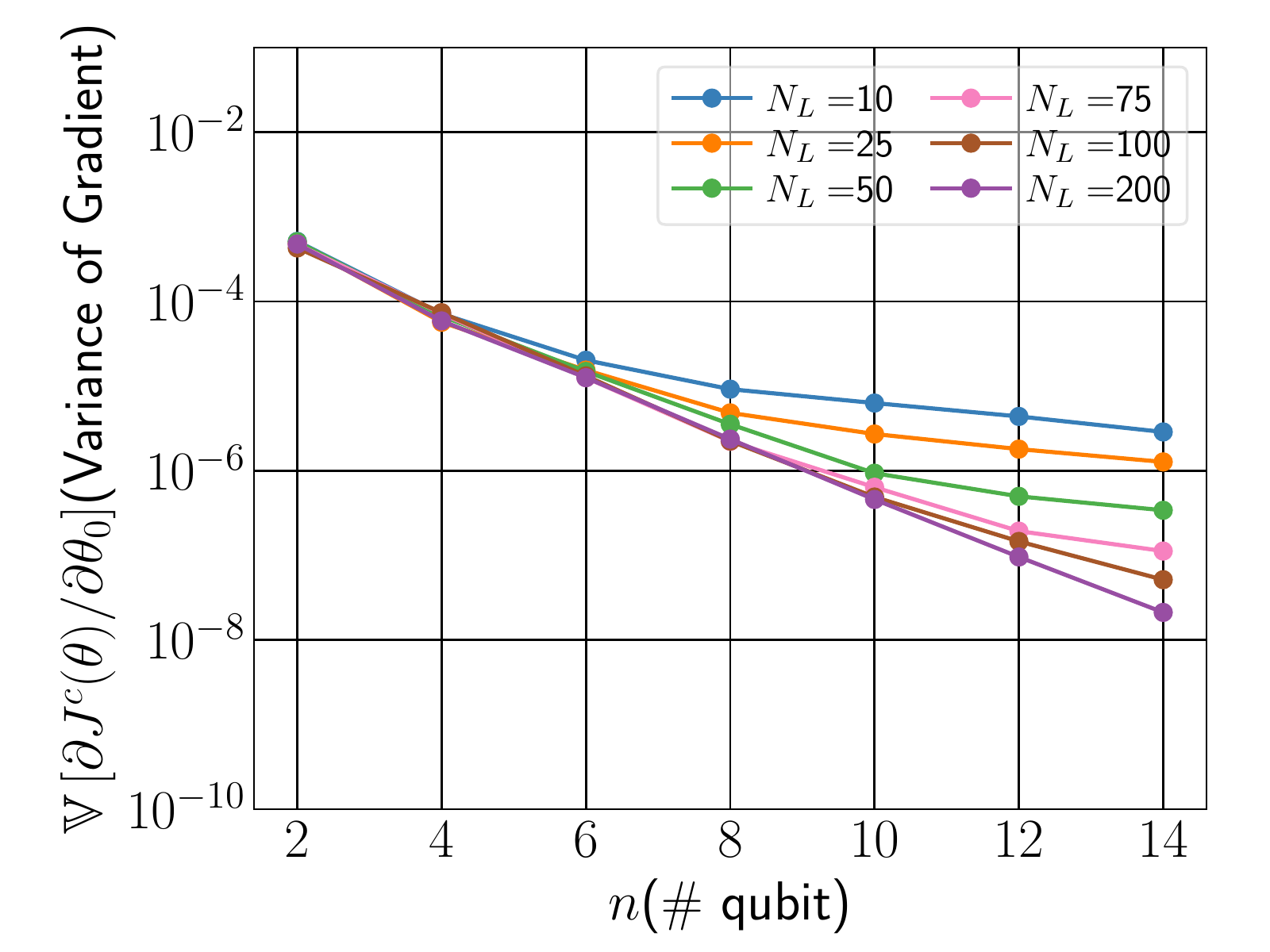}}
	\caption{
		Simulation results of relationship between the number of qubit and the variance of the gradient of the proposed loss, which is based on local cost.
	}\label{fig:cost_barren_plateau_app}
\end{figure}

\end{document}